 \documentclass[referee]{ejm}

\let\realverbatim=\verbatim
\let\realendverbatim=\endverbatim
\renewcommand\verbatim{\par\addvspace{6pt plus 2pt minus 1pt}\realverbatim}
\renewcommand\endverbatim{\realendverbatim\addvspace{6pt plus 2pt minus 1pt}}

\ifprodtf \else
  \checkfont{eurm10}
  \iffontfound
    \IfFileExists{upmath.sty}
      {\typeout{^^JFound AMS Euler Roman fonts on the system,
                   using the 'upmath' package.^^J}%
       \usepackage{upmath}}
      {\typeout{^^JFound AMS Euler Roman fonts on the system, but you
                   don't seem to have the}%
       \typeout{'upmath' package installed. EJM.cls can take advantage
                 of these fonts,^^Jif you use 'upmath' package.^^J}%
       \providecommand\mathup[1]{##1}%
       \providecommand\mathbup[1]{##1}%
      }
  \else
    \providecommand\mathup[1]{#1}%
    \providecommand\mathbup[1]{#1}%
  \fi
\fi

\ifprodtf \else
  \checkfont{msam10}
  \iffontfound
    \IfFileExists{amssymb.sty}
      {\typeout{^^JFound AMS Symbol fonts on the system, using the
                'amssymb' package.^^J}%
       \usepackage{amssymb}%
         \let\leq=\leqslant
         \let\geq=\geqslant
      }{}
  \fi
\fi

\ifprodtf \else
  \IfFileExists{amsbsy.sty}
    {\typeout{^^JFound the 'amsbsy' package on the system, using it.^^J}%
     \usepackage{amsbsy}}
    {\providecommand\boldsymbol[1]{\mbox{\boldmath $##1$}}}
\fi

\newsavebox{\astrutbox}
\sbox{\astrutbox}{\rule[-5pt]{0pt}{20pt}}

\newdefinition{definition}[theorem]{Definition}
\newtheorem{lemma}{Lemma}

 \usepackage[]{graphicx}
\usepackage[]{amsmath,amssymb,epsfig,cite,subfigure,comment,array, algorithm,algorithmic, url,placeins,soul}

\usepackage{lineno}
\usepackage{setspace}
\usepackage{multicol}
\usepackage{multirow}
\usepackage{color}
\usepackage{colortbl}
\usepackage{hyperref}
\usepackage{verbatim}

\hypersetup{
    colorlinks=true,
    citecolor=red,
    linkcolor=blue,
    urlcolor=blue
  }

\newcommand{\mb}[1]{\mbox{\boldmath$#1$}}

\definecolor{dkgreen}{rgb}{0,0.6,0}
\definecolor{gray}{rgb}{0.5,0.5,0.5}
\definecolor{mauve}{rgb}{0.58,0,0.82}

\title[European Journal of Applied Mathematics]{Detection of Core--Periphery Structure in Networks Using Spectral Methods and Geodesic Paths}
\author[Mihai~Cucuringu et al.]{%
Mihai~Cucuringu$\,^1$, \ns
Puck~Rombach$\,^2$  \ns
Sang~Hoon~Lee$\,^3$ \ns
\and
Mason~A.~Porter$\,^4$
}

\affiliation{%
  $^1\,$Department of Mathematics, UCLA, Los Angeles, CA (mihai@math.ucla.edu). This work was initiated while the author was affiliated with the Program in Applied and Computational Mathematics (PACM) at Princeton University, Princeton, NJ.\\
    email\textup{\nocorr: \texttt{mihai@math.ucla.edu}}\\
  $^2\,$Oxford Centre for Industrial and Applied Mathematics, Mathematical Institute, University of Oxford, Oxford, UK; and Department of Mathematics, UCLA, Los Angeles, CA.\\
  $^3\,$School of Physics, Korea Institute for Advanced Study, Seoul,
Korea; and Integrated Energy Center for Fostering Global Creative Researcher (BK 21 plus) and Department of Energy Science, Sungkyunkwan University, Suwon, Korea; and Oxford Centre for Industrial and Applied Mathematics, Mathematical Institute, University of Oxford, Oxford, UK
$^4\,$Oxford Centre for Industrial and Applied Mathematics, Mathematical Institute, University of Oxford, Oxford, UK; and CABDyN Complexity Centre, University of Oxford, Oxford, UK.
}

\date{\today}
\pubyear{2016}
\volume{000}
\pagerange{\pageref{firstpage}--\pageref{lastpage}}

\begin{document}
\label{firstpage}
\maketitle

\begin{abstract}

We introduce several novel and computationally efficient methods for detecting ``core--periphery structure'' in networks. Core--periphery structure is a type of mesoscale structure that includes densely-connected core vertices and sparsely-connected peripheral vertices. Core vertices tend to be well-connected both among themselves and to peripheral vertices, which tend not to be well-connected to other vertices. Our first method, which is based on transportation in networks, aggregates information from many geodesic paths in a network and yields a score for each vertex that reflects the likelihood that a vertex is a core vertex. Our second method is based on a low-rank approximation of a network's adjacency matrix, which can often be expressed as a tensor-product matrix. Our third approach uses the bottom eigenvector of the random-walk Laplacian to infer a coreness score and a classification into core and peripheral vertices. We also design an objective function to (1) help classify vertices into core or peripheral vertices and (2) provide a goodness-of-fit criterion for classifications into core versus peripheral vertices. To examine the performance of our methods, we apply our algorithms to both synthetically-generated networks and a variety of networks constructed from real-world data sets. 

\end{abstract}

\begin{keywords}
05C82, 68R10, 91D30, 05C85, 15A18
\end{keywords}

\tableofcontents



\section{Introduction} \label{sec:intro}

Network science has grown explosively during the past two decades \cite{booknewman}, and myriad new journal articles on network science appear every year. One focal area in the networks literature is the development and analysis of algorithms for detecting local, mesoscale, and global structures in various types of networks \cite{fortunato,masonams}.  Mesoscale features are particularly interesting, as they arise neither at the local scale of vertices (i.e., nodes) and edges nor at the global scale of summary statistics. In the present paper, we contribute to research on mesoscale network structures by developing and analyzing new (and computationally-efficient) algorithms for detecting a feature known as \textit{core--periphery structure}, which consists of densely-connected \textit{core} vertices and sparsely-connected \textit{peripheral} vertices.


The importance of investigating mesoscale network structures is acknowledged widely \cite{fortunato,masonams}, but almost all of the research on this topic concerns a specific type of feature known as \textit{community structure}. In studying community structure, one typically employs some algorithm to detect sets of vertices called \textit{communities} that consist of vertices that are densely connected to each other, such that the connection density between vertices from different communities is comparatively sparse \cite{masonams,fortunato,mason6,Girv04}. A diverse array of methods exist to detect community structure, and they have been applied to numerous areas, such as committee networks in political science \cite{mason11}, friendship networks \cite{mason13,mason14}, protein interaction networks \cite{Chen06,mason15}, functional brain networks \cite{masonbrain}, and mobile phone networks \cite{mason16}. Popular methods include the optmization of a quality function called ``modularity''~\cite{Newm2003,Girv04,New06}, spectral partitioning~\cite{Spiel1996,Guat1995}, dynamical approaches based on random walkers or other dynamical systems~\cite{Pons05,arenas2006,Rosvall2008fi,Picc11,jeub2014}, and more.  Most community-detection methods require a vertex to belong to a distinct community, but several methods also allow the detection of overlapping communities (see, e.g., \cite{Pall05,mason9,mason10,jeub2014}). 
 
Core--periphery structure is a mesoscale feature that is rather different from community structure. The main difference is that core vertices are well-connected to peripheral vertices, whereas the standard perspective on community structure views communities as nearly decomposable modules (which leads to trying to find the best block-diagonal fit to a network's adjacency matrix) \cite{puckmason,XZhang2014}. Core--periphery structure and community structure are thus represented by different types of block models \cite{jeub2014, peixoto2014}. The quantitative investigation of core--periphery structure has a reasonably long history \cite{cp-review}, and qualitative notions of core--periphery structure have long been considered in fields such as international relations \cite{wallerstein1974,steiber1979,chase1989,smithwhite}, sociology \cite{laumann1976,doreian1985}, and economics \cite{krugman1996} (and have been examined more recently in applications such as neuroscience \cite{masonbrain2}, transportation \cite{corePerApp}, and faculty movements in academia \cite{clauset2015}), but the study of core--periphery structure remains poorly developed --- especially in comparison to the study of community structure \cite{masonams,fortunato}. Most investigations of core--periphery structure tend to use the perspective that a network's adjacency matrix has an intrinsic block structure (which is different from the block structure from community structure)~\cite{BorgattiCore,Comr62,puckmason}. Very recently, for example, Ref.~\cite{XZhang2014} identified core--periphery structure by fitting a stochastic block model (SBM) to empirical network data using a maximum likelihood method, and the SBM approach in Ref.~\cite{peixoto2014} can also be used to study core--periphery structure. Importantly, it is possible to think of core--periphery structure using a wealth of different perspectives, such as overlapping communities~\cite{Jure13}, $k$-cores~\cite{Holme05}, network capacity~\cite{Silva08}, and random walks~\cite{Dell13}. It is also interesting to examine growth mechanisms to generate networks with core--periphery structure \cite{verma2016}. The notion of ``nestedness'' \cite{eco-nested} from ecology is also related to core--periphery structure \cite{shl-nested}. The main contribution of the present paper is the development of novel algorithms for detecting core--periphery structure. Our aim is to develop algorithms that are both computationally efficient and robust to high levels of noise in data, as such situations can lead to a blurry separation between core vertices and peripheral vertices.

The rest of this paper is organized as follows. In Section \ref{sec:corePerIntro}, we give an introduction to the notion of core--periphery structure and briefly survey a few of the existing methods to detect such structure. In Section \ref{sec:PathCore}, we introduce the {\sc Path-Core} method, which is based on computing shortest paths between vertices of a network, for detecting core--periphery structure. In Section \ref{sec:objFuncSync}, we introduce an objective function for detecting core--periphery structure that leverages our proposed algorithms and helps in the classification of vertices into a core set and periphery set. In Section \ref{sec:rank2}, we propose the spectral method {\sc LowRank-Core}, which detects core--periphery structure by considering the adjacency matrix of a network as a low-rank perturbation matrix. In Section \ref{sec:Laplacian}, we investigate two Laplacian-based methods ({\sc Lap-Core} and {\sc LapSgn-Core}) for computing core--periphery structure in a network, and we discuss related work in community detection that uses a similar approach. In Section \ref{sec:numSims}, we compare the results of applying the above algorithms using several synthetically-generated networks and real-world networks. Finally, we summarize and discuss our results in Section \ref{sec:future}, and we also discuss several open problems and potential applications. In Appendix 1, we detail the steps of our proposed {\sc Path-Core} algorithm for computing the Path-Core scores, and we include an analysis of its computational complexity. In Appendix 2, we discuss the spectrum of the random-walk Laplacian of a graph (and of the random-walk Laplacian of its complement). In Appendix 3, we detail an experiment with artificially planted high-degree peripheral vertices that illustrates the sensitivity of a degree-based method (which we call {\sc Degree-Core} and which uses vertex degree as a proxy to measure coreness) to such outlier vertices. Finally, in Appendix 4, we calculate Spearman and Pearson correlation coefficients between the coreness scores that we obtain from the different methods applied to several real-world networks.


\section{{\sc Core-Score}: Density-Based Core--Periphery Structure in Networks}\label{sec:corePerIntro}

The best-known quantitative approach to studying core--periphery structure was introduced by Borgatti and Everett \cite{BorgattiCore}, who developed algorithms for detecting discrete and continuous versions of core--periphery structure in weighted, undirected networks. (For the rest of the present paper, note that we will use the terms ``network'' and ``graph'' interchangeably.) Their discrete methods start by comparing a network to an ideal block matrix in which the core is fully connected, the periphery has no internal edges, and the periphery is well-connected to the core. 

Borgatti and Everett's main algorithm for finding a discrete core--periphery structure assigns each vertex either to a single ``core'' set of vertices or to a single ``periphery'' set of vertices. One seeks a vector ${\bf C}$ of length $n$ whose entries are either $1$ or $0$, depending on whether or not the associated vertex has been assigned to the core ($1$) or periphery ($0$). We let $H_{ij}=1$ if $C_i=1$ (i.e., vertex $i$ is assigned to the core) or $C_j=1$ (i.e., vertex $j$ is assigned to the core), and we otherwise let $H_{ij}=0$ (because neither $i$ nor $j$ are assigned to the core). We define $\rho_C = \sum_{i,j} A_{ij}H_{ij}$, where $A$ (with elements $A_{ij}$) is the adjacency matrix of the (possibly weighted) network $G$. Borgatti and Everett's algorithm searches for a value of $\rho_C$ that is high compared to the expected value of $\rho$ if ${\bf C}$ is shuffled such that the number of $0$ and $1$ entries is preserved but their order is randomized. The final output of the method is the vector ${\bf C}$ that gives the highest $z$-score for $\rho_C$.  In a variant algorithm for detecting discrete core--periphery structure, Borgatti and Everett still let $H_{ij}=1$ if both $C_i$ and $C_j$ are equal to $1$ and let $H_{ij} = 0$ if neither $i$ nor $j$ are assigned to the core, but they now let $H_{ij}=a \in [0,1]$ if either $C_i=1$ or $C_j=1$ (but not both). To detect a continuous core--periphery structure \cite{BorgattiCore}, Borgatti and Everett assigned a vertex $i$ a core value of $C_{i}$ and let $H_{ij}=C_i C_j$. A recent method that builds on the continuous notion of core--periphery structure from \cite{BorgattiCore} was proposed in \cite{puckmason}. It calculates a {\sc Core-Score} for weighted, undirected networks; and it has been applied (and compared to community structure) in the investigation of functional brain networks~\cite{masonbrain2}.

The method of core--periphery detection in the popular network-analysis software {\sc UCINet} \cite{ucinet} uses the so-called \emph{minimum residual} ({\sc MINRES}) method \cite{Comr62}, which is a technique for factor analysis. One uses factor analysis to describe observed correlations between variables in terms of a smaller number of unobserved variables called the ``factors'' \cite{darl73}. {\sc MINRES} aims to find a vector $C$ that minimizes
\begin{equation*}
	S(A,{\bf C})=\sum_{i=1}^n\sum_{\substack{j \neq i}} \left( A_{ij}-C_iC_j \right) ^2\,,
\end{equation*}	
where $C_i \geq 0$ for all vertices $i$. One ignores the diagonal elements of the network's adjacency matrix. Additionally, because ${\bf C}{\bf C}^T$ is symmetric, this method works best for undirected networks $G$. For directed networks, one can complement the results of {\sc MINRES} with a method based on a singular value decomposition (SVD)~\cite{boyd10}. In practice, {\sc UCINet} reports ${\bf C}/\sqrt{\sum_i C_i^2}$. 

In \cite{Jure13}, it was argued that core--periphery structure can arise as a consequence of community structure with overlapping communities. They presented a so-called ``community-affiliation graph model'' to capture dense overlaps between communities. In the approach in \cite{Jure13}, the likelihood that two vertices are adjacent to each other is proportional to the number of communities in which they have shared membership. Della Rossa et al. recently proposed a method for detecting a continuous core--periphery profile of a (weighted) network by studying the behavior of a random walker on a network \cite{Dell13}. Approaches based on random walks and other Markov processes have often been employed in the investigation of community structure~\cite{Pons05,Rosvall2008fi,Picc11,jeub2014}, and it seems reasonable to examine them for other mesocale structures as well.  Very recently, Ref.~\cite{XZhang2014} identified core--periphery structure by fitting a stochastic block model (SBM) to empirical network data using a maximum-likelihood method. The review article \cite{cp-review} discusses several other methods to detect core--periphery structure in networks.


\section{{\sc Path-Core}: Transport-Based Core--Periphery Detection via Shortest Paths in a Network} \label{sec:PathCore}

In transportation systems, some locations and routes are much more important than others. This motivates the idea of developing notions of core--periphery structure that are based on transportation. In this section, we restrict our attention to undirected and unweighted networks, although we have also examined transport-based core--periphery structure in empirical weighted and directed networks \cite{corePerApp}. In Section \ref{sec3.1}, we explain the intuition behind the proposed {\sc Path-Core} algorithm, and we examine its performance on several synthetic networks. We end this section by commenting on a randomized version of the {\sc Path-Core} algorithm that samples a subset of edges in a graph and computes shortest paths only between the endpoints of the associated vertices.


\subsection {{\sc Path-Core}} \label{sec3.1}

The first transport-based algorithm that we propose for detecting core--periphery structure is reminiscent of \emph{betweenness centrality} (BC) in networks \cite{Anth71,Freeman1977,New05rw}. One seeks to measure the extent to which a vertex controls information that flows through a network by counting the number of shortest paths (i.e., ``geodesic'' paths) on which the vertex lies between pairs of other vertices in the network. Geodesic vertex betweenness centrality is defined as
\begin{equation}
	B_C(i) = \sum_{j,k \in V(G) \backslash i} \frac{\sigma_{jk}(i)}{\sigma_{jk}}\,,
\label{def:BCcentrality}
\end{equation}
where $\sigma_{jk}$ is the number of different shortest paths (i.e., the ``path count'') from vertex $j$ to vertex $k$, and $\sigma_{jk}(i)$ is the number of such paths that include vertex $i$. Our approach also develops a scoring methodology for vertices that is based on computing shortest paths in a network. Such a score reflects the likelihood that a given vertex is part of a network's core. Instead of considering shortest paths between all pairs of vertices in a network, we consider shortest paths between pairs of vertices that share an edge \emph{when that edge is excluded from the network}. Specifically, we calculate
\begin{equation}
	\textrm{\sc Path-Core}(i)= \sum_{(j,k) \in E( V(G) \backslash i ) }\frac{\sigma_{jk}(i) \vert_{G \setminus (j,k)}}{\sigma_{jk}\vert_{G \setminus (j,k)}}\,,
\label{def:PathCoreCentrality}
\end{equation}
where $\sigma_{jk}(i) \vert_{G \setminus (j,k)}$ and $\sigma_{jk}\vert_{G \setminus (j,k)}$ are defined, respectively, as the path counts $\sigma_{jk}$ and $\sigma_{jk}(i)$ in the graph $G \setminus (j,k)$, and $E(X)$ denotes the edge set induced by the vertex set $X$. The network $G \setminus (j,k)$ denotes the subgraph of $G$ that one obtains by removing the edge $(j,k) \in E$. Alternatively, one can define the {\sc Path-Core} score of a vertex $i$ as the betweenness centrality of this vertex when considering paths only between pairs of adjacent vertices $j$ and $k$, but for which the edge $e_{jk}$ incident to the two vertices is discarded. Note that one can apply {\sc Path-Core} to weighted graphs by using generalizations of betweenness centrality to weighted graphs.

A related approach was used in \cite{valente2010bridging} to derive measures of ``bridging'' in networks based on the observation that edges that reduce distances in a network are important structural bridges. In the measure in \cite{valente2010bridging}, which employed a modification of closeness centrality, one systematically deletes edges and measures changes in the resulting mean path lengths. We also note the recent paper \cite{VE2016} about bridging centrality.


Let $G(V,E)$ be a graph with a vertex set $V$ of size $n$ (i.e., there are $|V| = n$ vertices) and an edge set $E$ of size $m$. The set of core vertices is $V_C$ (and its size is $n_c$), and the set of peripheral vertices is $V_P$ (and its size is $n_p$). We also sometimes use the notation ${\tt C} = |V_C|$ for the size of the core set. Suppose that a network (i.e., a graph) contains exactly one core set and exactly one periphery set, and that these sets are disjoint: $V_C \cup V_P=V$ and $V_C \cap V_P=\emptyset$. The goal of the {\sc Path-Core} algorithm is to compute a score for each vertex in the graph $G$ that reflects the likelihood that that vertex belongs to the core. In other words, high-scoring vertices have a high probability of being in the core, and low-scoring vertices have a high probability of being in the periphery. Throughout the paper, we use the term ``{\sc Path-Core} scores'' to indicate the scores that we associate with a network's vertices by using the {\sc Path-Core} algorithm.

\begin{table} 
\begin{minipage}[b]{0.99\linewidth}
\centering
\begin{tabular}{|c|c|}
 \multicolumn{1}{c}{}  & \multicolumn{1}{c}{} \\
\hline
$A_{\{CC\}}$ & $A_{\{CP\}}$ \\
 \hline
$A_{\{CP\}}$ &  $A_{\{PP\}}$ \\
\hline
\end{tabular}
\caption{Block model for the ensemble of graphs $G(p_{cc},p_{cp},p_{pp},n_c,n_p)$.  Note that either $p_{cc} \geq p_{cp} > p_{pp}$ or $p_{cc} > p_{cp} \geq p_{pp}$.
}
\label{tab:generalBlockModel}
\end{minipage}
\end{table}

We illustrate our methodology in the context of a generalized block model, such as the one in Table \ref{tab:generalBlockModel}, where the submatrices $A_{\{CC\}}$, $A_{\{CP\}}$, and $A_{\{PP\}}$ represent the interactions between a pair of core vertices, a core vertex and a peripheral vertex, and a pair of peripheral vertices, respectively. Suppose that $A_{\{CC\}}$ and  $A_{\{PP\}}$ are adjacency matrices that we construct using the $G(n,p)$ random graph model\footnote{In the random graph model $G(n,p)$ on $n$ vertices, an edge is present between each pair of vertices independently with probability $p$ \cite{gilbert1959random,erdos1959random}. It is common to abuse terminology and use the name ``Erd\H{o}s--R\'{e}nyi random graph'' for $G(n,p)$.} by considering $G(n_c,p_{cc})$ and $G(n_p,p_{pp})$, respectively, and that $A_{\{CP\}}$ is the adjacency matrix of a random bipartite graph $G(n_c,n_p,p_{cp})$ in which each edge that is incident to both a core and peripheral vertex is present with independent probability $p_{cp}$. As indicated by the above notation, $p_{cc}$ denotes the probability that there is an edge between any given pair of core vertices, and $p_{pp}$ denotes the probability that there is an edge between any given pair of peripheral vertices. In the context of the above block model, core--periphery structure arises naturally when either $p_{cc} \geq p_{cp} > p_{pp}$ or $p_{cc} > p_{cp} \geq p_{pp}$. The above family of random networks, which we denote by $G(p_{cc},p_{cp},p_{pp},n_c,n_p)$, was also considered in Ref.~\cite{puckmason}. It contains exactly one set of core vertices, and the remaining vertices are peripheral vertices. More complicated core--periphery structures can also occur \cite{puckmason}, such as a mix of (possibly hierarchical) community structures and core--periphery structures.

We now present the intuition behind the {\sc Path-Core} algorithm and the reason that the resulting {\sc Path-Core} score is a good indicator of the likelihood that a vertex is in the core or in the periphery. If $i$ and $j$ are adjacent core vertices, then it is likely that shortest paths between $i$ and $j$ consist entirely of other core vertices. If $i \in V_C$ and $j \in V_P$, then a shortest path between $i$ and $j$ should also mostly contain core vertices. Finally, even when  $i,j \in V_P$, it is still likely that a shortest path between $i$ and $j$ is composed of many core vertices and few peripheral vertices. Intuitively, once a shortest path reaches the set $V_C$, it is likely to stay within the core set $V_C$ until it returns to the periphery set $V_P$ and reaches the terminal vertex $j$, because $p_{cc} \geq p_{cp} \geq p_{pp}$. To summarize, we expect core vertices to be on many shortest paths in a graph, whereas peripheral vertices should rarely be on such shortest paths. In other words, because shortest paths between a pair of core vertices are the ones that should on average contain the largest fraction of vertices that are in the core, we find that oversampling such paths is an effective way to extract core parts of a graph. Importantly, it is not sufficient in general to simply use a quantity like weighted BC. For example, for a stock-market correlation network that was examined in Ref.~\cite{corePerApp}, weighted BC cannot distinguish the importance of vertices at all, whereas coreness measures (in particular, {\sc Core-Score} and {\sc Path-Core}) are able to successfully detect core vertices.

To illustrate the effectiveness of the {\sc Path-Core} algorithm, we consider (see Fig.~\ref{fig:EXX3PSCORES}) several instances of the random-graph ensemble $G(p_{cc},p_{cp},p_{pp},n_c,n_p)$ with $p_{cc} > p_{cp} >p_{pp}$. Let $\beta =n_p/n$, where $n=n_c+n_p$, denote the fraction of vertices in the core. We assign the edges independently at random according to the following procedure. The edge probabilities for the core--core, core--periphery, and periphery--periphery pairs of vertices are given by the vector $\mathbf{p}=(p_{cc}, p_{cp}, p_{pp})$, where $p_{cc}= \kappa^2 p$, $p_{cp}=\kappa p$, and $p_{pp}=p$. In our simulations, we fix $n=100$, $\beta =0.5$, and $p=0.25$, and we compute core--periphery structure for 10  instances of the above random-graph ensemble for each of the parameter values $\kappa=1.1, 1.2, \dots, 1.9, 2$. To illustrate the effectiveness of the {\sc Path-Core} algorithm, we show in Fig.~\ref{fig:EXX3PSCORES} the {\sc Path-Core} for all vertices for three different instances of the above block model. We use the parameter values $\kappa=1.3$ (which yields $\mathbf{p}=(0.4225, 0.325, 0.25)$), $\kappa=1.5$ (which yields $\mathbf{p}=(0.5625,0.375, 0.25)$), and $\kappa=1.8$ (which yields $\mathbf{p}=(0.81,0.45,0.25)$).

\begin{figure}[h!]
\begin{center}
\subfigure[$ \kappa = 1.3 $]{\includegraphics[width=0.3\columnwidth]{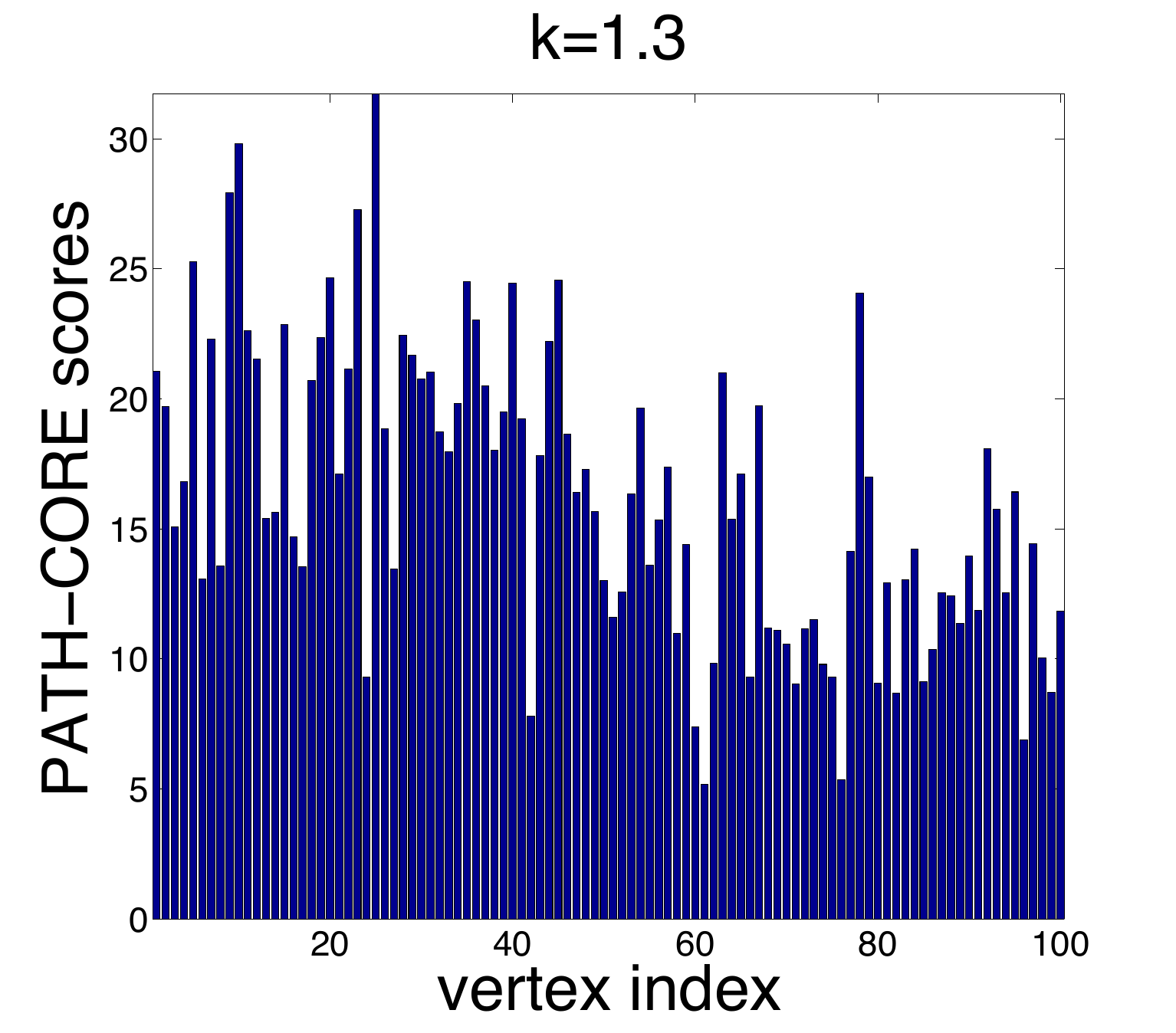}}
\subfigure[$ \kappa = 1.5 $]{\includegraphics[width=0.3\columnwidth]{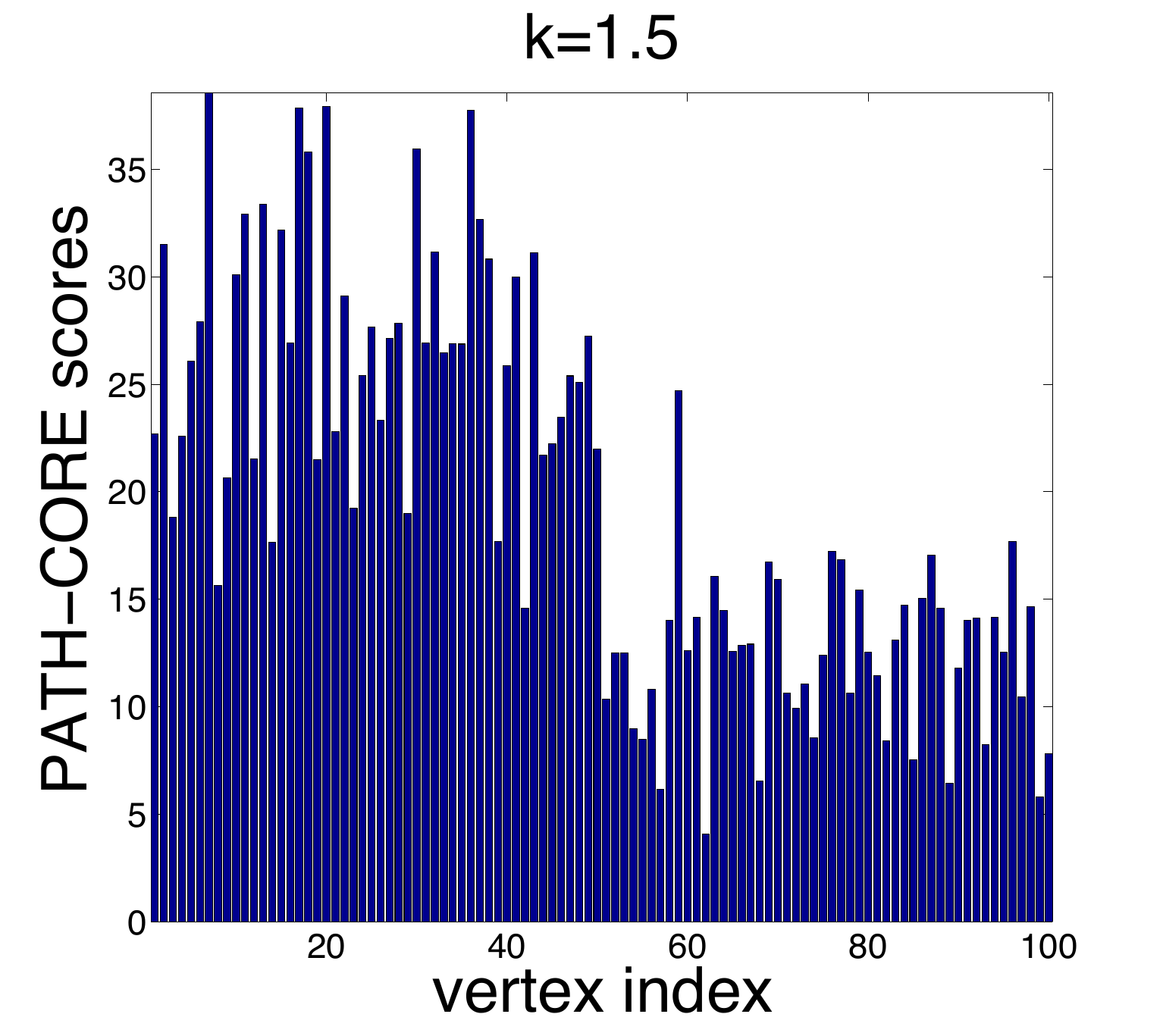}}
\subfigure[$ \kappa = 1.8 $]{\includegraphics[width=0.3\columnwidth]{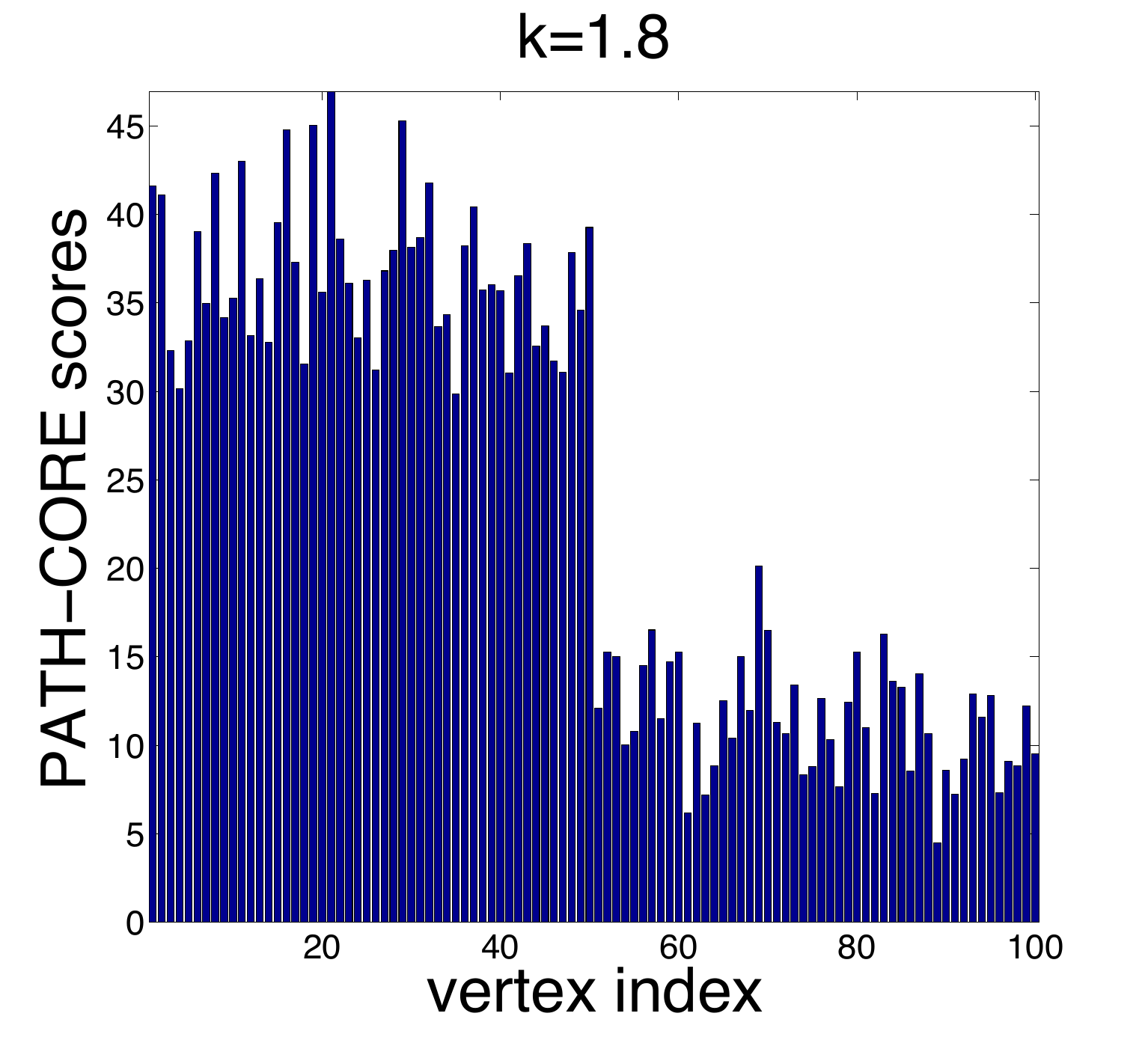}}
\end{center}
\caption{{\sc Path-Core} scores of all $n=100$ vertices, for graphs drawn from three different random-graph ensembles in the family $G(p_{cc},p_{cp},p_{pp},n_c,n_p)$. The vector $\mathbf{p}=(p_{cc}, p_{cp},p_{pp})$ gives the edge probabilities between between a pair of core vertices ($p_{cc}$), a core vertex and a peripheral vertex ($p_{cp}$), and a pair of peripheral vertices ($p_{pp}$).  These probabilities are $p_{cc}= \kappa^2 p$, $p_{cp}=\kappa p$, and $p_{pp}=p$, and we use the fixed value $p=0.25$. The scalar $\kappa$ then parametrizes the ensemble. The values of $\kappa$ are (left) 1.3, (center) 1.5, and (right) 1.8. The first 50 vertices are the planted core vertices, and the remaining 50 vertices are the planted peripheral vertices. 
}
\label{fig:EXX3PSCORES}
\end{figure}

For each of the plots in Fig.~\ref{fig:EXX3PSCORES}, we place the core vertices in the first 50 positions on the horizontal axis, and we place the peripheral vertices in the remaining 50 positions. The vertical axis indicates the {\sc Path-Core} score associated to each vertex. As expected, vertices in the core set have larger {\sc Path-Core} scores than vertices in the periphery set. For $\kappa=1.3$ (left panel), the separation between core and peripheral vertices is not very clear.  As we increase $\kappa$, the separation becomes clearer, and $\kappa = 1.8$ (right panel) exhibits a clear separation between core and peripheral vertices. As expected, larger differences between the edge probabilities $p_{cc}\geq p_{cp} \geq p_{pp}$ in the random-graph ensemble result in clearer separations between core and periphery sets.

For some networks, it is sufficient to have a coreness measure that reflects the probability that a vertex is a core or peripheral vertex. In such a scenario, we view such scores as akin to centrality values \cite{puckmason}. In other situations, however, it is desirable to obtain a classification of a network's vertices as part of a core set or a periphery set. With this in mind, we let {\sc Path-Core}$(i)$ denote the {\sc Path-Core} score of vertex $i$, and we assume without loss of generality that {\sc Path-Core}($1$) $\geq$ {\sc Path-Core}($2$) $\geq \dots \geq$ {\sc Path-Core}($n-1$) $\geq$ {\sc Path-Core}($n$). 
Because the {\sc Path-Core} score gives our calculation for the likelihood that a vertex is in the core set or periphery set (a high {\sc Path-Core} suggests a core vertex), we are left with inferring what constitutes a good ``cut'' of {\sc Path-Core} values to separate core vertices from peripheral ones. In other words, we seek to determine a threshold $\xi$ such that we classify $i$ as a core vertex if {\sc Path-Core}$(i) \geq \xi$ and we classify $i$ as a peripheral vertex if {\sc Path-Core}$(i) < \xi$.

If the size $n_c=\beta n$ of the core set is known, then the problem becomes significantly easier, as we can select the top $n_c$ vertices with the largest {\sc Path-Core} scores and classify them as core vertices. That is, we set $ a = n_c= \beta n$. However, in most realistic scenarios, the size of the core is not known a priori, and it should thus be inferred from the graph $G$ (or from the graph ensemble) and the distribution of the {\sc Path-Core} scores. One possible heuristic approach to obtain such a separation is to sort the vector of {\sc Path-Core} scores in decreasing order and to infer $a$ by searching for a large jump in the sizes of the vector elements. That is, one can seek a ``natural'' separation between high and low {\sc Path-Core} scores (if one exists). An alternative approach is to detect two clusters in the vector of {\sc Path-Core} scores using a clustering algorithm (such as $k$-means clustering). The examples in Fig.~\ref{fig:sortedpscore} (which we generate from the random-graph ensemble $G(p_{cc},p_{cp},p_{pp},n_c,n_p)$ with $p_{cc}= \kappa^2 p$, $p_{cp}=\kappa p$, and $p_{pp}=p$ for $\kappa \in \{1.3, 1.5, 1.8\}$)  illustrate this heuristic very well, as there exists a natural cut point that corresponds to a {\sc Path-Core} score of approximatively $a=20$. This cut correctly assigns the first 50 vertices to the core set and the remaining 50 vertices to the periphery set. In our experiments, note that we fix $p=0.25$ and $\kappa \in [1,2]$, which implies that $p_{cc}, p_{cp}, p_{pp} \in [0,1]$.

Unfortunately, for ``noisy'' networks from this graph ensemble (and for many empirical networks), for which the edge probabilities $p_{cc}$, $p_{cp}$, and $p_{pp}$ are not well-separated, the aforementioned heuristic procedure can yield unsatisfactory results, so a more systematic approach is desirable. In Section \ref{sec:objFuncSync}, we thus introduce the {\sc Find-Cut} algorithm, which maximizes an objective function for partitioning a network into a set of core vertices and a set of peripheral vertices. Using the vector of {\sc Path-Core} scores as an input --- or, indeed, using any other vector of scores that reflects the likelihood that each vertex belongs to the core set --- we consider a large number of possible values of the vector to attempt to find an optimal separation of vertices into a core set and a periphery set that maximizes the objective function in \eqref{cpobjDens}.  See Section \ref{sec:objFuncSync} for a discussion of this objective function and how we maximize it.

\begin{figure}[h!]
\begin{minipage}[b]{0.98\linewidth}
\begin{center}
\subfigure[$ \kappa = 1.5 $]{\includegraphics[width=0.23\columnwidth]{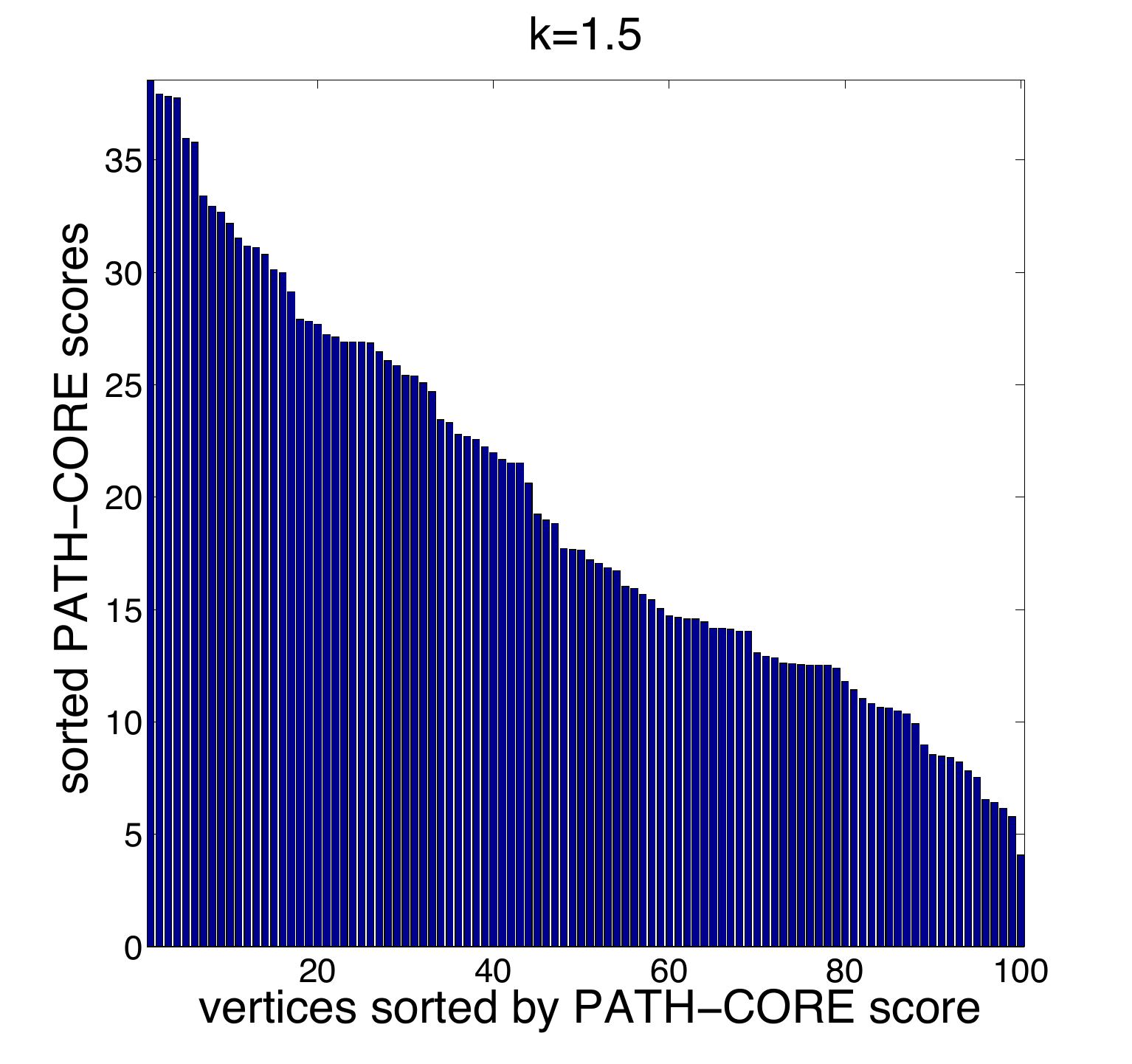}}
\subfigure[$ \kappa = 1.8 $]{\includegraphics[width=0.24\columnwidth]{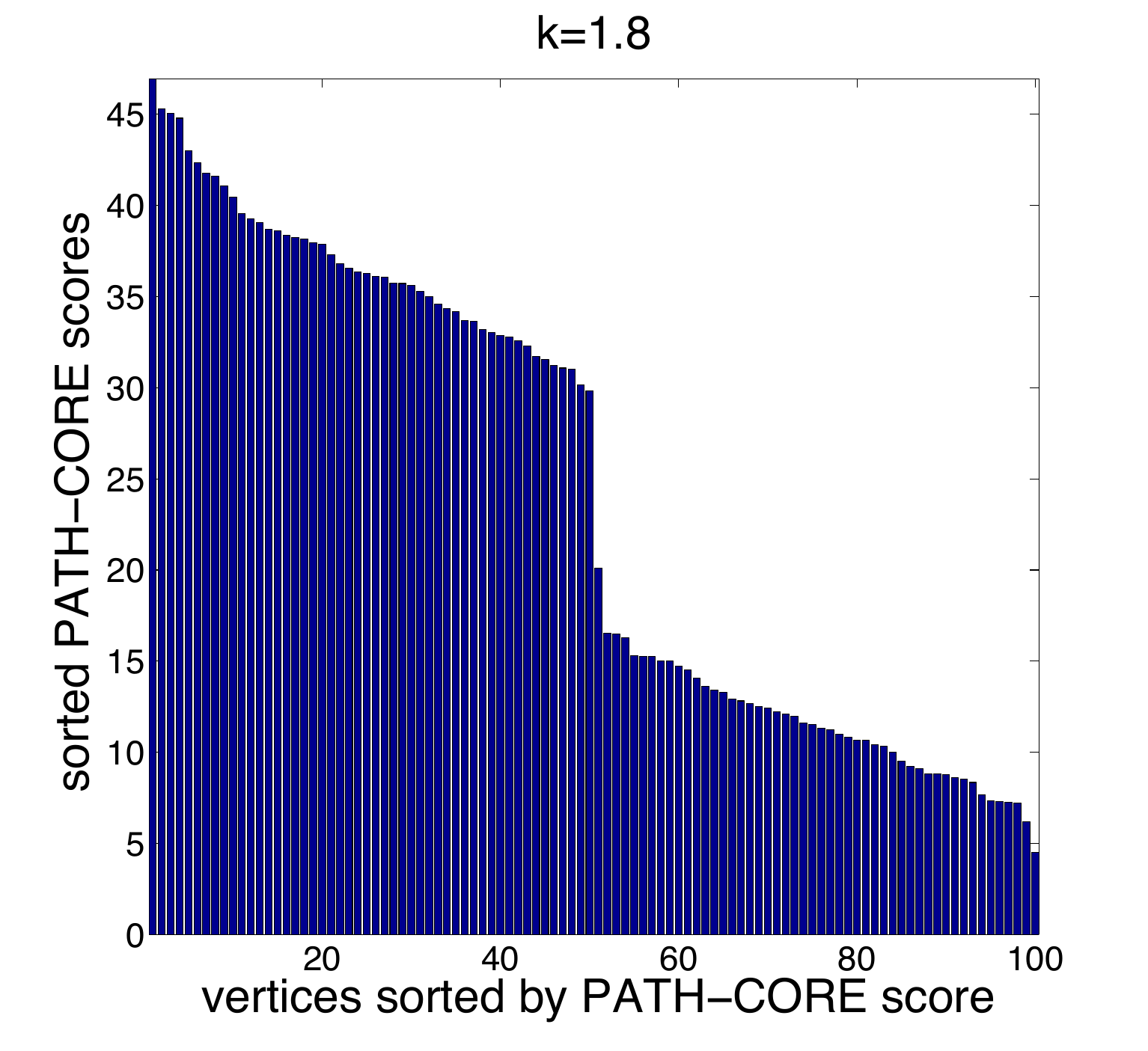}}
\subfigure[$ \kappa = 1.5 $]{
\includegraphics[width=0.23\columnwidth]{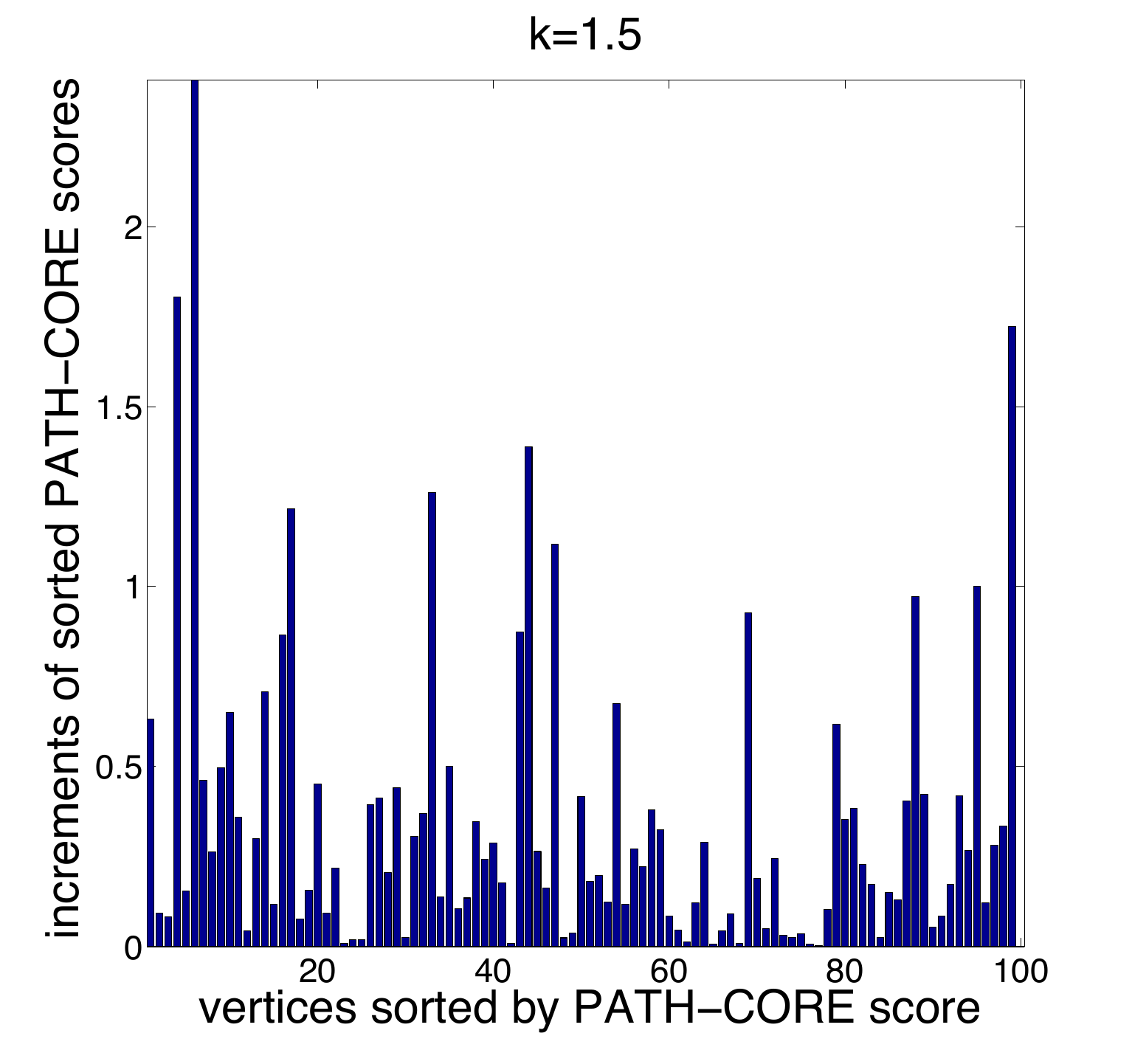}}
\subfigure[$ \kappa = 1.8 $]{
\includegraphics[width=0.23\columnwidth]{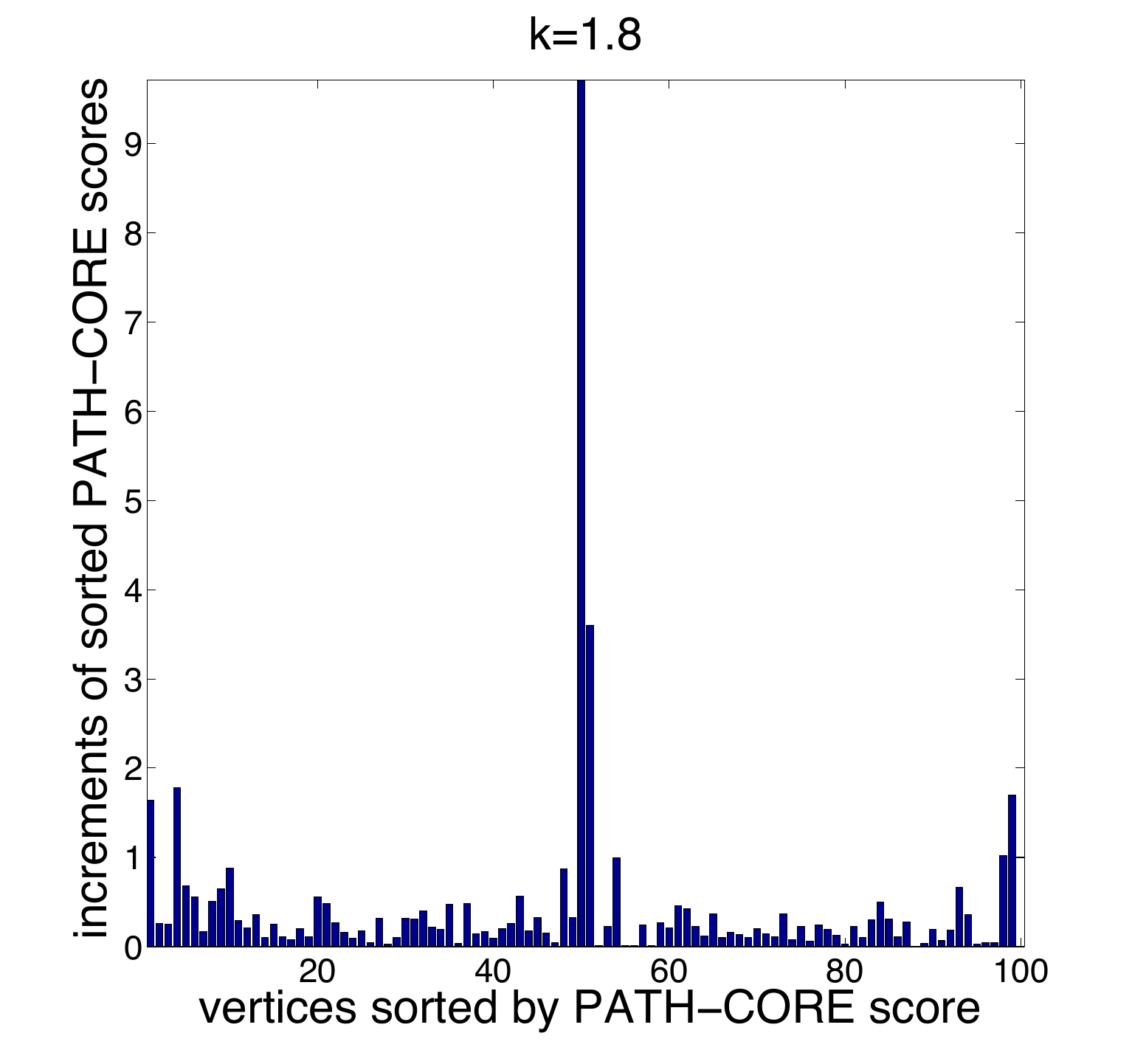}}
\end{center}
\caption{{\sc Path-Core} scores, sorted in decreasing order, for the random-graph ensemble $G(p_{cc},p_{cp},p_{pp},n_c,n_p)$ with $p_{cc}= \kappa^2 p$, $p_{cp}=\kappa p$, $p_{pp}=p$ and parameter values $p = 0.25$, and (a) $\kappa = 1.5$ and (b) $\kappa = 1.8$. When the core--periphery structure is sufficiently prominent, it is possible to separate the vertices by sorting the vector of {\sc Path-Core} scores and inferring the threshold between core and peripheral vertices by considering the largest increment that occurs between two consecutive entries in the vector of sorted {\sc Path-Core} scores. We show the result for $\kappa = 1.5$ in (c) and the result for $\kappa = 1.8$ in (d). In the center and right panels of Fig.~\ref{fig:EXX3PSCORES} and  panels (a) and (b) of the present figure, the largest {\sc Path-Core} score of a peripheral vertex is approximately 20, whereas the lowest {\sc Path-Core} score of a core vertex is approximately 30 (the difference of 10 is revealed by the peak in plot (d) of this figure), and we obtain a clear discrete classification into a set of core vertices and a set of peripheral vertices.
}
\label{fig:sortedpscore}
\end{minipage}
\end{figure}

We present an explicit algorithm for {\sc Path-Core} in Algorithm \ref{pathscorepuck} for the case of unweighted and undirected graphs. This algorithm runs in $\mathcal{O}(m^2)$ time, where we recall that $m = |E|$ is the number of edges in the graph. Intuitively, this is the best that one can achieve (even when computing a {\sc Path-Core} score for just a single vertex), because one must separately consider each graph $G \setminus e$ for all $e \in E$, and finding shortest paths between two vertices has a complexity of $\Theta(m)$. In Appendix 1, we prove the above complexity results and provide pseudocode for the algorithm.

One potential way to drastically reduce the temporal complexity is to sample edges from $G$ via some random process and compute shortest paths only for pairs of adjacent vertices that use these sampled edges. An investigation of the trade-off between accuracy and computational efficiency of this method is beyond the scope of our paper, but it is an interesting direction for future research.


\section{An Objective Function for Detecting Core--Periphery Structure} \label{sec:objFuncSync}

In this section, we introduce an objective function that is suitable for detecting core--periphery structure when there is exactly one core set of vertices and one periphery set. Our function bears some similarity to the rich-club coefficient \cite{colizza2006detecting}, although a crucial difference is that it takes the connectivity of the core, the periphery, and the inter-connectivity between the two into account. (Unlike with rich clubs, low-degree vertices can be core vertices \cite{XZhang2014}.) Using this objective function, we propose the {\sc Find-Cut} algorithm for partitioning the vertex set $V$ into core and periphery sets. As an input, {\sc Find-Cut} takes a vector of scores that reflect the likelihood that each vertex belongs in a network's core set (the probability of belonging to the core set is higher for larger scores), and it attempts to find an optimal separation that maximizes the proposed objective function. That is, instead of trying to find a global optimum of the objective function, the algorithm {\sc Find-Cut} optimizes the objective function over all partitions in which the core vertices have higher likelihood scores than the periphery vertices. A fast general optimization algorithm for this objective function is likely very difficult to achieve, and it is beyond the scope of this paper. We believe that the construction of a suitable objective function brings three advantages. First, the subject of network community structure has benefited greatly from having objective functions to optimize \cite{masonams,fortunato}, and we expect similar benefits for investigations of core--periphery structure. Second, it allows a local-refinement search after the initial algorithm has been applied (in the spirit of Kernighan--Lin vertex-swapping steps for community detection \cite{New06,Richardson2009} and gradient-descent refinement steps in non-convex optimization \cite{NesterovOpt}). Finally, it allows one to compare distinct methods by comparing the corresponding value of the objective function. Nevertheless, one has to proceed cautiously: a value of an objective function need not provide a definitive answer, and it can be misleading \cite{good2010,peixoto2014}. 

Before introducing an objective function for studying core--periphery structure, we first revisit a well-known graph-partitioning problem to highlight the similarity between the two situations. {\sc Min-Cut}, an instance of a graph-partitioning problem, is concerned with dividing a graph into two (similarly-sized) subgraphs while minimizing the number of edges that are incident to vertices in both subgraphs. More generally, a large family of graph-partitioning problems seek to decompose a graph into $k$ disjoint subgraphs (i.e., clusters) while minimizing the number of cut edges (i.e., edges with endpoints in different clusters).  Given the number $g$ of clusters, the $g$-way graph-partitioning problem searches for a partition $V_1, \dots ,V_g$ of the vertex set $V$ that minimizes the number of cut edges
\begin{equation}
 	\mbox{Cut}(V_1, \dots ,V_g) = \sum_{i=1}^{g} | E(V_i, \overline{V_i}) | \,,
\label{cutObj}
\end{equation}
where $\overline{X} = V \setminus X$ and the number of edges between $X \subset V$ and $Y \subset V$ is $|E(X,Y)| = \sum_{i \in X, j \in Y} A_{ij}$. However, it is well-known that trying to minimize $\mbox{Cut}(V_1, \dots ,V_g)$ favors cutting off weakly-connected individual vertices from a graph and can thus lead to trivial partitions. To penalize clusters $V_i$  of small size, Shi and Malik \cite{shimalik} suggested minimizing the normalized cut 
\begin{equation}
	 \mbox{NCut}( V_1, \dots ,V_g ) = \sum_{i=1}^g \frac{\mbox{Cut}(V_i,\overline{V_i})}{\mbox{SK}(V_i)}\,,
\label{ncutObj}
\end{equation}
where $\mbox{SK}(V_i) = \sum_{i \in V_i} d_i$ and $d_i$ denotes the degree of vertex $i$ in the original graph $G$.

A natural choice for an objective function to detect core--periphery structure is to maximize the number of edges between pairs of core vertices and also between core and peripheral vertices, while allowing as few edges as possible between pairs of peripheral vertices. In other words, our approach is complementary to that of the graph-cut objective function (\ref{cutObj}).  However, instead of minimizing the number of cut edges across the core and periphery sets (i.e., across clusters), we maximize the connectivity between pairs of core vertices and between core and peripheral vertices while minimizing the connectivity between pairs of peripheral vertices. We thus want to maximize
\begin{equation}
	\mbox{CP-connectivity}(V_C,V_P) = E(V_C,V_C) + E(V_C,V_P) -  E(V_P,V_P)\,.
\label{cpobj1}
\end{equation}
Our aim is to find a partition $\{V_C,V_P\}$ of the vertex set $V$ that maximizes \mbox{CP-connectivity}$(V_C,V_P)$, under the constraint that $\vert V_C \vert ,  \vert V_P \vert \geq b$, where $b$ is the minimum number of core or peripheral vertices (hence, $n-b$ is the maximum number of core or peripheral vertices) to avoid a large imbalance between the sizes of the core and periphery sets. In other words, we seek a balanced partition, and a higher value of $b$ indicates a smaller difference between the sizes of the core and periphery sets. This constraint is required to avoid a trivial solution in which all of the vertices are placed in the core set. Furthermore, note that the objective function \eqref{cpobj1} has only one variable because of the constraint $E(V_C,V_C) + E(V_C,V_P) + E(V_P,V_P) = m$. In practice, we have found this approach to be rather unstable in the sense that \eqref{cpobj1} often attains its maximum at $|V_C| = b$ or $|V_P| = b$. It thereby leads to disproportionately-sized sets of core and peripheral vertices compared to the ``ground truth'' in problems with planted core--periphery structure (e.g., from the block model $G(p_{cc},p_{cp},p_{pp},n_c,n_p)$, where we recall (see Section \ref{sec:PathCore}) that $n_c$ (respectively, $n_p$) denotes the size of the core (respectively, periphery) sets, $p_{cc}$ is the probability that there is an edge between a given pair of core nodes, $p_{cp}$ is the probability that there is an edge between a core node and a peripheral node, and $p_{pp}$ is the probability that there is an edge between a pair of peripheral nodes.
This situation is analogous to the trivial solution that one obtains for unconstrained graph-partitioning problems. We have been able to ameliorate this problem (though not remove it completely) by incorporating a normalization term in the spirit of the normalized cut function \eqref{ncutObj}. Instead of maximizing the number of edges between core vertices and between core and peripheral vertices while minimizing the number of edges between peripheral vertices, we choose to maximize the edge \emph{density} among core vertices and between core and peripheral vertices while minimizing the edge density among peripheral vertices. Finally, we also add a term to the objective function that penalizes imbalances between the sizes of the core and periphery sets (or penalizes a deviation from the expected proportion of core vertices) if such information is available. The maximization of our new objective function is over the set of all possible partitions of the vertex set into two disjoint sets (the core set $V_C$ and the periphery set $V_P$). The function is
\begin{equation}
	\mbox{CP-density}(V_C,V_P) = \frac{ |E(V_C,V_C)|}{ \mbox{Vol}(V_C,V_C)} + \frac{|E(V_C,V_P)|}{ \mbox{Vol}(V_C,V_P)} - \frac{|E(V_P,V_P)|}{ \mbox{Vol}(V_P,V_P) } - \gamma \left| \frac{|V_C|}{n}-\beta \right|\,,
\label{cpobjDens}
\end{equation}
where 
\begin{equation}
	 \mbox{Vol}(X,Y) = \left\{
	     \begin{array}{cl}
	    	                  |X| |Y|\,,    &  \;   \textrm{if } X \neq Y \\
	      \frac{1}{2}  |X| (|X|-1)\,,  &  \;  \textrm{if } X=Y  \\
	     \end{array}
	   \right.
	\label{defVolXY}
\end{equation}
denotes the total possible number of edges between sets $X$ and $Y$. In the penalty term, $\beta$ denotes the prescribed fraction of core vertices in the graph (if it is known in advance), and $\gamma$ tunes the sensitivity of the objective function to the size imbalance between the core and periphery sets. Note that $\beta$ can either be prescribed in advance or construed as a parameter that guides the maximization towards a solution with a certain target size for the core set. For simplicity, we limit ourselves to the case $\gamma=0$. That is, we assume no prior knowledge of the ratio between the number of core and peripheral vertices. 
In practice, however, we do implicitly assume a lower bound on the sizes of the core and periphery sets of vertices to ameliorate a ``boundary effect'' that yields solutions with a very small number of vertices in the core set or periphery set. If one explicitly wants to allow the possibility of a small set of core or peripheral vertices, then one can set $b=0$. For some of our experiments on synthetic graphs in Section \ref{sec:numSims}, we compare the performance of our proposed algorithms both when $\beta$ is known and when it is unknown.

\begin{algorithm}[h!]
\caption{ {\sc Find-Cut}: Classifies the vertices of a graph $G$ into a set $V_C$  of core vertices and a set $V_P$ of peripheral vertices based on a score associated to each vertex that reflects the likelihood that it is in the core. 
}
\label{FindCut}
\begin{algorithmic}[1]
\REQUIRE Vector of scores ${\bf s} = (  s_1,\dots,s_n ) \in \mathbb{R}^n$ associated to the $n$ vertices of a graph.
\STATE Sort the entries of the vector ${\bf s}$ in decreasing order. Assume without loss of generality that $s_1 \geq s_2 \geq \dots \geq s_{n-1} \geq s_n$.
\STATE  Let $X_C = \{ 1, \dots, n_c \}$ and $Y_C = \{ n_c +1, \dots, n \}$ for any $n_c \in \{1,\dots,n\}$. Find the value $n_c$ that maximizes the objective function given by Eq.~\eqref{cpobjDens} with $ \gamma=0$). That is, we find
\begin{equation}
   \Phi^*= \frac{1}{n} \left[\underset{n_c \in \{b,\dots,n-b \}}{\text{max}}
     \left( \frac{ |E( X_C , X_C)|}{  \mbox{Vol}(X_C, X_C) } + \frac{ |E( X_C, Y_C )|}{  \mbox{Vol}(X_C, Y_C) }  - \frac{ |E( Y_C , Y_C)|} {  \mbox{Vol}(Y_C, Y_C)  } \right)\right]\,,
\label{eq:FindCut}
\end{equation}
where $b$ denotes a lower bound on the size of the core and periphery sets (which we use to avoid 
 solutions with either a very small core set or a very small periphery set).
\STATE Define the core set $V_C =\{ 1,\dots,n_c \}$ and the periphery set $V_P =  \{ n_c+1,\dots,n \}$ 
\end{algorithmic}
\end{algorithm}

We summarize the {\sc Find-Cut} approach in Algorithm \ref{FindCut}, and we remark that one can also add an iterative post-processing refinement step that is reminiscent of the gradient-descent algorithm \cite{NesterovOpt} or of Kernighan--Lin vertex swaps~\cite{New06,Richardson2009}.  At each iteration, one can choose to move the vertex from the core set to the periphery set (or the other way around) that leads to the largest increase in the objective function \eqref{cpobjDens}. Alternatively, if one wishes to maintain the current size of the core and periphery sets, then one can choose to swap a pair of vertices from their assignments (of core or periphery) that leads to the largest increase in the objective function.


\section{{\sc LowRank-Core}: Core--Periphery Detection via Low-Rank Matrix Approximation}  \label{sec:rank2}

Another approach for detecting core--periphery structure in an unweighted network\footnote{For weighted graphs, one needs to think further about how to use such an approach, as we are relying on perturbing a low-rank matrix.} is to interpret its adjacency matrix as a perturbation of a low-rank matrix. Perturbations of low-rank matrices were used recently in  \cite{barranca2015} for classifying networks and identifying small-world structure --- by capturing the dense connectivity of nodes within communities and the sparse connectivity between communities --- and this type of an approach should also be useful for studying core--periphery structure.

Consider the block model 
\begin{equation}\label{nullmodel}
	A_0 =
		\begin{tabular}{|c|c|}
	\hline
	\mb{1}$_{n_c \times n_c}$ & \mb{1}$_{n_c \times n_p}$ \\
	 \hline
	 \mb{1}$_{n_p \times n_c}$ & \mb{0}$_{n_p \times n_p}$ \\
	 \hline
		\end{tabular}\,,
\end{equation}
which assumes that core vertices are fully connected among themselves and with all vertices in the periphery set and that no edges exist between any pair of peripheral vertices. The block model in Eq.~\eqref{nullmodel} corresponds to an idealized block model that Borgatti and Everett \cite{BorgattiCore} employed in a discrete notion of core--periphery structure. The rank of the matrix $A_0$ is 2, as any $3\times3$ submatrix has at least two identical rows or columns. Consequently, $\mathrm{det}(A_0) = 0$. Alternatively, when the core and periphery sets have the same size, $n_c=n_p$ with $n=n_c+n_p$, one can write the matrix $A_0$ as the following tensor product of matrices:
\begin{equation}\label{my2nullmodel2} 
	\bar{A}_0 =
		\begin{tabular}{|c|c|}
	\hline
	\mb{1}$_{n_c \times n_c}$ & \mb{1}$_{n_c \times n_c}$ \\
	 \hline
	 \mb{1}$_{n_c \times n_c}$ & \mb{0}$_{n_c \times n_c}$ \\
	\hline
		\end{tabular}
	=
	R \otimes\mb{1}_{n_c \times n_c}\,,
\qquad
R = 
\left[ \begin{array}{cc}
    1 & 1 \\
    1 & 0 \\
\end{array} \right]\,.
\end{equation}
The eigenvalues of $\bar{A}_0$ are direct products of the eigenvalues of $R$ and $\mb{1}_{n_c \times n_c}$. These eigenvalues are
\begin{equation*}
	\left\lbrace \frac{1 - \sqrt{5} } {2} ,  \frac{1 + \sqrt{5} }{2}   \right\rbrace    \otimes    \left\lbrace  n_c,0^{(n_c-1)}    \right\rbrace  = \left\lbrace  \left(n_c \frac{1 \pm \sqrt{5} }{2}\right)^{(2)}, 0^{(n-2)}  \right\rbrace\,,
\end{equation*}	
where a superscript denotes the multiplicity of an eigenvalue. 

The simplistic block models in equations (\ref{nullmodel},\ref{my2nullmodel2}) assume that a network has only one core set and one periphery set. Consequently, the block-model matrix has a rank of $2$. The matrix rank is higher for more complicated core--periphery block models. For example, the block model in Fig.~\ref{fig:blockCorPerLine} has a global community structure --- there are $g=4$ communities, which each correspond to a block in the block-diagonal matrix --- and a local core--periphery structure (because each community has a core--periphery structure). As indicated in Ref.~\cite{puckmason}, one can also construe such a structure (by permuting the rows and columns of the matrix) as having a global core--periphery structure and a local community structure.

\begin{figure}[h!]
\begin{center}
\includegraphics[width=0.25\columnwidth]{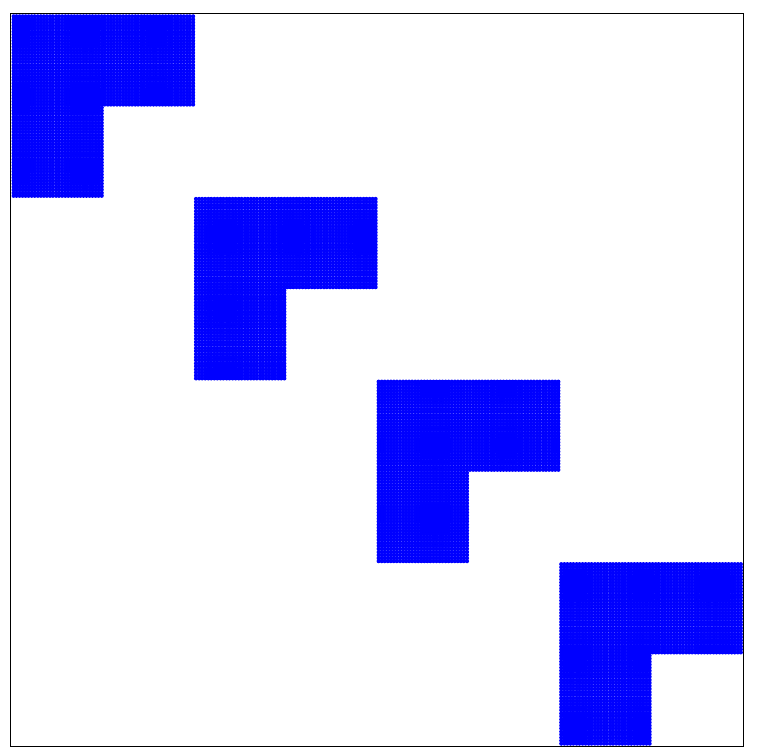}
\end{center}
\caption{A block model with $g=4$ diagonal blocks that are each of the form of the block model in Eq.~\eqref{nullmodel}.}
\label{fig:blockCorPerLine}
\end{figure}

Let $B_g(A_0)$ denote a ``hierarchical'' ensemble of size $n \times n$ that is composed of $g$ diagonal blocks that are of each of size $l \times l$ (thus, $n=lg$), where each diagonal block is of the form of the block model in Eq.~\eqref{nullmodel}. If we let $\lambda_1$ and $\lambda_2$ denote the two nonzero eigenvalues of $A_0$ and let $I_{g}$ denote the identity matrix of size $g$, then we can also write $B_g(A_0)$ as a tensor product of matrices: 
\begin{equation}
	B_g(A_0) =  I_g \otimes  A_0 \,, \quad \mbox{with  eigenvalues} \quad \{B_g(A_0)\} =  \{ 1^{g} \} \otimes \{\lambda_1, \lambda_2, 0^{l-2}\} = \left\{\lambda_1^{(g)},\lambda_2^{(g)},0^{(n-2g)} \right\}\,. 
\end{equation}
Therefore, in the simplistic scenario in which each diagonal block has one core set and one periphery set (and thus has rank $2$), the rank of  $B_g(A_0)$ is $2g$. 

Motivated by the low-rank structure of the above block-model networks, it is useful to consider the possibility of recovering a network's unknown structure using a simple low-rank projection of its adjacency matrix. For the remainder of this section, we focus on the simple core--periphery structure whose rank-$2$ block model is given by Eq.~\eqref{nullmodel} (with one core set and one periphery set). In practice, we construe the adjacency matrix $A$ of an observed graph $G$ as a low-rank perturbation of the block model $A_0$. In other words, we decompose $A$ as
\begin{equation}
	A = A_0 + W\,,
	\label{decomp}
\end{equation}
where $W$ is a ``noise matrix'' whose entries $\{-1,0,1\}$ are determined by a mixture model \cite{mclachlan2000} that involves block-model parameters.  The entries of $W$ are
\begin{equation}
	W_{ij} = \left\{
	     \begin{array}{rl}
	     -1\,, & \;\; \text{ with probability } \; 1 - p_{cc} \qquad \text{(i.e., if } i,j \in V_C)\,,  \\
	     -1\,, & \;\; \text{ with probability } \; 1 - p_{cp} \qquad \text{(i.e., if } i \in V_C \quad \mbox{and} \quad j \in V_P)\,, \\
	      1\,, & \;\; \text{ with probability } p_{pp} \qquad\qquad \text{(i.e., if } i,j \in V_P)\,, \\
	      0\,, & \;\; \text{ otherwise\,.} \\
          \end{array}
   \right.
\label{entries_W}
\end{equation}
Note that $W$ is a random block-structured matrix with independent entries, and its expected value is the rank-$2$ matrix with entries
\begin{equation}
	\mathbb{E}( W_{ij} )= \left\{
	     \begin{array}{rl}
	     p_{cc} -1\,,  & \;\; \text{ if } i,j \in V_C\,,  \\
	     p_{cp} -1\,,  & \;\; \text{ if } i \in V_C \quad \mbox{and} \quad j \in V_P\,, \\
	     p_{pp}\,,      & \;\; \text{ if } i,j \in V_P\,. \\
	     \end{array}
	   \right.
\label{xsxMeanEntriesW}
\end{equation}

To ``denoise'' the adjacency matrix $A$ and recover the structure of the block model, we consider its top two eigenvectors $\{{\bf v}_1,{\bf v}_2\}$, whose corresponding two largest (in magnitude) eigenvalues are $\{\lambda_1, \lambda_2\}$, and we compute the rank-$2$ approximation
\begin{equation}
	\hat{A} =
	 \left
	[ \begin{array}{cc} {\bf v}_1  &  {\bf v}_2  \\  \end{array} \right]
	\left[ \begin{array}{cc}
	    \lambda_1 & 0  \\
	   			     0 &\lambda_2 \\
	\end{array} \right]
	\left[    \begin{array}{c}
					{\bf v}_1^T  \\
					{\bf v}_2^T \\
					\end{array}
	\right]\,.
\label{rank2proj}
\end{equation}
As $A$ more closely approximates the block model, which we can construe as a sort of ``null model'', the spectral gap between the top two largest eigenvalues and the rest of the spectrum becomes larger (as illustrated by the plots in the second column of Fig.~\ref{fig:rank2recovery}).  In other words, as the amount of noise in (i.e., the perturbation of) the network becomes smaller, the top two eigenvalues $\{\lambda_1, \lambda_2\}$ become closer to the eigenvalues $\left\{\lambda_1 = n_c \left(\frac{1 + \sqrt{5} }{2}\right),  \lambda_2 = n_c \left(\frac{1 - \sqrt{5} }{2}\right)\right\}$ of the block model.

To illustrate the effectiveness of our low-rank projection in computing a coreness score, we consider two synthetically generated networks based on the  SBM that we introduced previously. We use the edge probabilities $(p_{cc}, p_{cp}, p_{pp}) = (0.7,0.7,0.2)$ and $(p_{cc}, p_{cp}, p_{pp}) = (0.8,0.6,0.4)$. In the left column of Fig.~\ref{fig:rank2recovery}, we show their corresponding adjacency matrices. The spectrum, which we show in the middle column, reveals the rank-2 structure of the networks. In the second example (which we show in the bottom row of the figure), the large amount of noise causes the second largest eigenvalue value to merge with the bulk of the spectrum.

We then use the denoised matrix $\hat{A}$ to classify vertices as part of the core set or the periphery set by considering the degree (i.e., the row sums of $\hat{A}$) of each vertex. We binarize $\hat{A}$ by setting its entries to $0$ if they are less than or equal to $0.5$ and setting them to $1$ if they are larger than $0.5$, and we denote the resulting binarized matrix by $\hat{A}_t$. 
We remark that, following the rank-2 projection, we observe in practice that all entries of $\hat{A}$ lie in the interval $[0,1]$. (We have not explored the use of other thresholds besides $0.5$ for binarizing $\hat{A}$.)  In the right column of Fig.~\ref{fig:rank2recovery}, we show the recovered matrix $\hat{A}_t$ for our two example networks. Note in both examples that the denoised matrix $\hat{A}_t$ resembles the core--periphery block model $G(p_{cc},p_{cp},p_{pp},n_c,n_p)$ much better than the initial adjacency matrix $A$. Finally, we compute the degree of each vertex in $\hat{A}_t$, and we call these degrees the {\sc LowRank-Core} scores of the vertices. We use the {\sc LowRank-Core} scores to classify vertices as core vertices or peripheral vertices. If one knows the fraction $\beta$ of core vertices in a network, then we choose the top $\beta n$ vertices with the largest {\sc LowRank-Core} score as the core vertices. Otherwise, we use the vector of {\sc LowRank-Core} scores as an input to the {\sc Find-Cut} algorithm that we introduced in Sec.~\ref{sec:objFuncSync}. Although a theoretical analysis of the robustness to noise of our low-rank approximation for core--periphery detection is beyond the scope of the present paper, we expect that results from the matrix-perturbation literature, such as  Weyl's inequality and the Davis--Kahan sin($\Theta$)-theorem \cite{Bhatia}, as well results on low-rank deformations of large random matrices \cite{RajRaoBenaychGeorges} (analogous to the results of F\'eral and P\'ech\'e on the largest eigenvalue of rank-1 deformations of real, symmetric random matrices \cite{FeralPeche}) could lead to theoretical results that characterize the sparsity and noise regimes for which the rank-2 projection that we proposed above is successful at separating core and peripheral vertices. A possible first step in this direction would be to consider a simplified version of the graph ensemble $G(p_{cc},p_{cp},p_{pp},n_p,n_c)$ by setting $p_{cc}=p_{cp}=1-p_{pp} = 1-\eta$, where $\eta \in (0,1)$.

\begin{algorithm}[h!]
\caption{{\sc LowRank-Core}: Detects core--periphery structure in a graph based on a rank-$2$ approximation.
}
\begin{algorithmic}[1]
\REQUIRE Adjacency matrix $A$ of the simple graph $G=(V,E)$ with $n$ vertices and $m$ edges.
\STATE Compute $\{\lambda_1$,$\lambda_2\}$, the top two largest (in magnitude) eigenvalues of $A$, together with their corresponding eigenvectors $\{{\bf v}_1,{\bf v}_2\}$.
\STATE Compute $\hat{A}$, a rank-$2$ approximation of $A$, as indicated in Eq.~\eqref{rank2proj}.
\STATE Threshold the entries of $\hat{A}$ at $0.5$ (so that entries strictly above $0.5$ are set to $1$ and all other entries are set to $0$), and let $\hat{A}_t$ denote the resulting graph.
\STATE Compute the {\sc LowRank-Core} scores as the 
degrees of $\hat{A}_t$.
\STATE If the fraction of core vertices $\beta$ is known, identify the set of core vertices as the top $\beta n$ vertices with the largest {\sc LowRank-Core} scores. 
\STATE If $\beta$ is unknown, use the vector of {\sc LowRank-Core} scores as an input to the {\sc Find-Cut} algorithm in Algorithm~\ref{FindCut}. 
\end{algorithmic}
\label{algoRank2}
\end{algorithm}

\vspace{-5pt}
\begin{figure}[h!]
\begin{center}
\subfigure[$ p_{cc}=0.7, p_{cp} = 0.7, p_{pp}=0.2 $]
{\includegraphics[width=0.23\columnwidth]{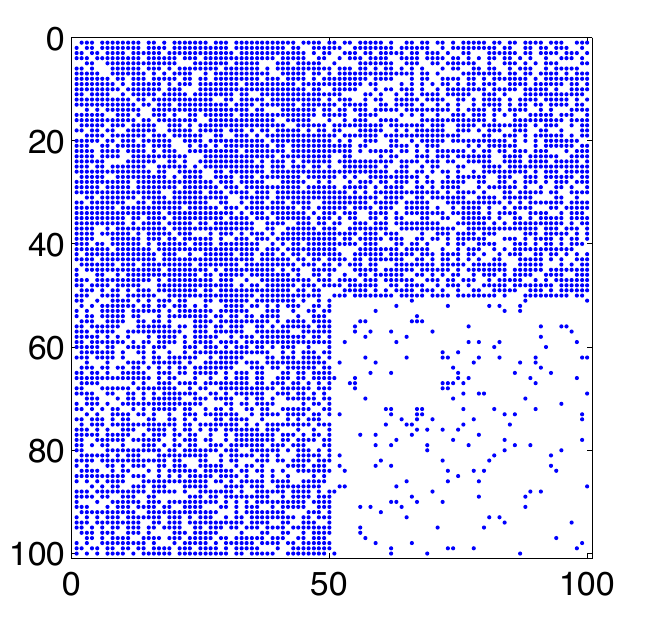}
\includegraphics[width=0.30\columnwidth]{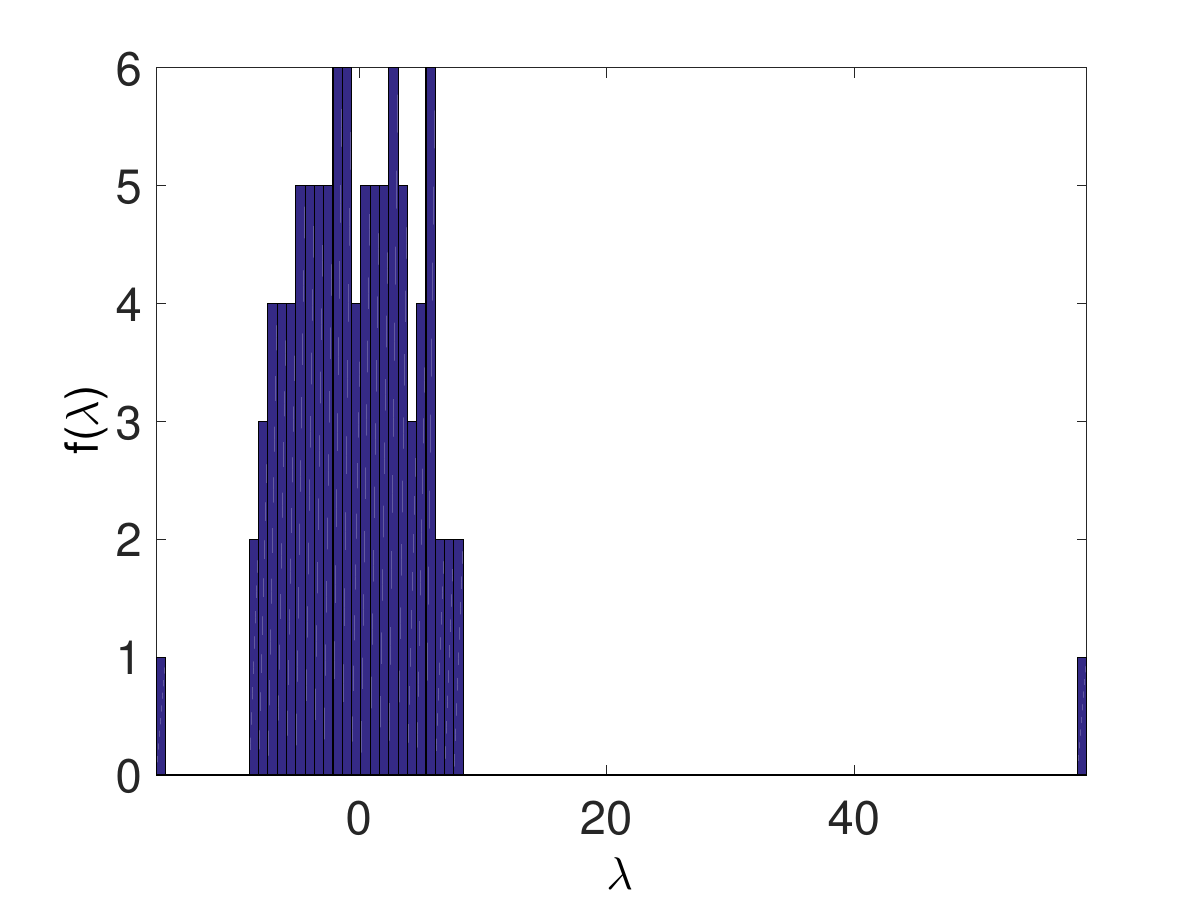}
\includegraphics[width=0.23\columnwidth]{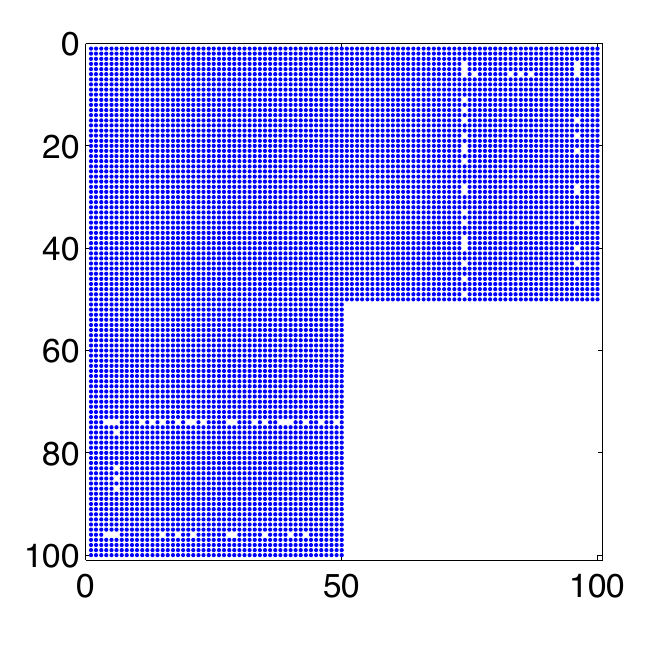}
\includegraphics[width=0.3\columnwidth]{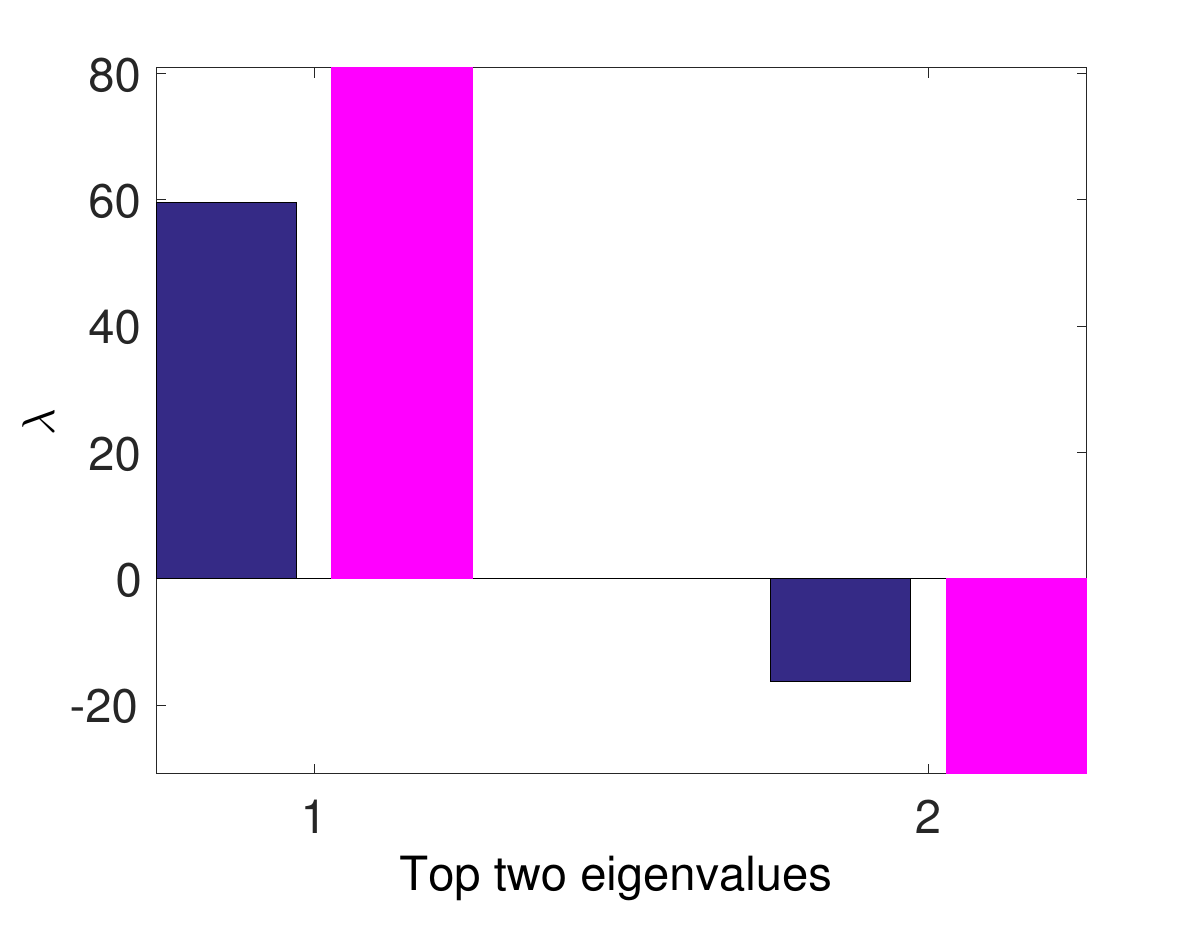}
}
\subfigure[$p_{cc}=0.8, p_{cp}=0.6, p_{pp}=0.4$]{
\includegraphics[width=0.23\columnwidth]{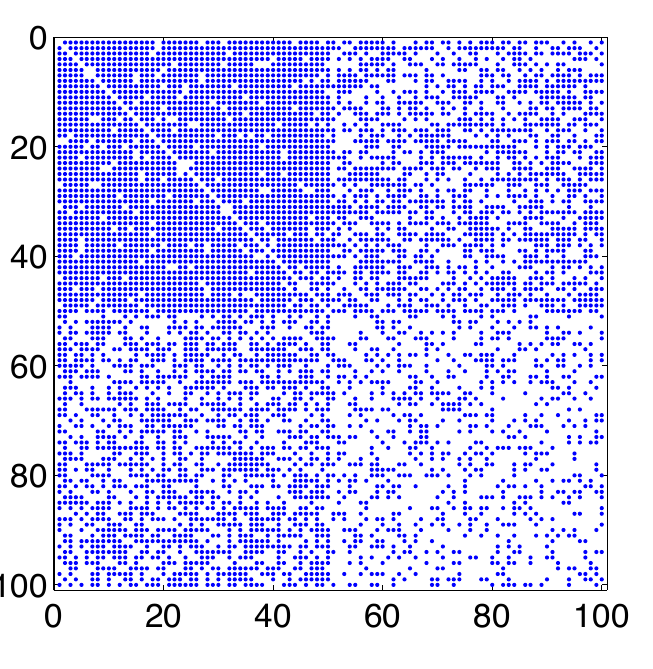}
\includegraphics[width=0.30\columnwidth]{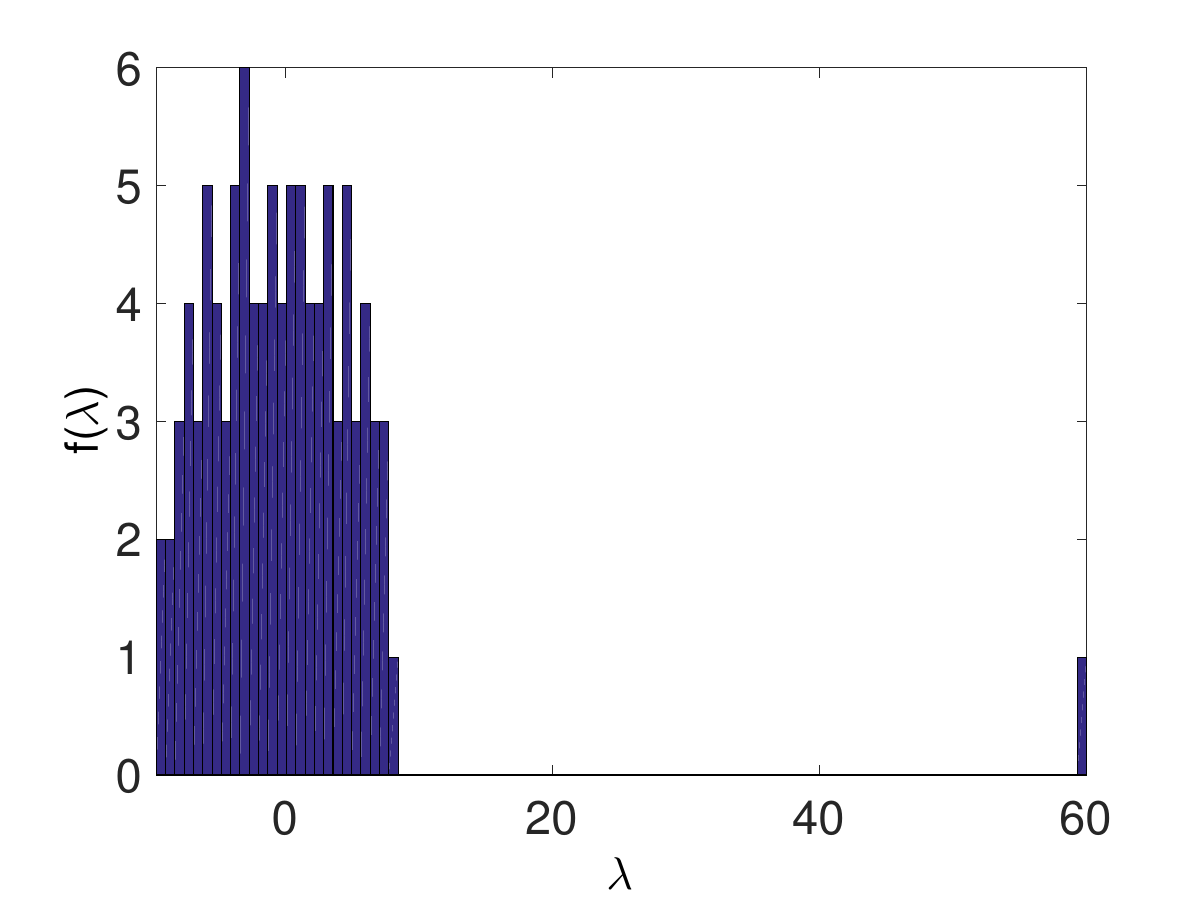}
\includegraphics[width=0.23\columnwidth]{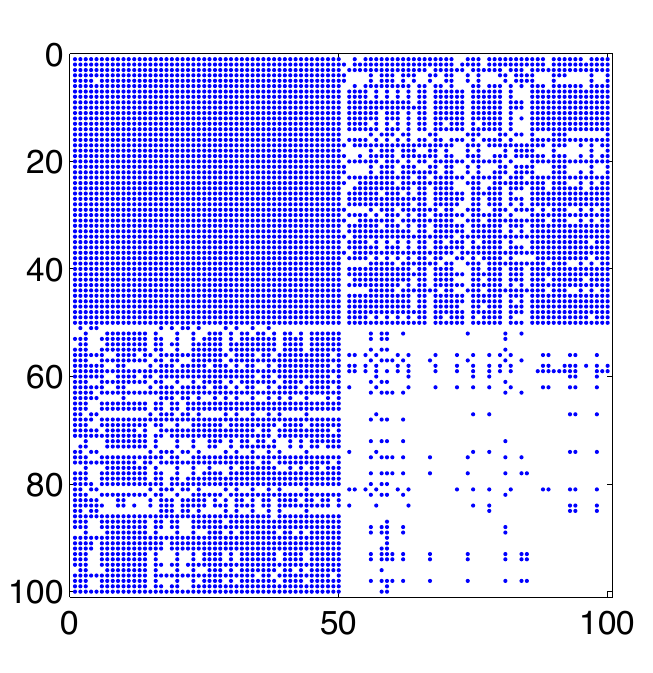}
\includegraphics[width=0.3\columnwidth]{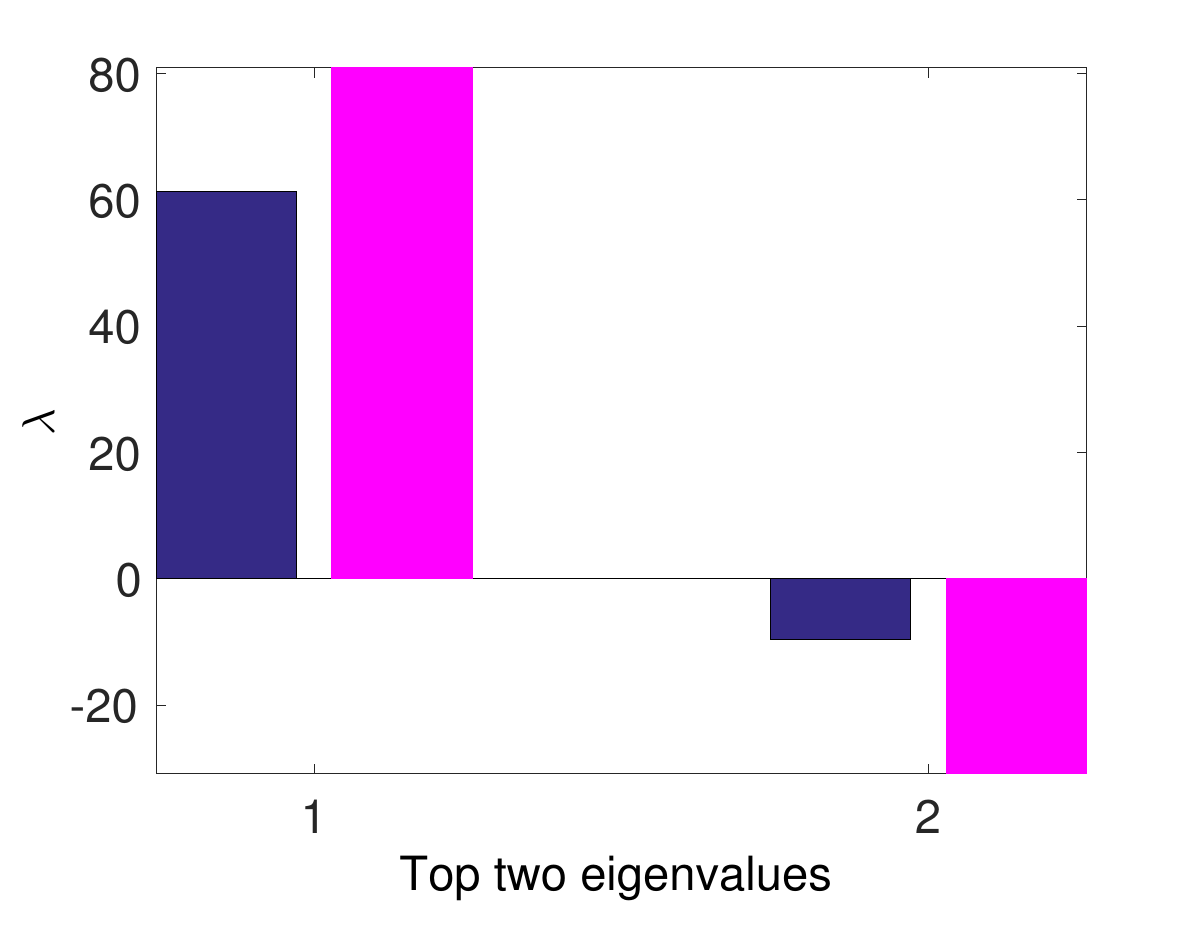}
}
\end{center}
\caption{(Column 1) Original adjacency matrices $A$ from the stochastic block model (SBM) $G(p_{cc},p_{cp},p_{pp},n_c,n_p)$ with edge probabilities $p_{cc}$ for edges between two core vertices, $p_{cp}$ for edges between core vertices and peripheral vertices, and $p_{pp}$ for edges between two peripheral vertices.
(Column 2) Histogram $f(\lambda)$ of the eigenvalues of the original adjacency matrices $A$. 
(Column 3) The matrices $\hat{A}_t$ that we obtain after the rank-$2$ projection and thresholding. 
(Column 4) Bar plot of the top two eigenvalues  $\left\{\lambda_1 = n_c \left(\frac{1 + \sqrt{5} }{2}\right),  \lambda_2 = n_c \left(\frac{1 - \sqrt{5} }{2}\right)\right\}$ of the block model (\ref{nullmodel},\ref{my2nullmodel2}) (blue/dark) versus the top two eigenvalues of many realizations of the SBM $G(p_{cc},p_{cp},p_{pp},n_c,n_p)$ (pink/light), averaged over 100 experiments.
}
\label{fig:rank2recovery}
\end{figure}


\section{{\sc Lap-Core}: Laplacian-Based Core--Periphery Detection} \label{sec:Laplacian}

In this section, we explore the utility of employing Laplacian eigenvectors for detecting core--periphery structure. (As with {\sc Path-Core}, this approach is applicable to either unweighted or weighted graphs.) The \emph{combinatorial Laplacian} matrix associated to the adjacency matrix $A$ of a graph $G$ is $F=D-A$, where $D$ is a diagonal matrix and $D_{ii}$ denotes the degree of vertex $i$ in the case of an unweighted graph. For a weighted graph, $D_{ii}$ denotes the sum of the weights associated to vertex $i$. The solutions of the generalized eigenvalue problem $F{\bf x} = \lambda D {\bf x}$ are related to the solutions of the eigenvalue problem $L{\bf x} = \lambda {\bf x}$, where $L=D^{-1}A$ is often called the \emph{random-walk Laplacian} of $G$. Using $L= I - D^{-1} F$, one can write the random-walk Laplacian in terms of the combinatorial Laplacian. Because $L$ is a row-stochastic matrix, one can interpret it as a transition probability matrix of a Markov chain whose states are the vertices of $G$. In this interpretation, the matrix element $L_{ij}$ denotes the transition probability that a random walker jumps from vertex $i$ to vertex $j$ in a single step. If the pair $(\lambda,{\bf v})$ is an (eigenvalue, eigenvector) solution to  $L{\bf x} = \lambda {\bf x}$, then $(1-\lambda,{\bf v})$ is a solution to $F{\bf x} = \lambda D {\bf x}$. The top\footnote{The top eigenvectors of the random-walk Laplacian $L$ are the eigenvectors that correspond to the largest eigenvalues of $L$. That is, these are the eigenvalues closest to $\lambda_1 = 1$, the largest eigenvalue of $L$. The bottom eigenvectors of $L$ correspond to the smallest eigenvalues of $L$. The eigenvalues $\lambda_1 = 1 \leq \lambda_2 \leq \dots \leq \lambda_n $ of $L$ satisfy $|\lambda_i| \leq 1$ for all $i \in \{1,\dots,n\}$.} eigenvectors of the random-walk Laplacian define the coarsest modes of variation (i.e., slowest modes of mixing) in a graph, and they have a natural interpretation in terms of a random walk on the graph (and thus as a toy model of a conservative diffusion process). There exists a rich literature in the machine-learning, data-analysis, and image-processing communities  \cite{SpielmanT96,Meila01arandom,BelkinPartha,CoifmanPNAS,roweis2000ndr,Coifman_Lafon} on the use of such eigenvectors for tasks like clustering, ranking, image partitioning, and data visualization. 

For core--periphery detection, it is useful to consider the bottom eigenvector of the associated random-walk Laplacian. Considering the block model in Eq.~\eqref{nullmodel} or the generalized block model $G(p_{cc},p_{cp},p_{pp},n_c,n_p)$ (see the depiction in Table \ref{tab:generalBlockModel}) with $p_{cc} \approx p_{cp} < p_{pp}$, the task of finding core--periphery structure in a given graph $G$ amounts to trying to detect a dense connected component between the peripheral vertices in the complement graph $\bar{G}$ (in which the $0$ non-diagonal entries of $A$ become $1$, and the $1$ entries become $0$), as such vertices have many non-edges between them in the original graph. If $p_{cc} \approx p_{cp} < p_{pp}$ (i.e., the above scenario) and there exists a single densely-connected component in a given graph --- such as in examples (a) and (b) in Fig.~\ref{fig:LaplacianSmLg} --- the eigenvector that corresponds to the second largest (in magnitude) eigenvalue of the associated random-walk Laplacian provides an accurate separation of the vertices in the dense component from the rest of the graph. The complement of the block-model graph has a periphery component of size $n_p$ that is fully connected (i.e., it is $K_{n_p}$, the complete graph on $n_p$ vertices), a core component without any edges between pairs of core vertices, and no edges between core and peripheral vertices. In practice, $\bar{G}$ is a perturbed version of the above complement block model; that is, the peripheral vertices are very well-connected among themselves, and there are few core--core and core--periphery connections. Our task then amounts to identifying a well-connected ``community'' of peripheral vertices. In other words, we have replaced the problem of identifying a core set and periphery set in $G$ with the problem of finding the periphery set in $\bar{G}$, for which we can use methods from the large set of available techniques for community detection \cite{masonams,fortunato}.

In many applications, the initial graph $G$ is rather sparse, and the above approach thus has the drawback that the complement graph $\bar{G}$ is very dense, which significantly increases the time that is necessary for the computational task of identifying communities \cite{Brandes2008} (though we note that we only seek to identify a single dense subgraph rather than a graph's entire community structure). As we discussed above, one way to find a dense subgraph of an initial graph is to use the first nontrivial eigenvalue (i.e., the second largest eigenvalue) of the random-walk Laplacian. In Fig.~\ref{fig:LaplacianSmLg}(a), we show an example of such a computation. In this case, we start with a block-model graph from $G(p_{cc}=0.8, p_{cp}=0.2, p_{pp}=0.2,n_c,n_p)$, for which the first nontrivial eigenvalue (see the second column) clearly separates the planted dense subgraph from the rest of the network. In the eigenvector computation for the random-walk Laplacian, note that every iteration of the power method is linear in the number of edges in the graph, and the number of iterations is strictly greater than $\mathcal{O}(1)$ because it depends on the spectral gap. For sparse graphs $G$, the complement $\bar{G}$ is a dense graph, which significantly increases the computational effort needed to find eigenvectors. Instead of working in the complement space, we turn our attention to the other end of the spectrum and consider the smallest eigenvalue of the random-walk Laplacian. Recall that all of the eigenvalues of the random-walk Laplacian are less than or equal to $1$ in magnitude \cite{chung}.

We now focus on the combinatorial Laplacian $F = D - A$. Let $\bar{F}$ denote the combinatorial Laplacian associated to the graph $\bar{G}$. Note that $\bar{A} = J_n - A - I_n$, where $J_n$ denotes the matrix of size $n \times n$ whose entries are all $1$ and $I_n$ is the $n \times n$ identity matrix. Additionally, $\bar{D} = (n-1) I_n - D$.  A well-known relationship \cite{chung} between the combinatorial Laplacian of a graph and that of its complement is given by 
\begin{equation}
    \bar{F} =\bar{D} - \bar{A} = (n-1) I_n - D - ( J_n - A - I_n) = n I_n - J_n - F\,.
\label{def:barF}
\end{equation}
If ${\bf x}$ is an eigenvector of $F$ (other than the trivial eigenvector ${\bf \mb{1}}_{n}$) with ${\bf x} \perp {\bf \mb{1}}_{n}$ (which implies that $J {\bf x}={\bf 0}$) and associated eigenvalue $\lambda$, then ${\bf x}$ is also an eigenvector of $\bar{F}$ (with associated eigenvalue $n-\lambda$). A result due to Kelmans \cite{Kelmans1,Kelmans2,Kelmans3} that connects the characteristic polynomial of the 
combinatorial Laplacian matrix of $G$ to that of its complement implies that 
\begin{equation}
	\lambda_j(\bar{F})  =  n - \lambda_{n+2-j}(F) \qquad \text{for all} \quad j \in \{2,\dots,n\}\,.
\label{KelMans}
\end{equation}
Equation \eqref{KelMans} relates the eigenvalues of the combinatorial Laplacian of $G$ to those of its complement $\bar{G}$. In other words, the spectrum exhibits a certain symmetry, and questions regarding $\lambda_{n+2-j}(F)$ of a graph are equivalent to questions about $\lambda_j(\bar{F})$ of its complement. Furthermore, keeping in mind the usefulness of the second largest eigenvector of the combinatorial Laplacian, we stress that questions involving $\lambda_2(\bar{F})$ (i.e., the case $j=2$) are equivalent to questions involving $\lambda_{n}(F)$.

In practice, none of the eigenvectors of the combinatorial Laplacian are able to distinguish a coherent core set and periphery set in a graph (or a single community in the graph's complement). We calculate the top and bottom eigenvectors (and intermediate ones) of the combinatorial Laplacian and find that none of them captures the distinction between core and periphery sets. Instead, we are able to effectively separate core and periphery sets if we use the random-walk Laplacian $L$, but with the goal of identifying a dense subgraph in the complement graph $\bar{G}$. To do this, one would calculate the second eigenvector $\bar{{\bf v}}_2$ of its associated Laplacian $\bar{L}$. However, because graphs are sparse in most applications, considering the complement of a sparse graph leads to a rather dense graph, which could render computations prohibitive for large $n$. Instead, we propose to use the following approach. Motivated by the analogy in the beginning of this section and the interplay between the bottom eigenvalues of a graph and the top eigenvalues (and their associated eigenvectors) of the graph's complement for the combinatorial Laplacians $F$ and $\bar{F}$, we propose to use the bottom eigenvalue (and its associated eigenvector) of the random-walk Laplacian $L$ associated with our initial graph $G$. The downside of working with the random-walk Laplacian $L$ is that (to the best of our knowledge) there does not exist a statement similar to Eq.~$\eqref{KelMans}$ that makes an explicit connection between the random-walk Laplacian eigenvalues of a graph and those of its complement. In Appendix 2, we explain that such a symmetry exists for the random-walk Laplacian only under certain restrictive conditions. When these conditions are not met, we still make an implicit analogy between the random-walk Laplacian eigenvalues of a graph and those of its complement, but we do not know how to characterize this relationship mathematically. 

In Algorithms \ref{algoLap} and \ref{LapSignAlgo}, we summarize the main steps of two viable algorithms for core--periphery detection using the random-walk Laplacian of a graph. The only difference between Algorithms \ref{algoLap} and \ref{LapSignAlgo} is as follows. The former uses the entries of ${\bf v}_n$ (the bottom eigenvector that corresponds to the smallest algebraic\footnote{Because all of the random-walk Laplacian eigenvalues are real and no larger than $1$ in magnitude, the smallest algebraic eigenvalue corresponds to the smallest real eigenvalue.} eigenvalue as an input to the {\sc Find-Cut} algorithm to infer an optimal separation of the vertices into core and periphery sets by maximizing the objective function \eqref{cpobjDens}.  By contrast, in Algorithm \ref{LapSignAlgo}, the same bottom eigenvector ${\bf v}_n$ of the random-walk Laplacian provides an implicit threshold (i.e., the value $0$), and one is able to classify each vertex as part of a core set or a periphery set by considering the sign of each entry. To choose a global sign, we multiple by $-1$ if necessary to maximize the objective function \eqref{cpobjDens} and ensure that the positive entries correspond to core vertices. (If ${\bf v}_n$ is an eigenvector of $L$, then so is $-{\bf v}_n$.)

\begin{algorithm}[h!]
\caption{{\sc Lap-Core}: Detects core--periphery structure in a graph using a core score that is based on the eigenvector corresponding to the smallest nonzero eigenvalue of the associated random-walk graph Laplacian.}
\label{algoLap}
\begin{algorithmic}[1]
\REQUIRE Adjacency matrix $A$ of the simple graph $G=(V,E)$ with $n$ vertices and $m$ edges.
\STATE Compute the random-walk Laplacian $L = D^{-1}A$, where $D$ is a diagonal matrix with elements $D_{ii} = \sum_{j=1}^{n} A_{ij}$ given by the strength (i.e., the sum of weights of the edges incident to the vertex) of vertex $i$ for each $i$.
\STATE Compute $\lambda_n$, which denotes the smallest nonzero eigenvalue of $L$, and its corresponding eigenvector ${\bf v}_n$. The eigenvector components give the {\sc Lap-Core} scores of the vertices.
\STATE If $\beta$ is known, identify the set of core vertices as the top $\beta n$ vertices with the largest {\sc Lap-Core} scores.
\STATE If $\beta$ is unknown, use the vector of {\sc Lap-Core} scores as an input to the {\sc Find-Cut} algorithm.
\end{algorithmic}
\end{algorithm}

\begin{algorithm}[h!]
\caption{{\sc LapSgn-Core}: Detects core--periphery structure in a graph using the signs of the components of the eigenvector corresponding to the smallest nonzero eigenvalue of the associated random-walk graph Laplacian.}
\label{LapSignAlgo}
\begin{algorithmic}[1]
\REQUIRE Adjacency matrix $A$ of the simple graph $G=(V,E)$ with $n$ vertices and $m$ edges.
\STATE Compute the random-walk Laplacian $L = D^{-1}A$.
\STATE Compute $\lambda_n$, which is the smallest nonzero
eigenvalue of $L$, and its corresponding eigenvector ${\bf v}_n$. The eigenvector components give the {\sc Lap-Core} scores of the vertices.
\STATE Set $z_i = \operatorname{sign}({\bf v}_n(i))$ for $i \in \{1,\dots,n\}$. Because the eigenvector ${\bf v}_n$ is determined up to a global sign change, do the following:
\STATE Let vertex $u_i \in V_C$ if $z_i \geq 0$, and otherwise let $u_i \in V_P$. Let $\eta_1$ denote the resulting value of the objective function \eqref{cpobjDens}. 
\STATE Let vertex $u_i \in V_C$ if $z_i \leq 0$, and otherwise let $u_i \in V_P$. Let $\eta_2$ denote the resulting value of the objective function \eqref{cpobjDens}.
\STATE If  $\eta_1 > \eta_2$, let the final solution be $u_i \in V_C$ if $z_i \geq 0$; otherwise, let $u_i \in V_P$.
\STATE If  $\eta_2 > \eta_1$, let the final solution be $u_i \in V_C$ if $z_i \leq 0$; otherwise, let $u_i \in V_P$.
\STATE If $\eta_1 = \eta_2$, there is no clear separation of the network vertices into core and periphery sets.
\end{algorithmic}
\end{algorithm}

To illustrate the above interplay between the top and bottom parts of the spectrum of the random-walk Laplacian matrix, we consider the SBM $G(p_{cc}, p_{cp},p_{pp},n_c,n_p)$, where we fix the core--core interaction probability $p_{cc}=0.8$ and the periphery--periphery interaction probability $p_{pp}=0.3$, but we vary the core--periphery interaction probability $p_{cp} \in [0.3, 0.7]$ in increments of $0.1$. The goal of these numerical experiments, whose results we show in Fig.~\ref{fig:LaplacianSmLg}, is to demonstrate the ability of the bottom eigenvector ${\bf v}_n$ of $L$ to reveal a core--periphery separation when one exists. To help visualize our results, we also employ a two-dimensional representation of the network vertices in which the core vertices (i.e., the vertices in the set $V_C$) are concentrated within a disc centered at the origin and the peripheral vertices (i.e., the vertices in the set $V_P$) lie on a circular ring around the core vertices. In Fig.~\ref{fig:LaplacianSpect}, we plot the spectrum of the random-walk Laplacian associated to each of the $p_{cp}$ values in the above experiment. Note that we disregard the trivial eigenvector ${\bf v}_1 = {\bf \mb{1}}_{n}$ that corresponds to the trivial eigenvalue $ \lambda_1 = 1$ of $L$.

\begin{figure}[h!]
\begin{center}
\subfigure[$p_{cc}=0.8$, $p_{cp}=0.3$, $p_{pp}=0.3$]{
\includegraphics[width=0.21\columnwidth]{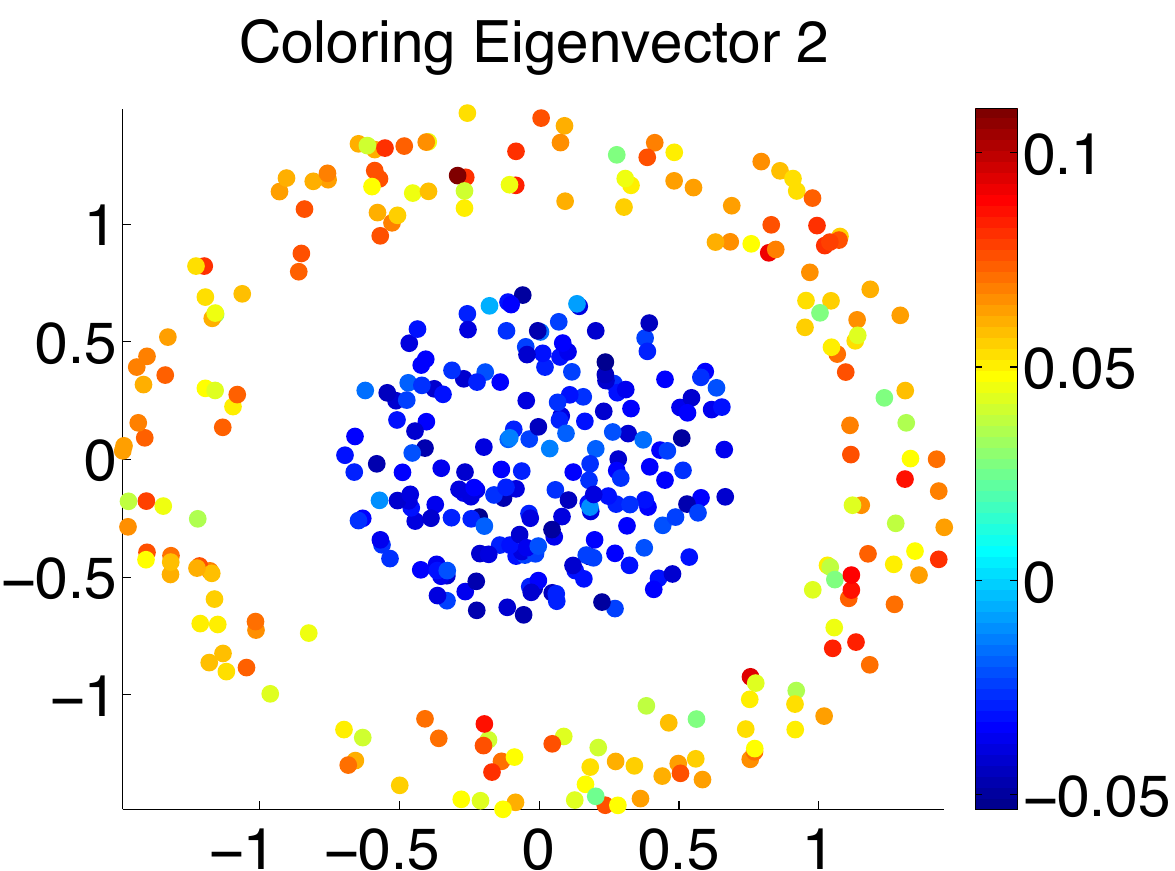}
\includegraphics[width=0.21\columnwidth]{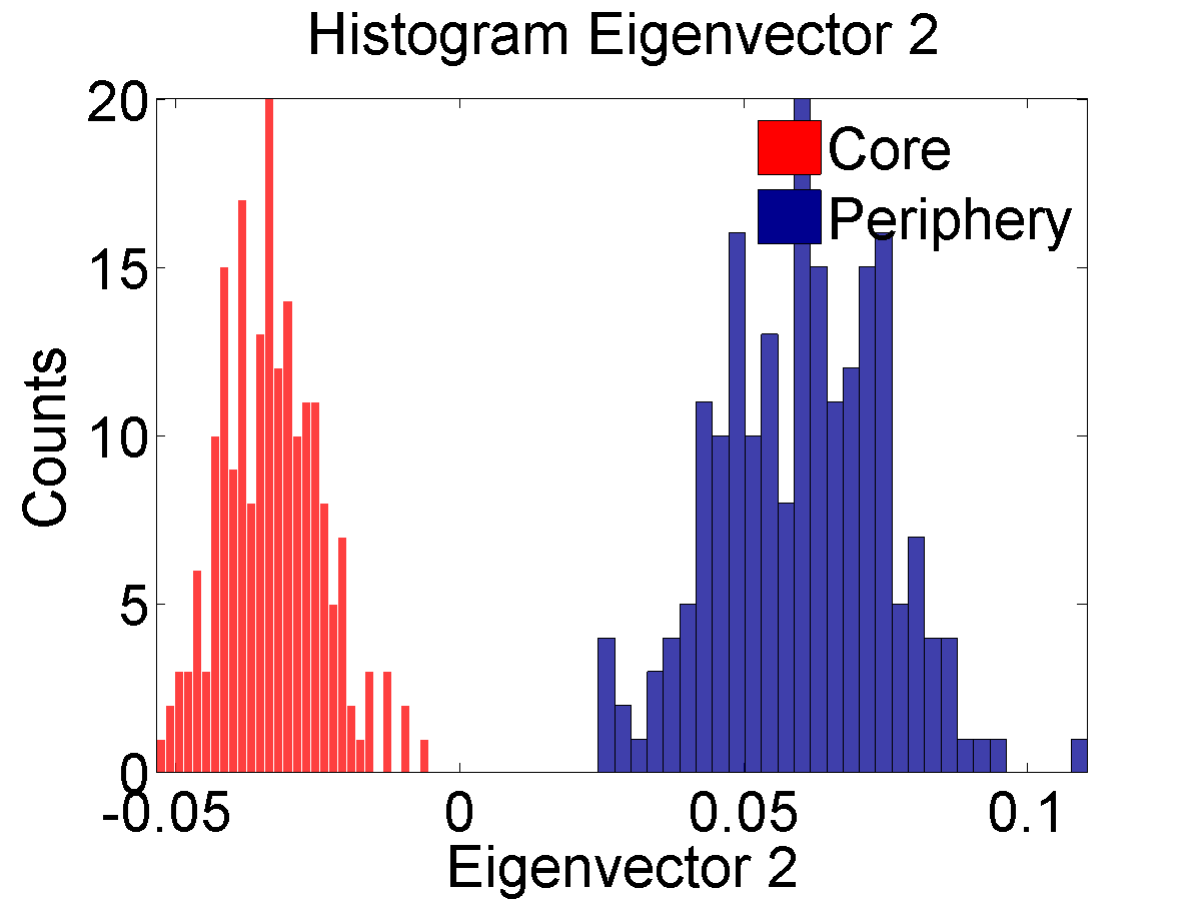}
\includegraphics[width=0.21\columnwidth]{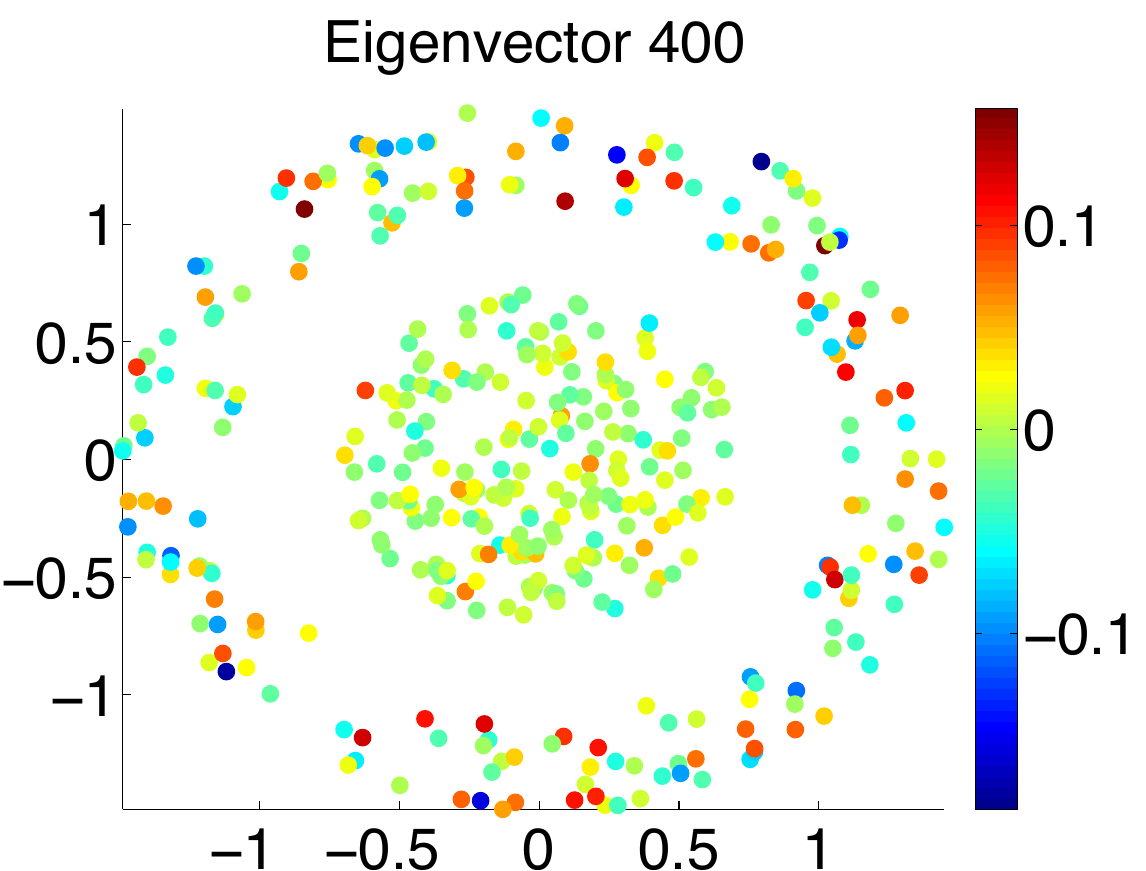}
\includegraphics[width=0.21\columnwidth]{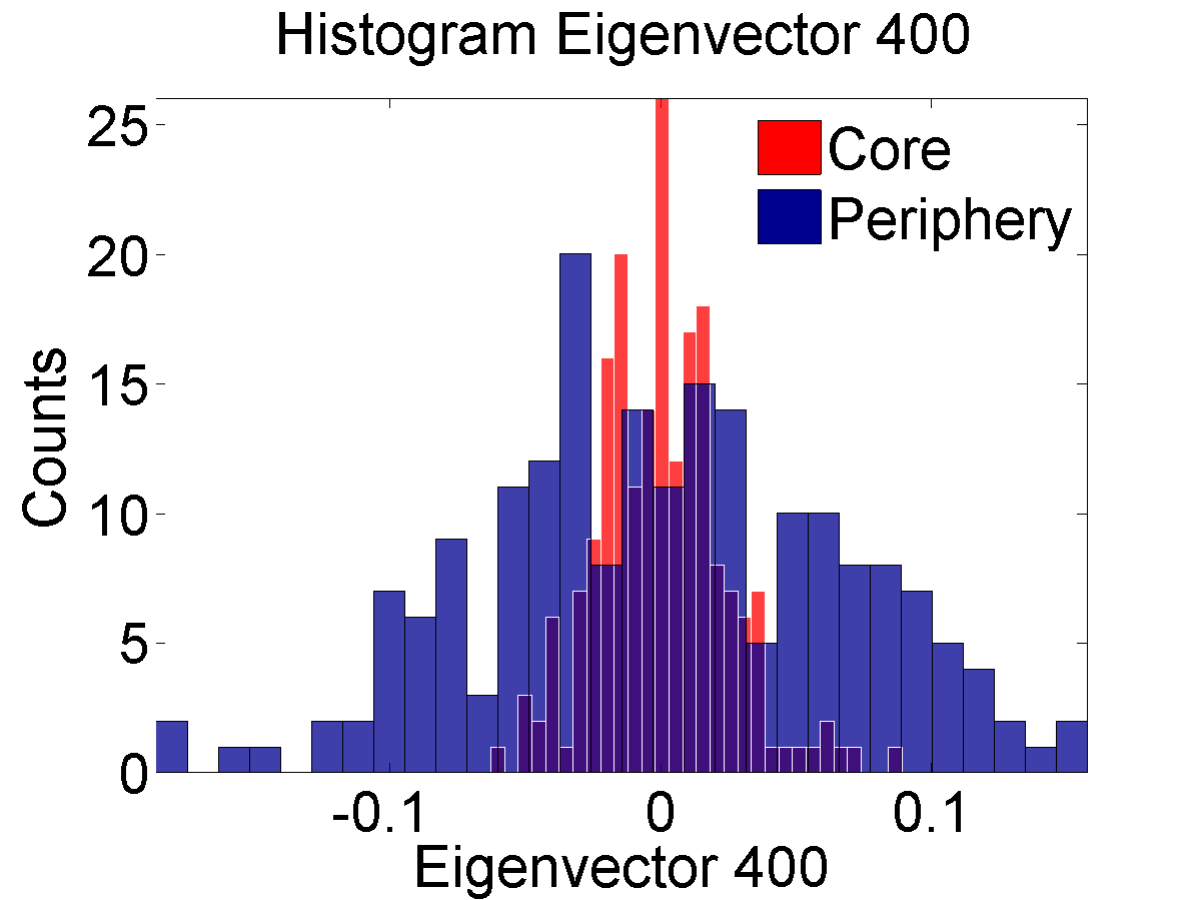}}

\subfigure[$p_{cc}=0.8$, $p_{cp}=0.4$, $p_{pp}=0.3$]{
\includegraphics[width=0.21\columnwidth]{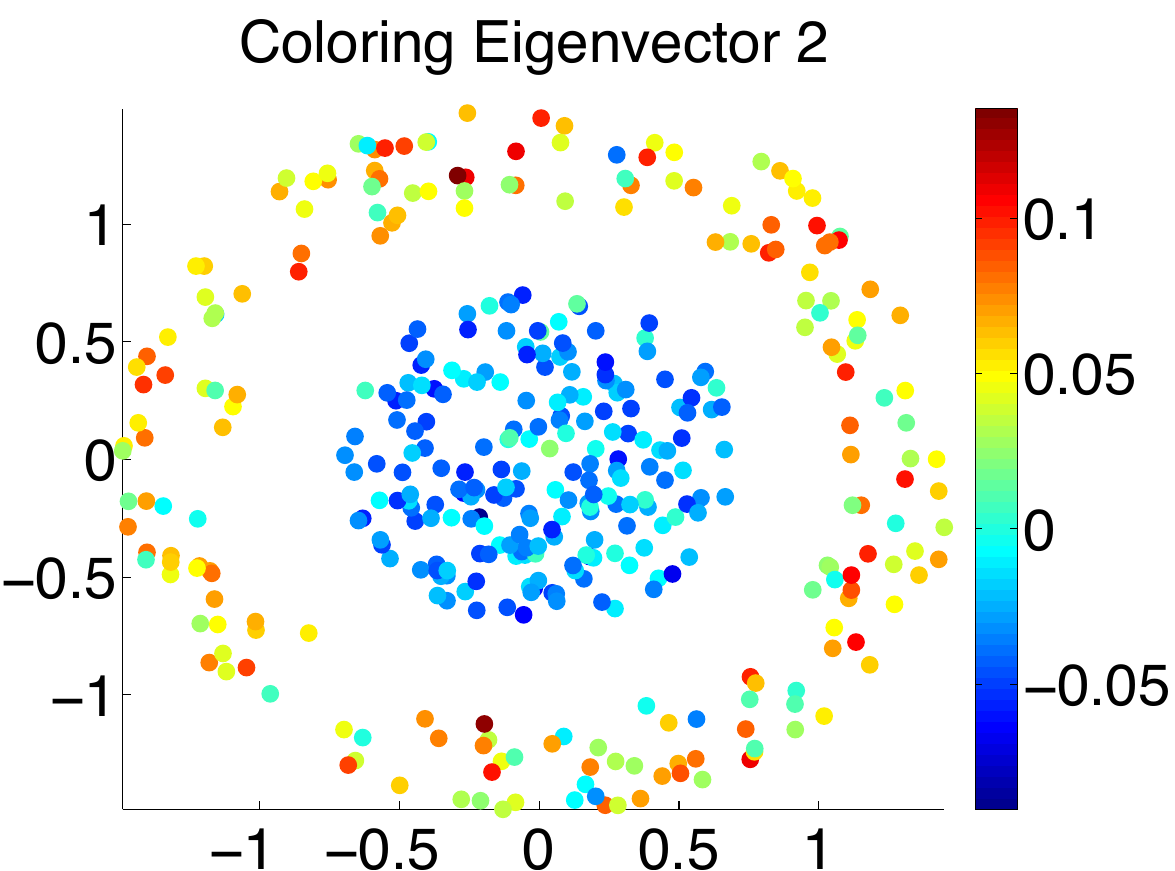}
\includegraphics[width=0.21\columnwidth]{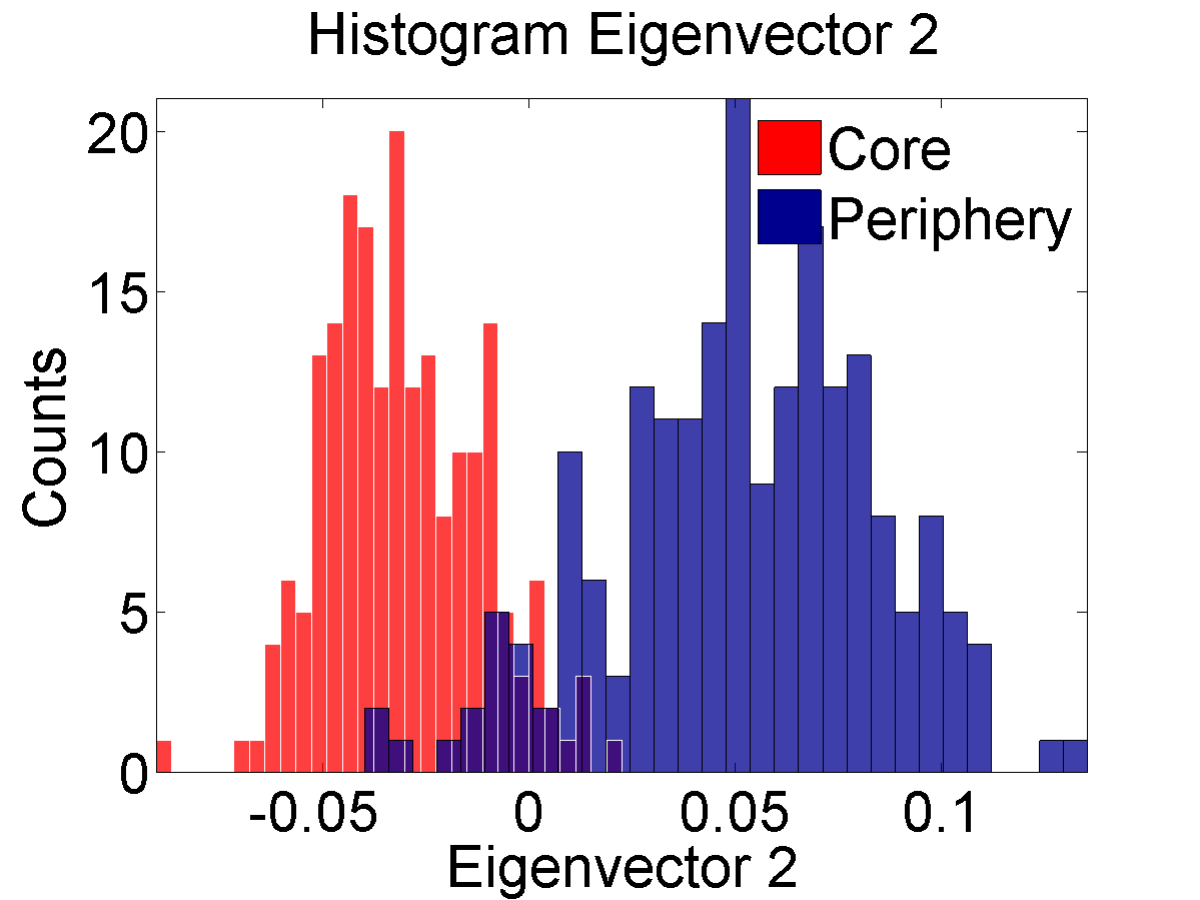}
\includegraphics[width=0.21\columnwidth]{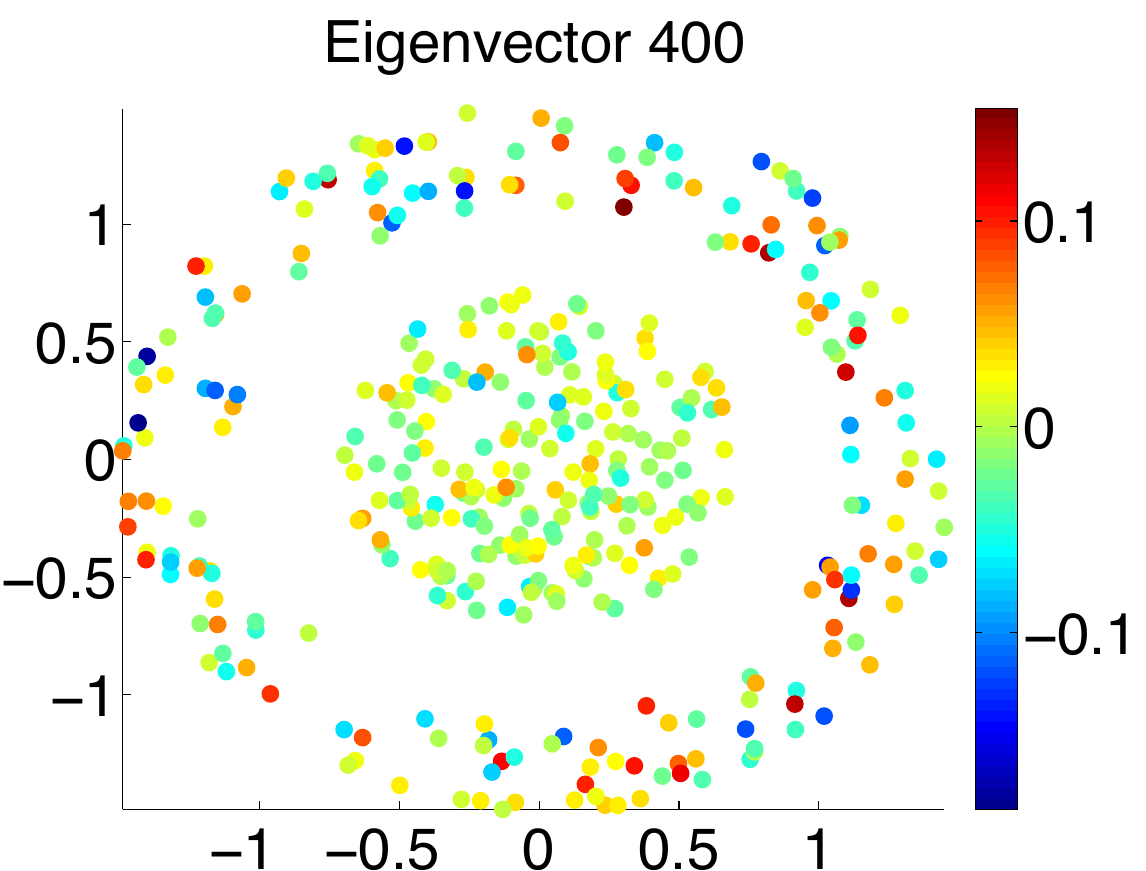}
\includegraphics[width=0.21\columnwidth]{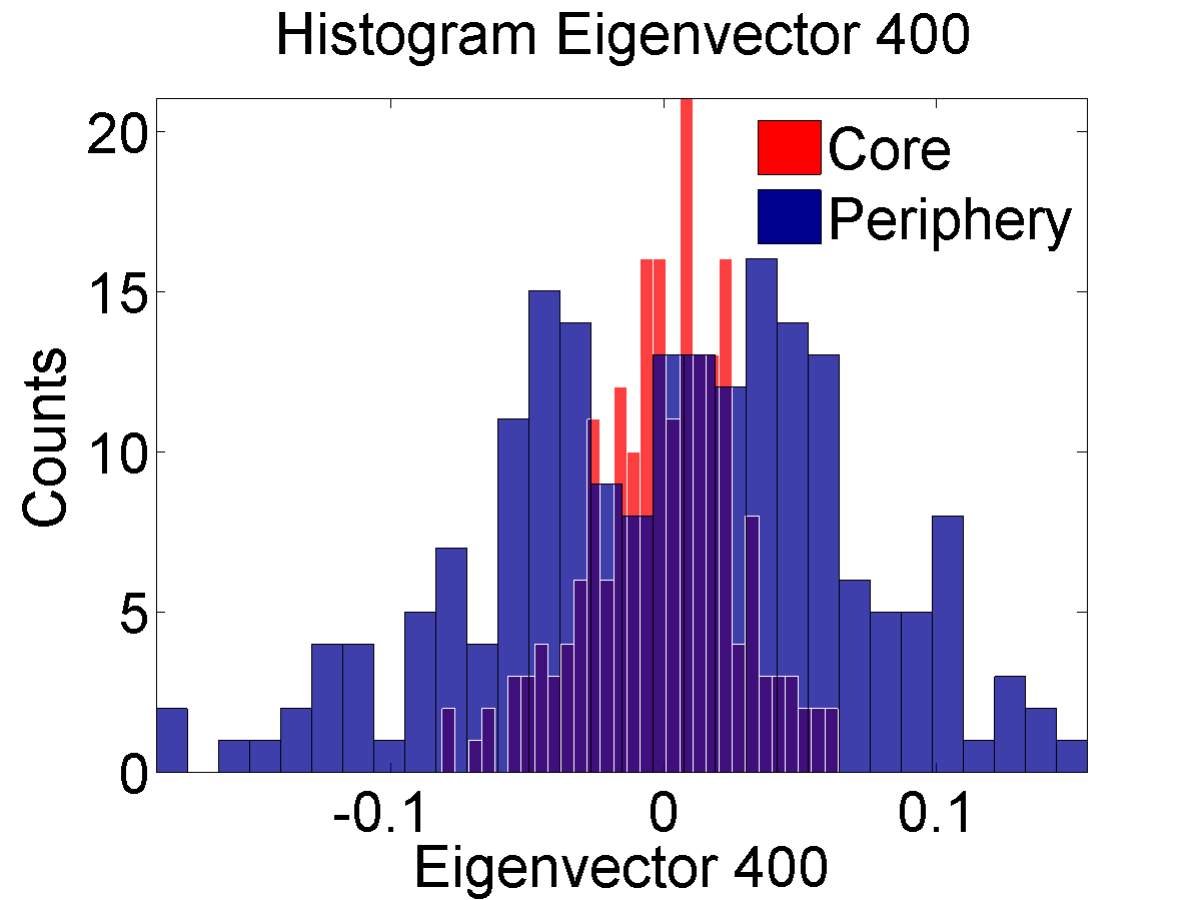}}

\subfigure[$p_{cc}=0.8$, $p_{cp}=0.5$, $p_{pp}=0.3$]{
\includegraphics[width=0.21\columnwidth]{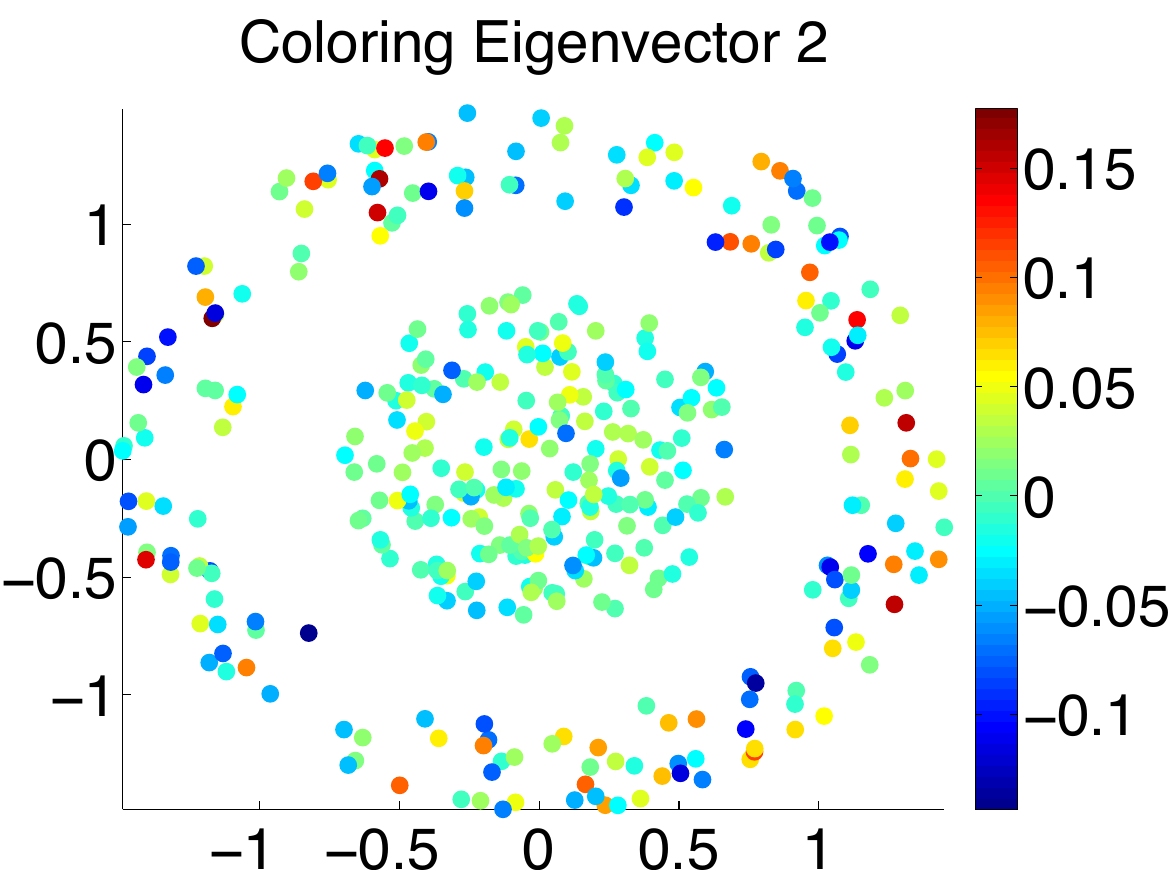}
\includegraphics[width=0.21\columnwidth]{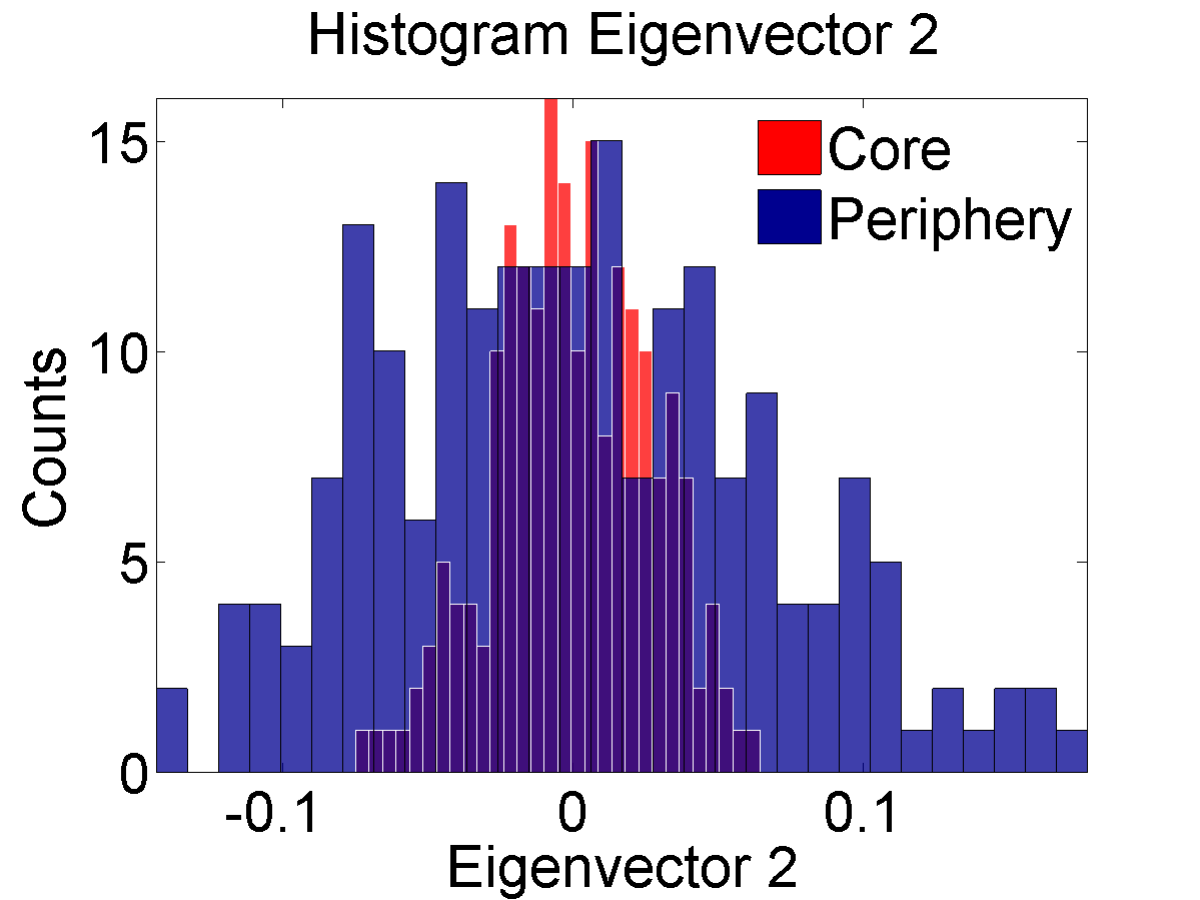}
\includegraphics[width=0.21\columnwidth]{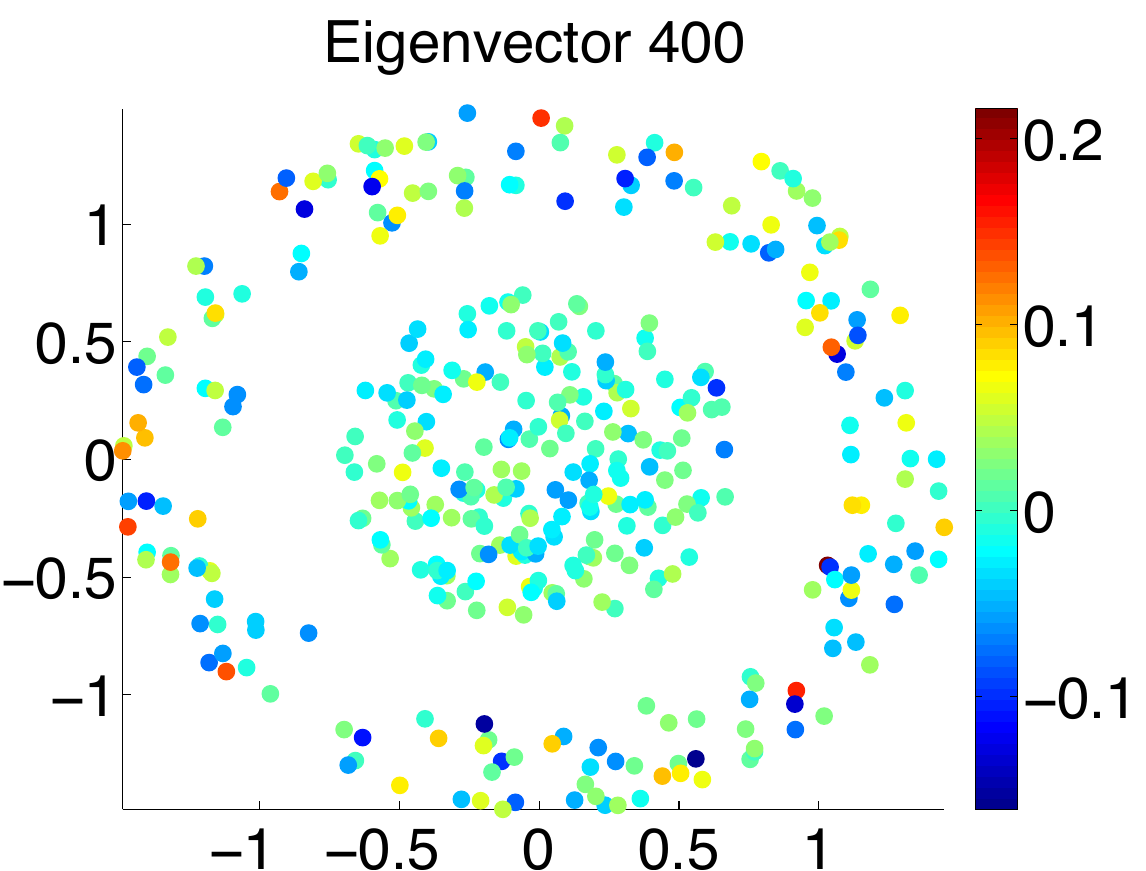}
\includegraphics[width=0.21\columnwidth]{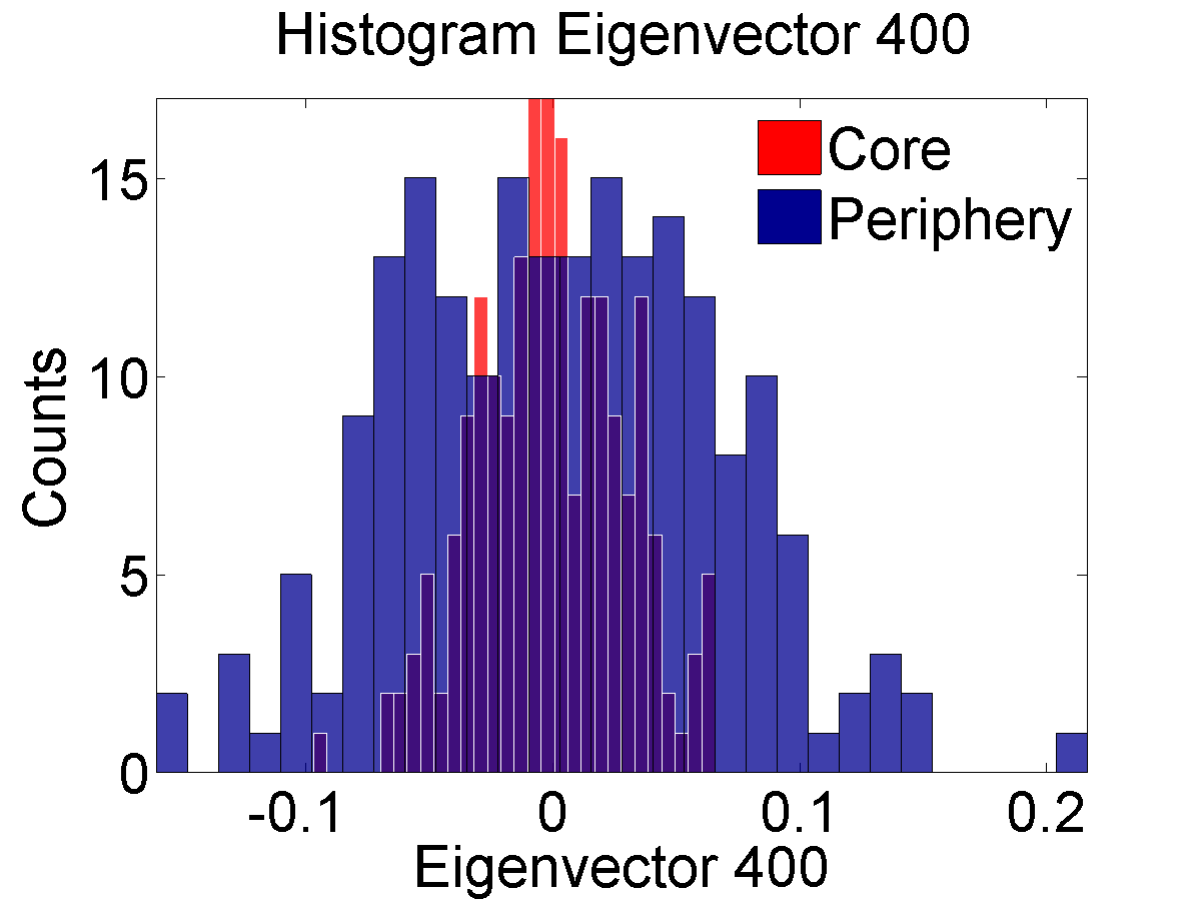}}

\subfigure[$p_{cc}=0.8$, $p_{cp}=0.6$, $p_{pp}=0.3$]{
\includegraphics[width=0.21\columnwidth]{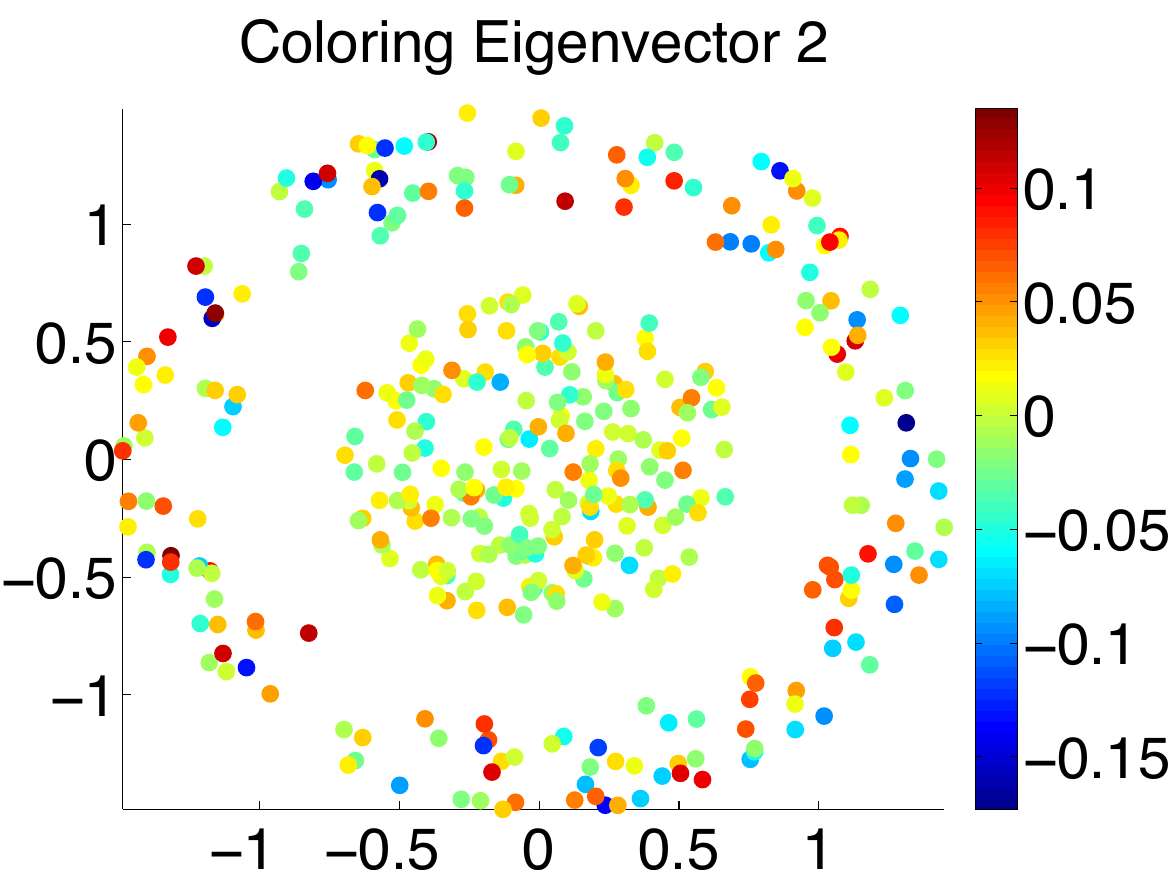}
\includegraphics[width=0.21\columnwidth]{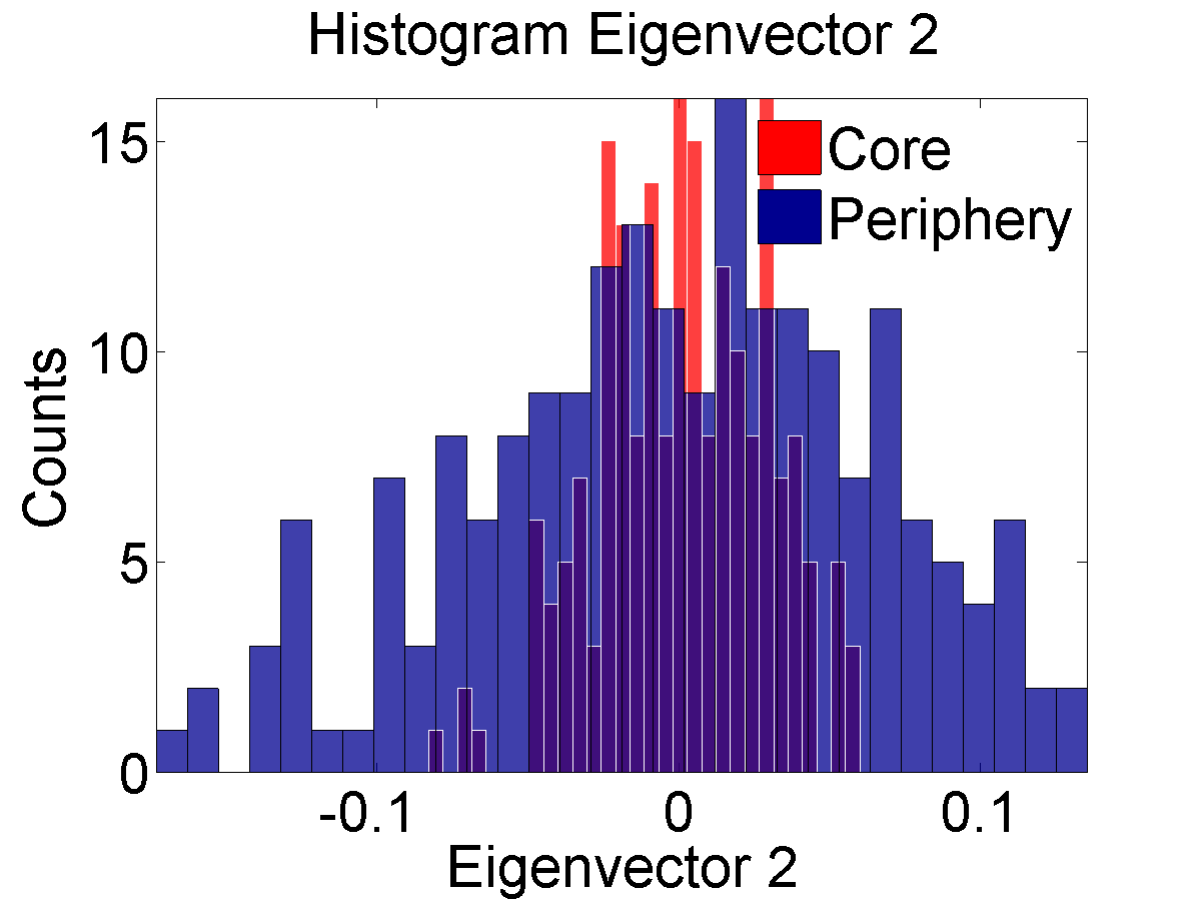}
\includegraphics[width=0.21\columnwidth]{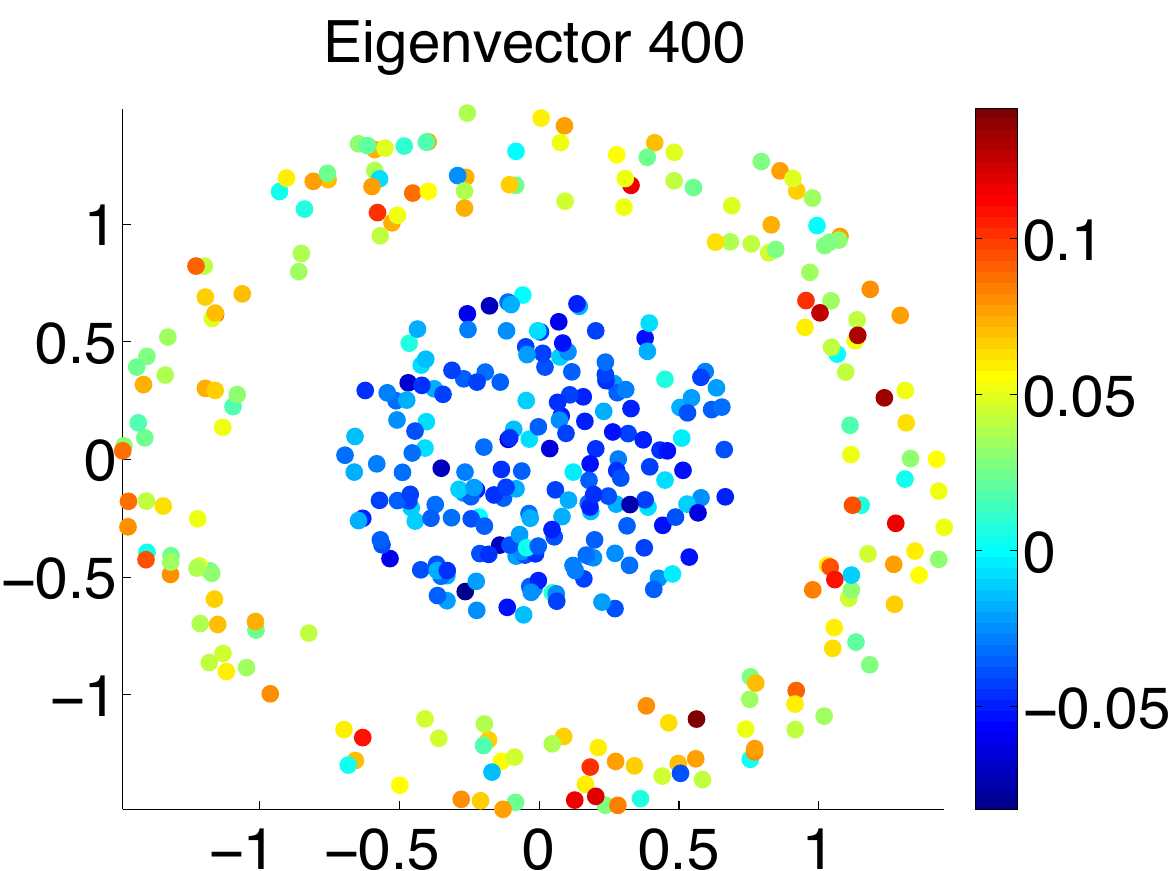}
\includegraphics[width=0.21\columnwidth]{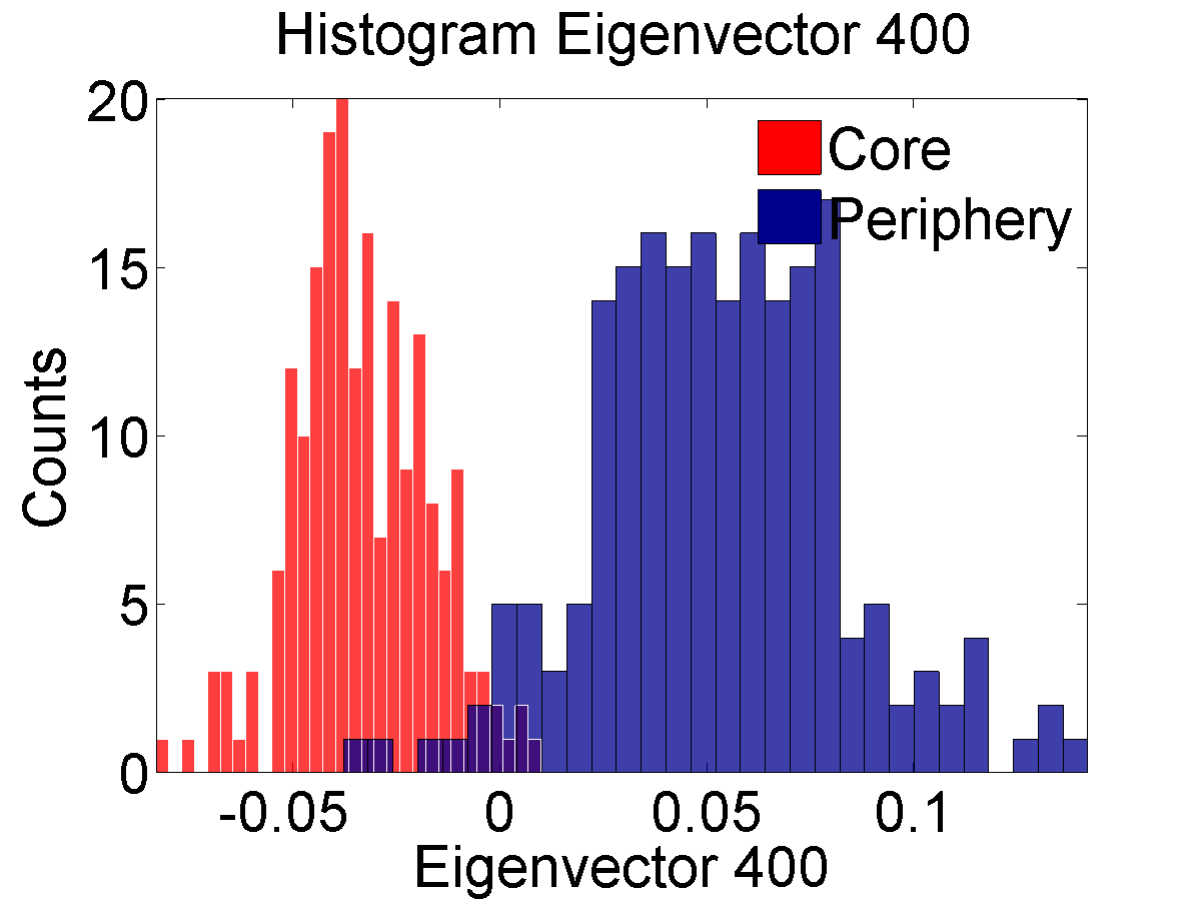}}

\subfigure[$p_{cc}=0.8$, $p_{cp}=0.7$, $p_{pp}=0.3$]{
\includegraphics[width=0.21\columnwidth]{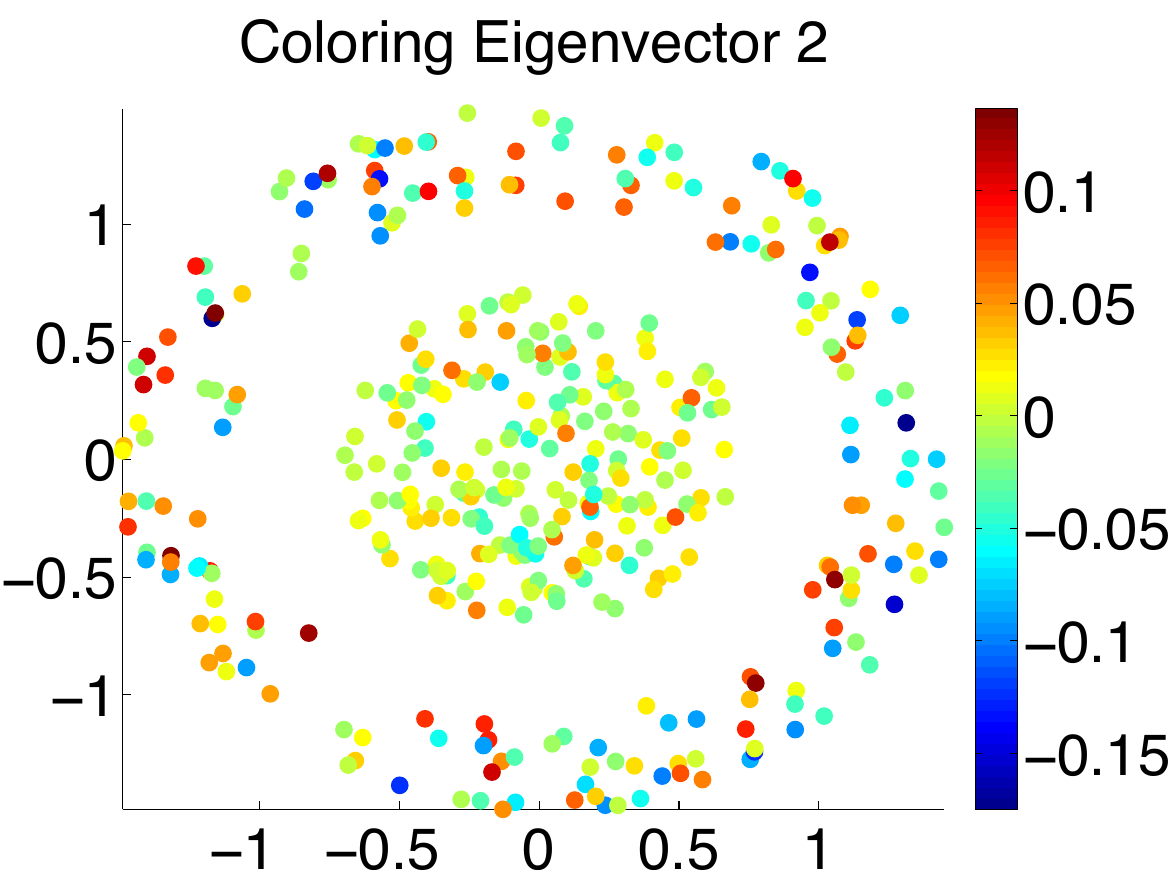}
\includegraphics[width=0.21\columnwidth]{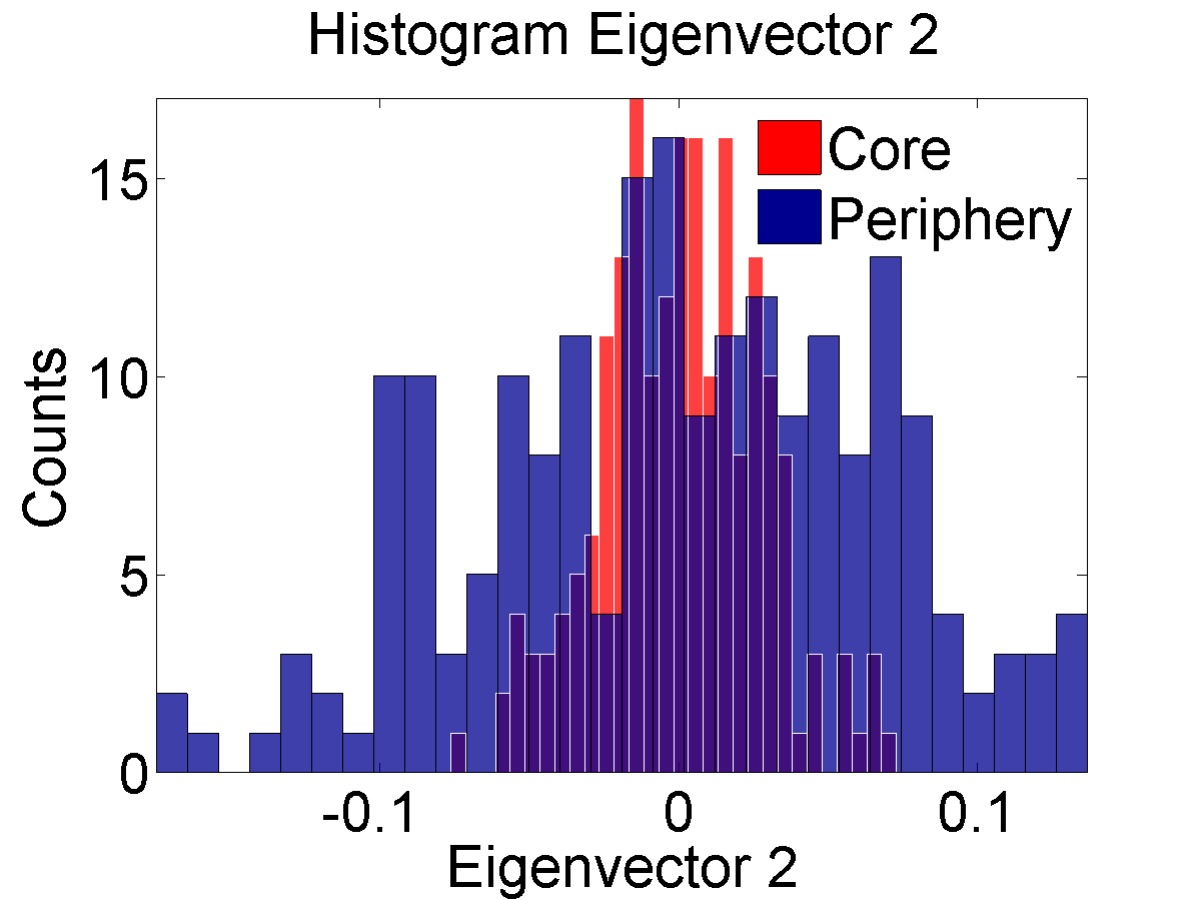}
\includegraphics[width=0.21\columnwidth]{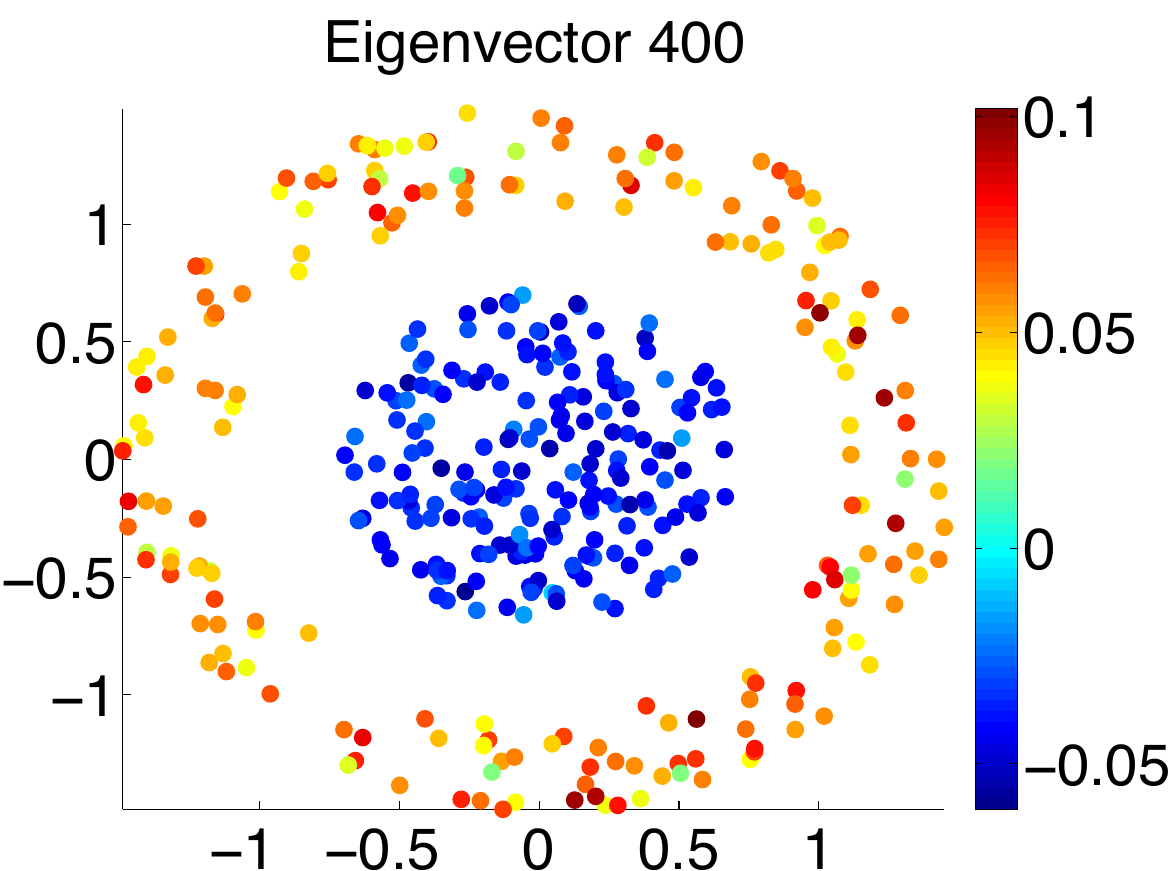}
\includegraphics[width=0.21\columnwidth]{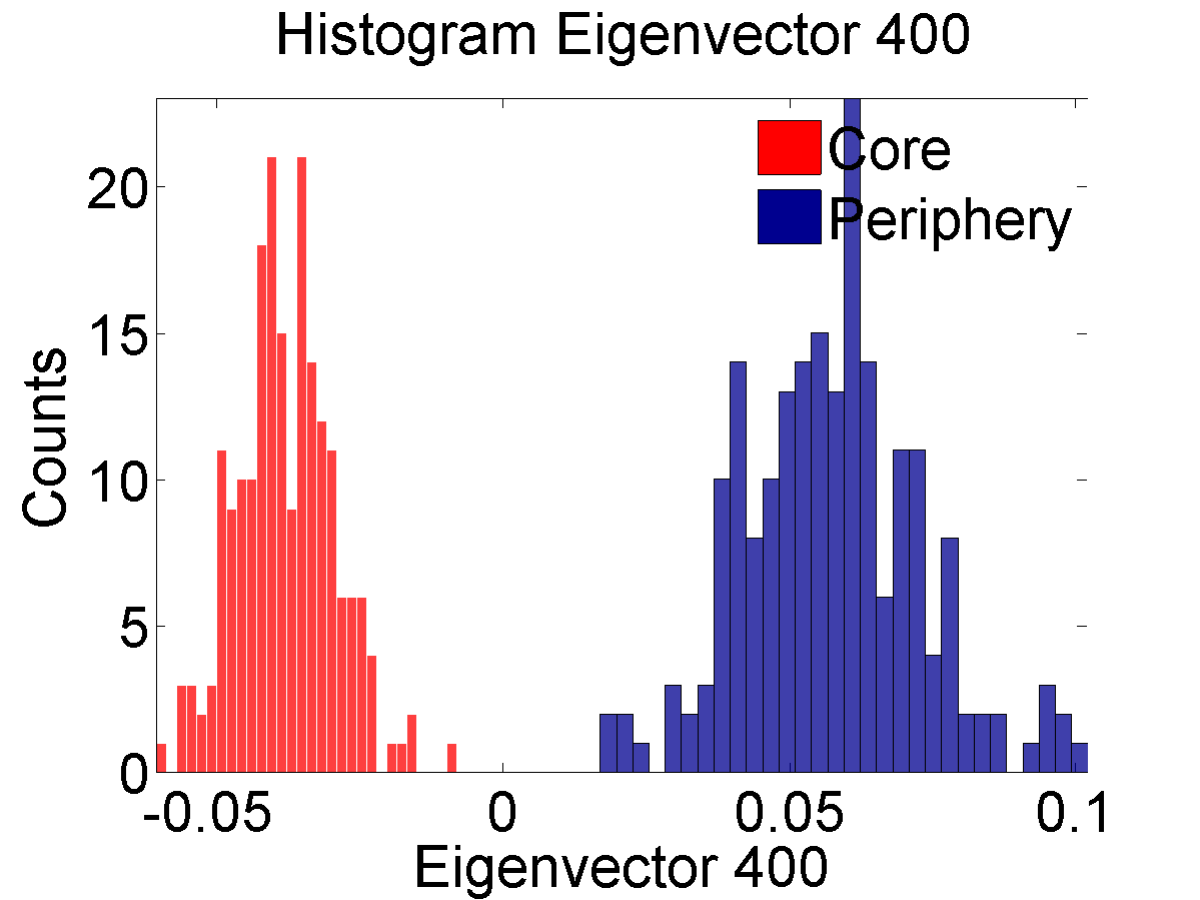}}

\end{center}
\caption{[Color] Our simulations illustrate the interplay between the top and bottom parts of the spectrum of the random-walk Laplacian matrix $L$ as a network transitions from a block model with block-diagonal ``community structure'' to a block model with core--periphery structure. Each row uses one network from the SBM $G(p_{cc}, p_{cp},p_{pp},n_c,n_p)$ with $n=400$ vertices (with 200 core and 200 peripheral vertices) with a fixed core--core interaction probability $p_{cc}=0.8$, a fixed periphery--periphery interaction probability $p_{pp}=0.3$, and a varying core--periphery interaction probability $p_{cp} \in [0.3, 0.7]$. We vary $p_{cp}$ in increments of $0.1$, so the top row has $p_{cp} = 0.3$, the second row has $p_{pp} = 0.4$, and so on.  The first and third columns give a coloring of a two-dimensional visualization of the graph vertices: the core vertices are contained in a disc that is centered at the origin, and the peripheral vertices lie on a ring around the core vertices. The second and fourth columns, respectively, show histograms of the entries of the eigenvectors ${\bf v}_2$ and ${\bf v}_{400}$. These eigenvectors correspond, respectively, to the largest (nontrivial) and smallest eigenvalues of the associated random-walk Laplacian matrix. The red color indicates core vertices, and the blue color indicates peripheral vertices. In Fig.~\ref{fig:LaplacianSpect}, we plot the spectrum associated to each of the above six networks. 
}
\label{fig:LaplacianSmLg}
\end{figure}

For small values of $p_{cp}$ (e.g., $p_{cp} = 0.3$ or $p_{cp} = 0.4$), the network does not exhibit core--periphery structure. Instead, it has a single community that is represented by the densely connected graph of vertices in the set $V_C$. As expected, the eigenvector ${\bf v}_2$ is able to highlight the separation between the $V_C$ and $V_P$ vertices very well, whereas the bottom eigenvector ${\bf v}_n$ is not particularly helpful. For $p_{cp}=0.5$, neither of the two eigenvectors above are able to capture the separation between $V_C$ and $V_P$. However, as $p_{cp}$ increases to $p_{cp}= 0.6$ and $p_{cp} = 0.7$ --- such that we are closer to the idealized block model in (\ref{nullmodel}) --- there now exists a densely-connected subgraph of $V_P$ in the complement graph $\bar{G}$. Instead of using the top nontrivial eigenvector $\bar{{\bf v}}_2$ of $\bar{L}$, we use the eigenvector ${\bf v}_n$ that corresponds to the smallest eigenvalue $\lambda_n$ of $G$, as this eigenvector is able to highlight core--periphery structure in $G$. In Fig.~\ref{fig:LaplacianSpect}(a), we show that there is a clear separation between $\lambda_2$ and the bulk of the spectrum. Similarly, Fig.~\ref{fig:LaplacianSpect}(e) illustrates a clear separation between $\lambda_n$ and the bulk of the spectrum. For intermediate values of $p_{cp}$, such a spectral gap is significantly smaller or even nonexistent.

\begin{figure}[h!]
\begin{center}
\subfigure[$p_{cc}=0.8$, $p_{cp} = 0.3$, $p_{pp}=0.3$]{
\includegraphics[width=0.40\columnwidth]{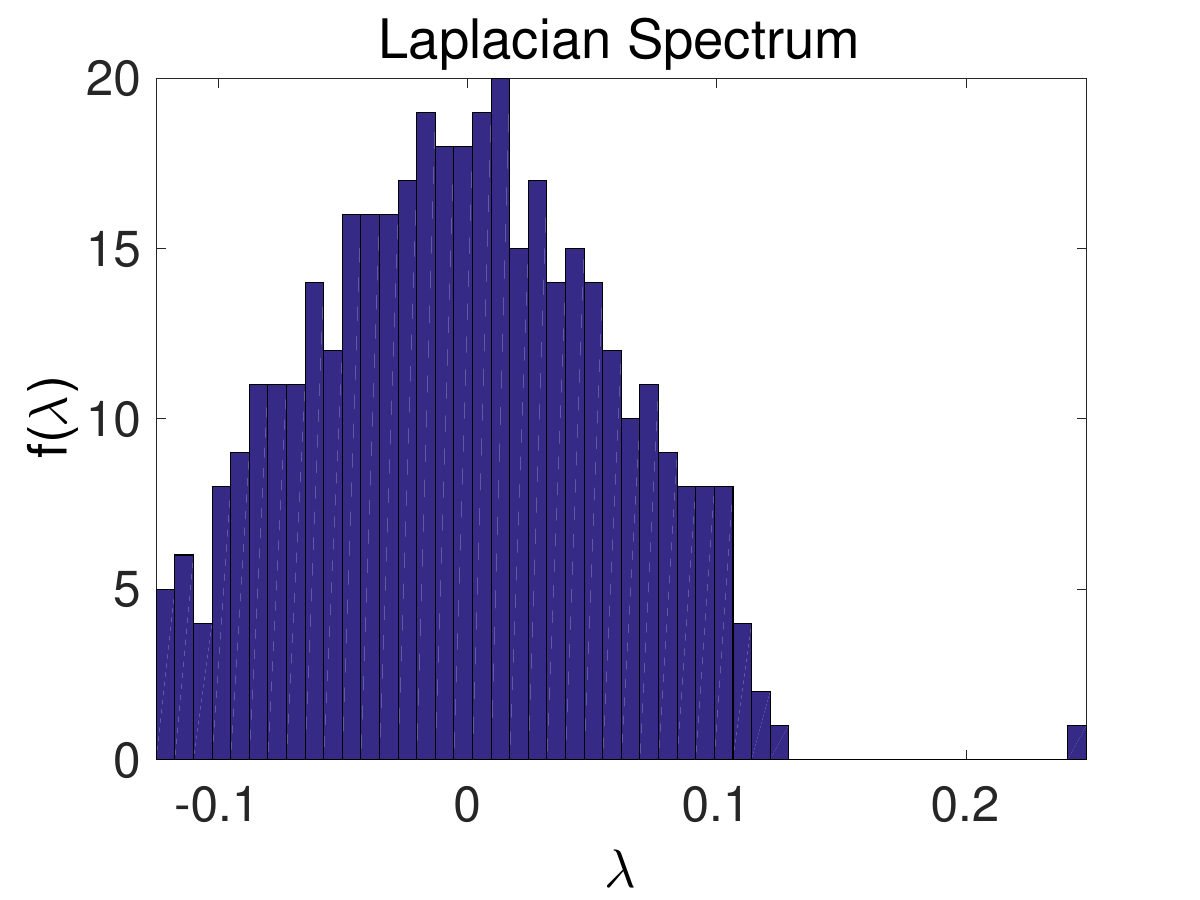}}\hspace{1cm}
\subfigure[$p_{cc}=0.8$, $p_{cp} = 0.4$, $p_{pp}=0.3$]{
\includegraphics[width=0.40\columnwidth]{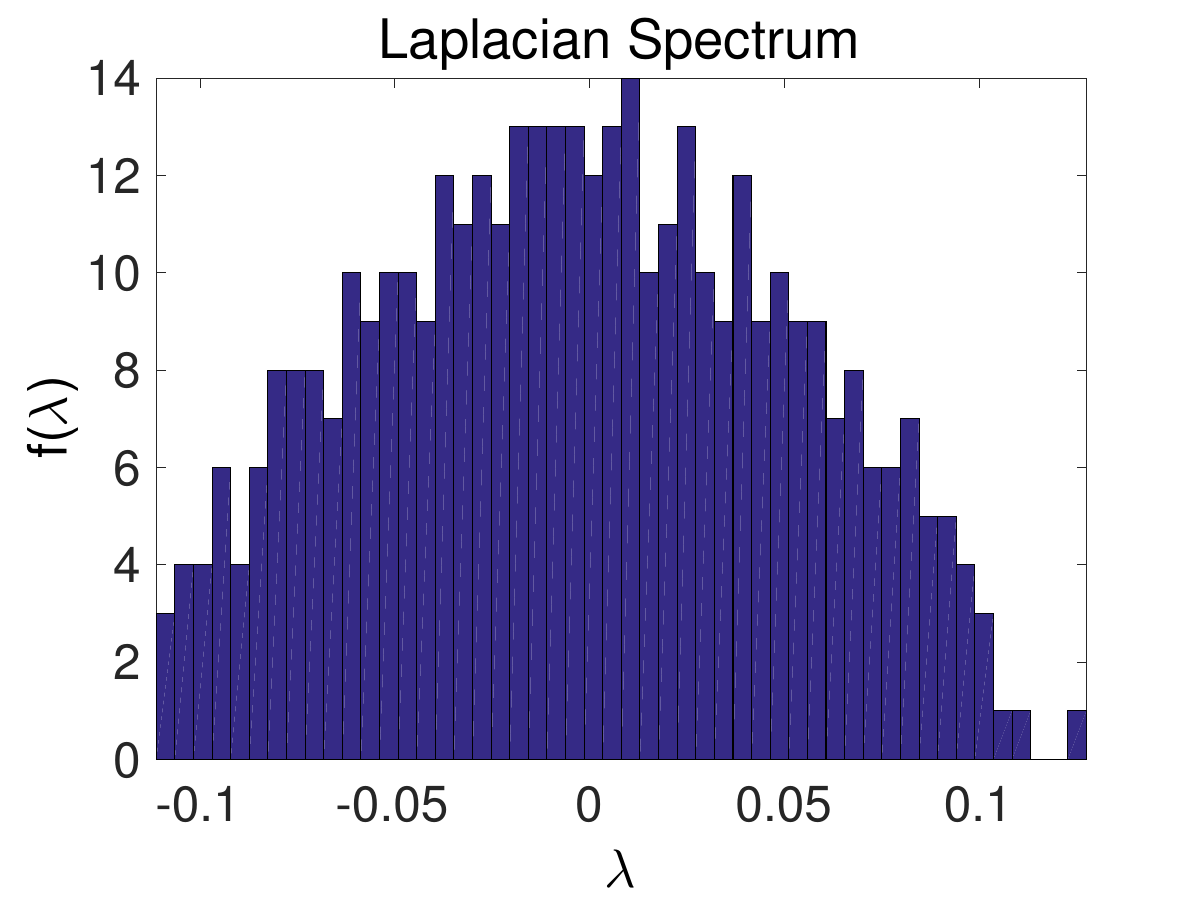}}
\subfigure[$p_{cc}=0.8$, $p_{cp} = 0.5$, $p_{pp}=0.3$]{
\includegraphics[width=0.40\columnwidth]{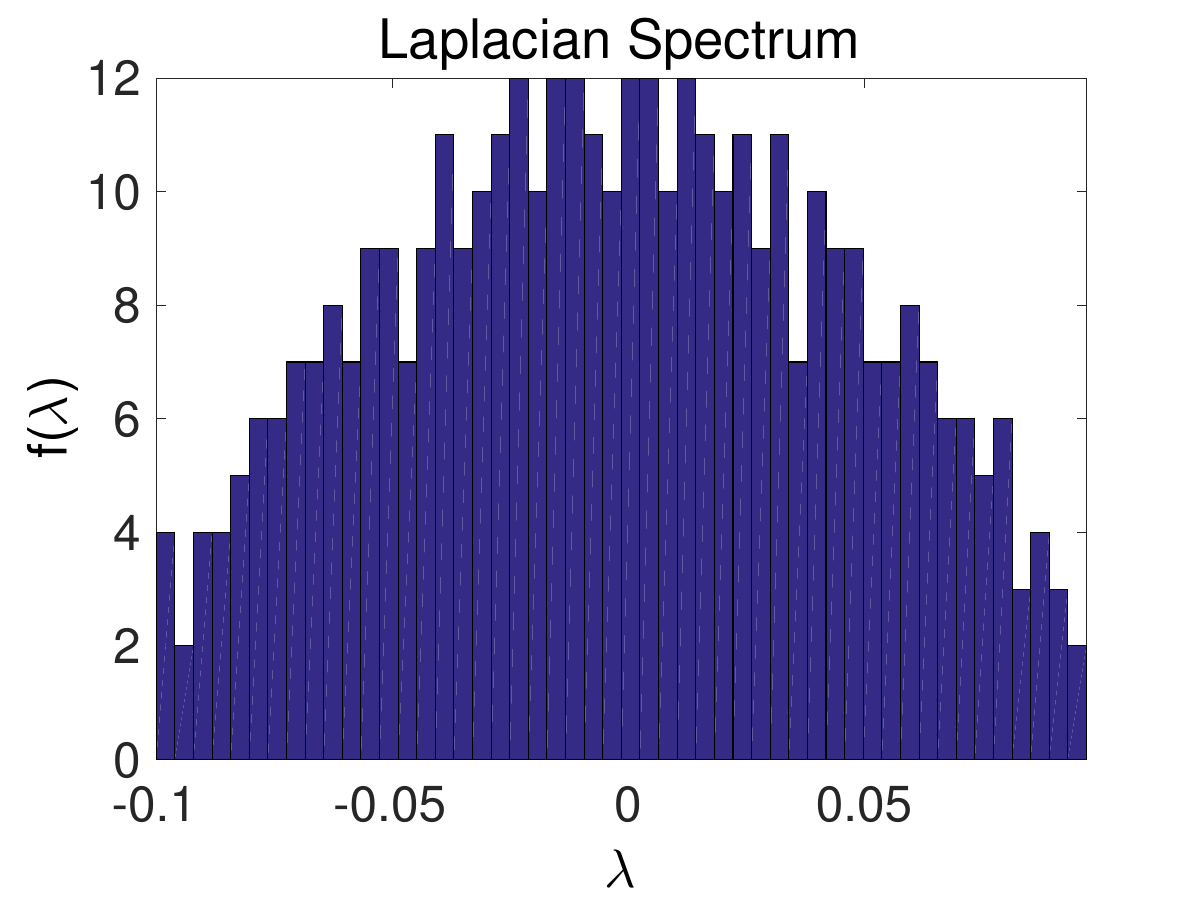}}\hspace{1cm}
\subfigure[$p_{cc}=0.8$, $p_6 = 0.6$, $p_{pp}=0.3$]{
\includegraphics[width=0.40\columnwidth]{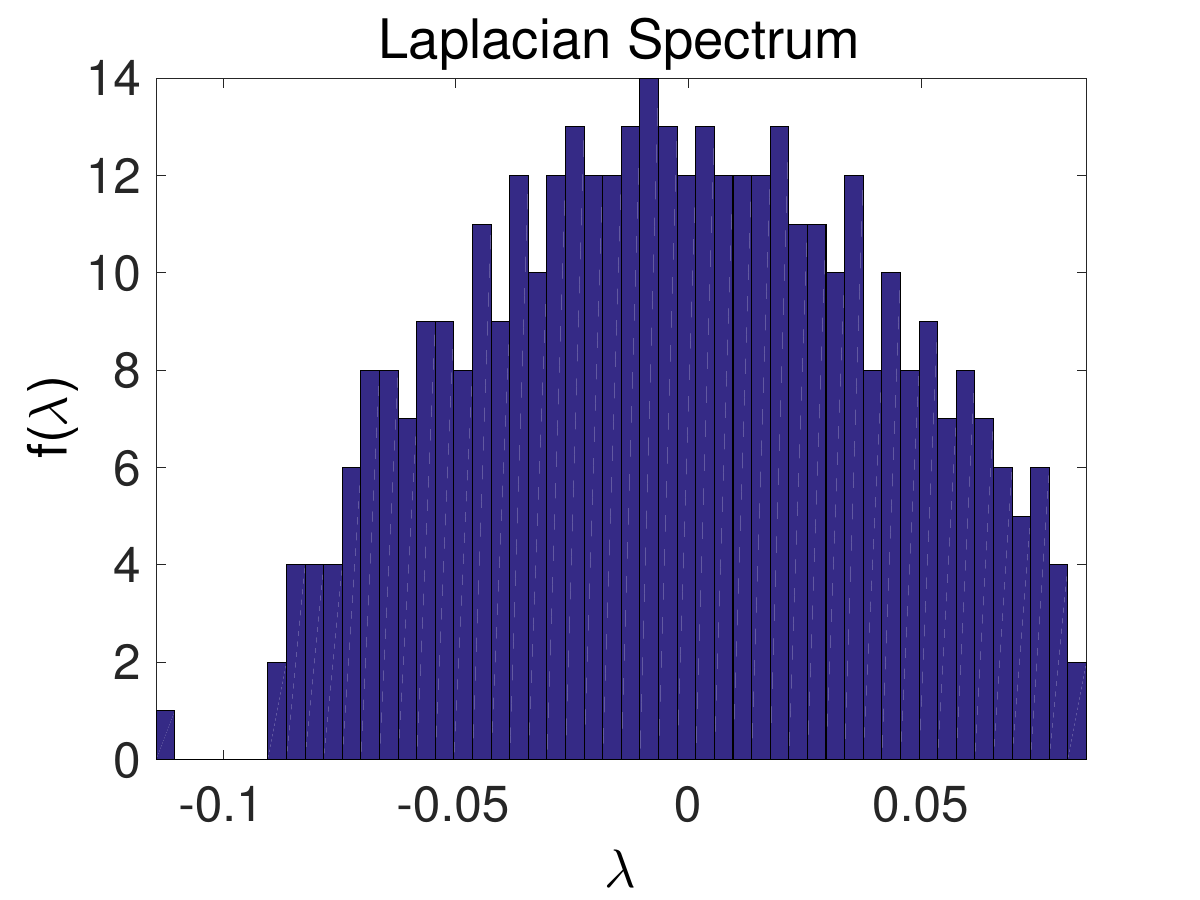}}\hspace{1cm}
\subfigure[$p_{cc}=0.8$, $p_6 = 0.7$, $p_{pp}=0.3$]{
\includegraphics[width=0.40\columnwidth]{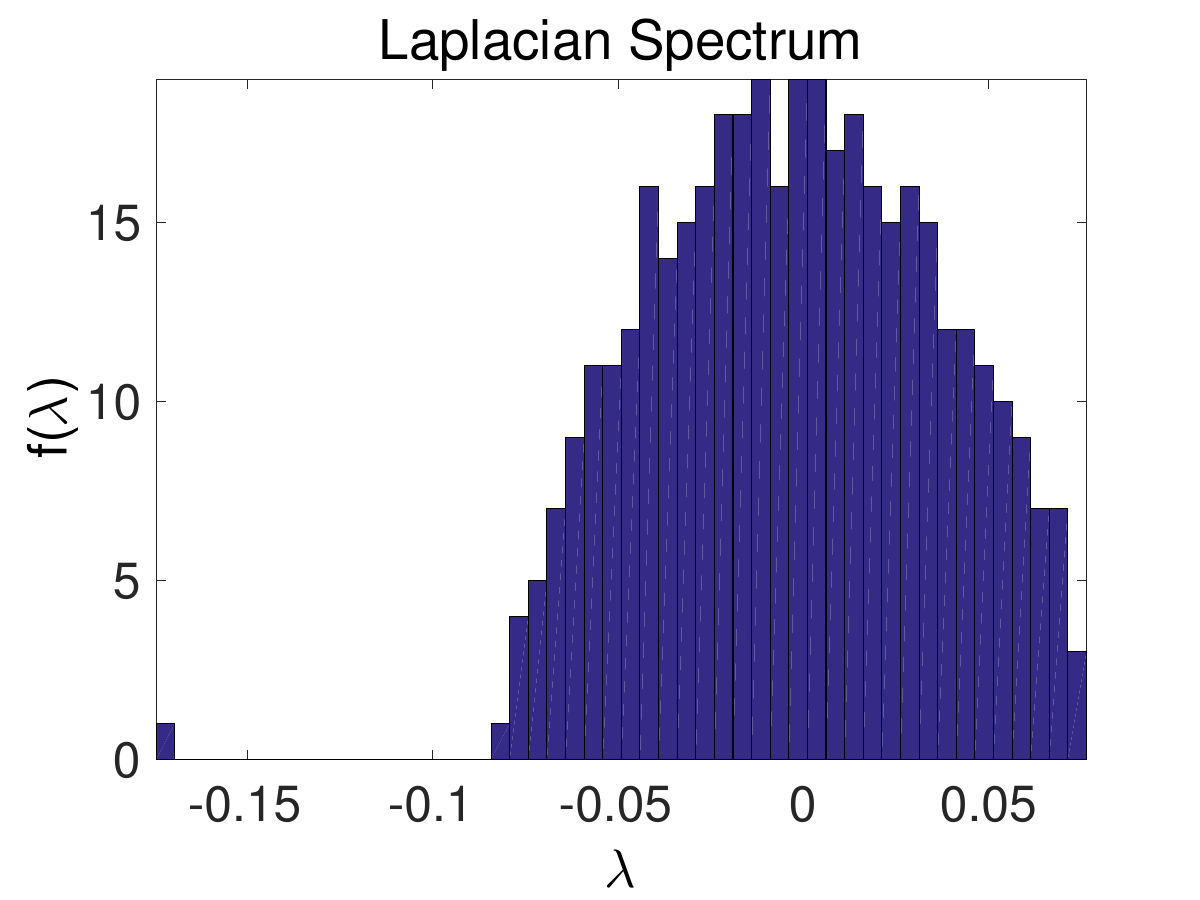}
}
\end{center}
\caption{Spectra $f(\lambda)$ of random-walk Laplacian matrices' eigenvalues $\lambda$ for several instances of the SBM $G(p_{cc},p_{cp},p_{pp},n_c,n_p)$. In Fig.~\ref{fig:LaplacianSmLg}, we plotted histograms of the eigenvectors corresponding to the smallest and second largest eigenvalues for these matrices.
} 
\label{fig:LaplacianSpect}
\end{figure}

In conclusion, for core--periphery detection, one should consider the eigenvector ${\bf v}_n$ as in Algorithm \ref{LapSignAlgo}, whereas one should use the eigenvector ${\bf v}_2$ when trying to detect a single dense community. As illustrated in Figure \ref{fig:LaplacianSpect}, one can also use the spectrum of the random-walk Laplacian as guidance. The former scenario is hinted by the presence of a spectral gap to the left of the bulk of the distribution, and the latter scenario is hinted by a spectral gap to the right of the bulk of the distribution.


\section{Numerical Experiments}  \label{sec:numSims}

In this section, we conduct a series of numerical experiments to compare different methods for detecting core--periphery structure and to assess the robustness of our methods to perturbations of a network. In Section \ref{sec:family_of_synthetic_networks}, we examine synthetic networks with a global community structure and local core--periphery structure. In Section \ref{sec:real_networks}, we apply our various methods for detecting core--periphery structure to several empirical data sets. In Appendix 3, we examine networks with ``planted'' high-degree vertices in the periphery, motivated by the recent work of \cite{XZhang2014} that demonstrated that degree-based separation is suboptimal for certain types of networks (in particular, ones with either a very weak or very strong core--periphery structure. Throughout this section, we use the term {\sc Degree-Core} to refer to the method of detecting core--periphery structure by simply computing the vertex degrees and then applying the {\sc FIND-CUT} method. In doing so, we assume that we have knowledge of the ``boundary'' sizes and thereby assume that there is a lower bound on the sizes of the core and periphery sets.

\begin{figure}[h!]
\begin{center}
\includegraphics[width=0.38\textwidth]{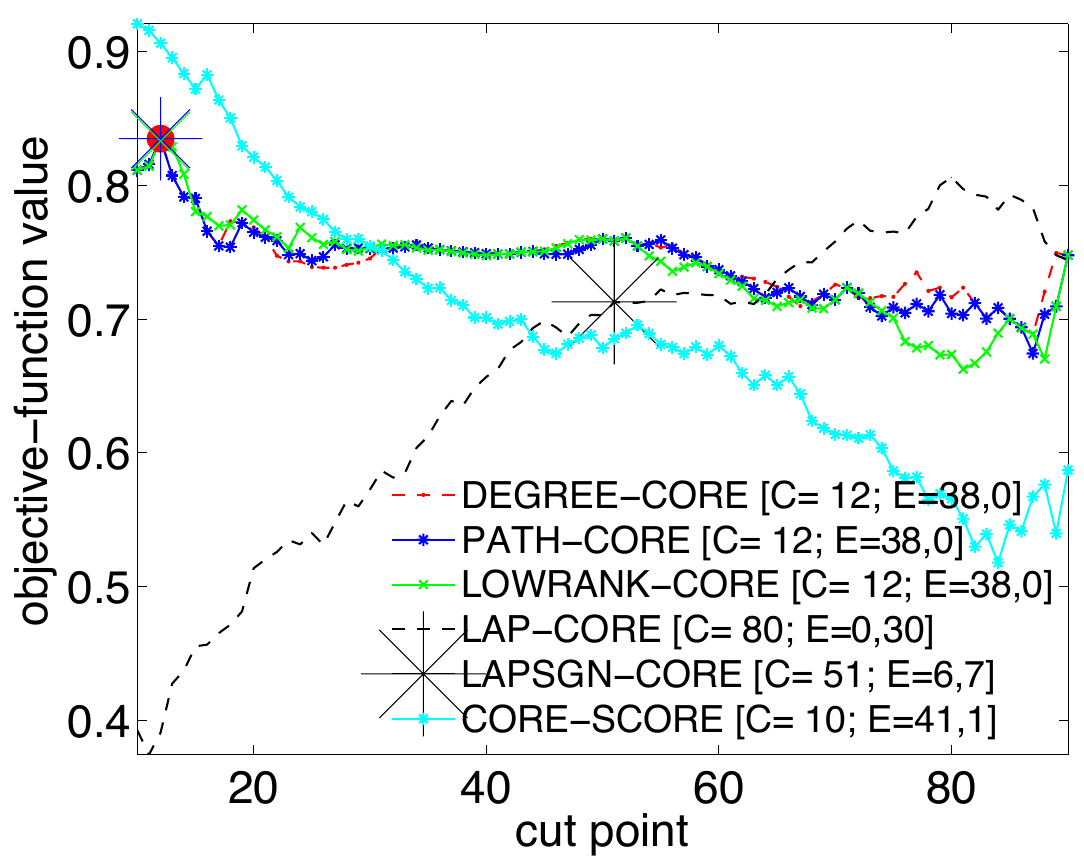}
\includegraphics[width=0.38\textwidth]{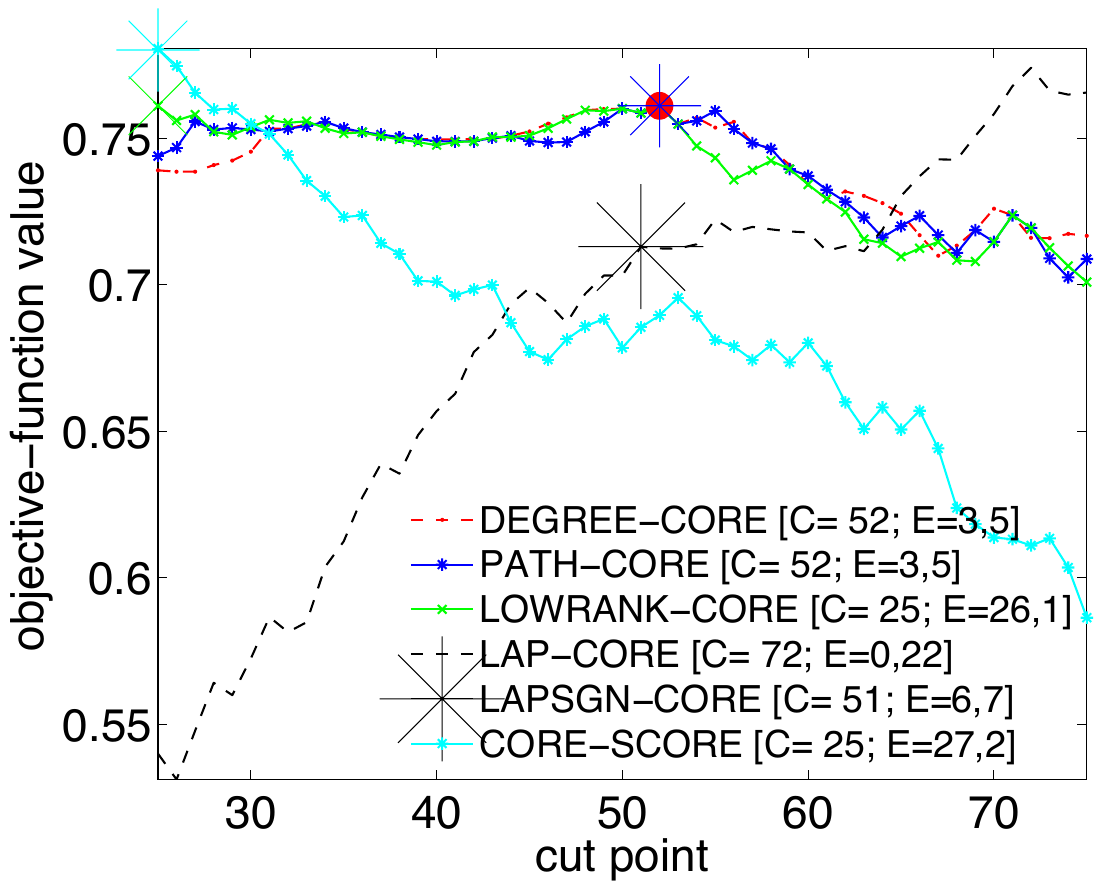}
\end{center}
\caption{Comparison of methods for detecting core--periphery structure for a graph from the ensemble $G(p_{cc},p_{cp},p_{pp},n_c,n_p)$ with $n = 100$ vertices (and, in particular, $n_c = 50$ core vertices and $n_p = 50$ peripheral vertices) and edge probabilities $(p_{cc}, p_{cp}, p_{pp}) = (0.5,0.5,0.27)$ for the objective function in Eq.~\eqref{eq:FindCut}. We assume a minimum size for the core and periphery sets of at least (left) 10 vertices and (right) 25 vertices. We mark the cut points that maximize the objective function in Eq.~\eqref{eq:FindCut} on the curves with a large asterisk for {\sc LapSgn-Core} and using other symbols whose colors match the colors of the corresponding curves for the other methods. The cut point refers to the number of core vertices. In the legends, ${\tt C}$ denotes the size of the core set that maximizes the objective function Eq.~\eqref{eq:FindCut}, and ${\bf E} = (y_1,y_2)$ denotes the corresponding $2$-vector of errors. The first component of ${\bf E}$ indicates the number of core vertices that we label as peripheral vertices, and the second indicates the number of peripheral vertices that we label as core vertices.
}
\label{fig:ex100_5_5_27}
\end{figure}

As we illustrate in Fig.~\ref{fig:ex100_5_5_27}, the {\sc LapSgn-Core} method yields the same results whether or not we impose lower bounds on the sizes of the core and periphery sets, as it does not rely on information about the size of the core set. As we discussed in Section \ref{sec:Laplacian}, it depends only on the sign of the entries of the top eigenvector of $L$.  All of the other methods that we examine suffer from a ``boundary effect,'' as the {\sc Find-Cut} algorithm finds a global optimum at (or very close to) the boundary of the search interval. 
When $\beta$ is known, we are planting core and periphery sets of known sizes, so we can examine the number of false-positive errors (i.e., vertices incorrectly assigned to the core set) and false-negative errors (i.e., vertices incorrectly assigned to the periphery set) for the various methods for detecting core--periphery structure. If we enforce a minimum size of 20 for the core and periphery sets, we find that {\sc LapSgn-Core} is the only method that yields satisfactory results from this perspective, because all other methods find a maximum of the objective function that lies close to the boundary. When we increase the lower bound of the core and periphery sets from 20 to 50, the {\sc Degree-Core} and {\sc Path-Core} methods yield very good results (in terms of the numbers of false positives and false negatives), followed by {\sc LapSgn-Core}, {\sc Lap-Core}, {\sc LowRank-Core}, and {\sc Core-Score}. When the fraction of vertices that belong to the core is known, then {\sc Degree-Core}, {\sc Path-Core}, and {\sc LowRank-Core} again yield the best results, followed by {\sc LapSgn-Core}, {\sc Lap-Core}, and {\sc Core-Score}.

Again evaluating the methods in terms of the number of false positives and false negatives, one can increase the accuracy of the methods to detect core--periphery structure by considering other local maxima of the objective function \eqref{cpobjDens}, especially if one is searching further away from the boundary. However, for these examples, the {\sc LapSgn-Core} and {\sc Core-Score} methods still yield unsatisfactorily results even when considering additional local minima. Interestingly, their objective functions are monotonic (increasing for the former and decreasing for the latter) with respect to the vector of sorted scores. After assigning vertices to a core set or peripheral set using any of the methods above, one can also add a post-processing step in the spirit of either the gradient-descent refinement step in non-convex optimization \cite{NesterovOpt} or Kernighan--Lin vertex swaps in community detection \cite{New06,Richardson2009}.

The critical eye may object that a separation based on vertex degree yields results that are as good as the other best-performing methods. However, the recent work of \cite{XZhang2014} demonstrated that {\sc Degree-Core} separation is suboptimal for certain types of networks, although {\sc Degree-Core} appears to be good enough when there is only a weak core--periphery structure. When a network's core and periphery are separated very strongly, examining vertex degree also appears to be reasonable. However, for pronounced core--periphery structure that is neither too weak nor too strong (i.e., in the most relevant situation for applications \cite{XZhang2014}), one needs to use methods that are more sophisticated than simply considering vertex degrees. Reference~\cite{puckmason} also includes a salient discussion of examining a network's core--periphery structure simply by computing vertex degrees. To illustrate the sensitivity of the {\sc Degree-Core} method to the presence of high-degree peripheral vertices, we perform a pair of numerical experiments in which we purposely plant high-degree vertices in the periphery set (see Appendix 3). In these experiments, the {\sc LapSgn-Core} method achieves the lowest number of errors, whereas {\sc Degree-Core} is one of the worst performers. In addition, one can see from Table \ref{tab:correlation_real_networks}, which gives the Pearson and Spearman correlation coefficients for various coreness measures, that the results of our proposed methods are often only moderately correlated with {\sc Degree-Core}, and they can thus return solutions that differ significantly from naive separation based on vertex degree. From the perspective of applications, we note the work of Kitsak et al. \cite{kitsak2010identification} on the identification of influential spreaders in networks. Kitsak et al. argued that the position of a vertex relative to the organization of a network determines its spreading influence to a larger extent than any local property (e.g., degree) of a vertex. Their findings also suggest that a network's core vertices (as measured by being in the $k$-core of a network with high $k$) are much better spreaders of information than vertices with merely high degree. Recent followup work has also suggested that many core spreaders need not have high degrees \cite{morone2015}, further highlighting the substantive difference between core vertices and high-degree vertices.


\subsection{A Family of Synthetic Networks}
\label{sec:family_of_synthetic_networks}

In this section, we detail our numerical results when applying our methods to a family of synthetic networks with a planted core--periphery structure.  We again examine the performance of the methods with respect to how many core and peripheral vertices they classify correctly.

We use variants of the random-graph ensemble that was introduced in \cite{puckmason}. Let $C_1(n,\beta,p,\kappa)$ denote a family of networks with the following properties: $n$ is the number of vertices, $\beta$ is the fraction of vertices in the core, and the edge probabilities for core--core, core--periphery, and periphery--periphery connections are given by $\mathbf{p}=(p_{cc}, p_{cp}, p_{pp})$ with $p_{cc}=\kappa^2 p$, $p_{cp}=\kappa p$, and $p_{pp}=p$. Let $C_2(n,\beta,p,\kappa)$ denote a family of networks, from a slight modification of the above model, in which the edge probabilities are now given by $\mathbf{p}=(p_{cc}, p_{cp}, p_{pp})$ with $p_{cc}=\kappa^2 p$, $p_{cp}=\kappa p$, and $p_{pp}=\kappa p$. In our simulations, we fix $n=100$, $\beta=0.5$, and $p=0.25$, and we examine core--periphery structure using each of the proposed methods. We average our results over 100 different instantiations of the above graph ensembles for each of the parameter values $\kappa=1.1, 1.2, \dots, 2$. We also compare our results with the {\sc Core-Score} algorithm introduced in \cite{puckmason}, and we remark that the results of {\sc Core-Score} correspond are for only single networks drawn from the above ensembles. The inefficient running time of the {\sc Core-Score} algorithm renders it infeasible to average over 100 different instantiations of a graph ensemble.

\begin{figure}[h!]
\begin{center}
\subfigure[Without knowledge of $\beta$; single experiment] 
{\includegraphics[width=0.38\textwidth]{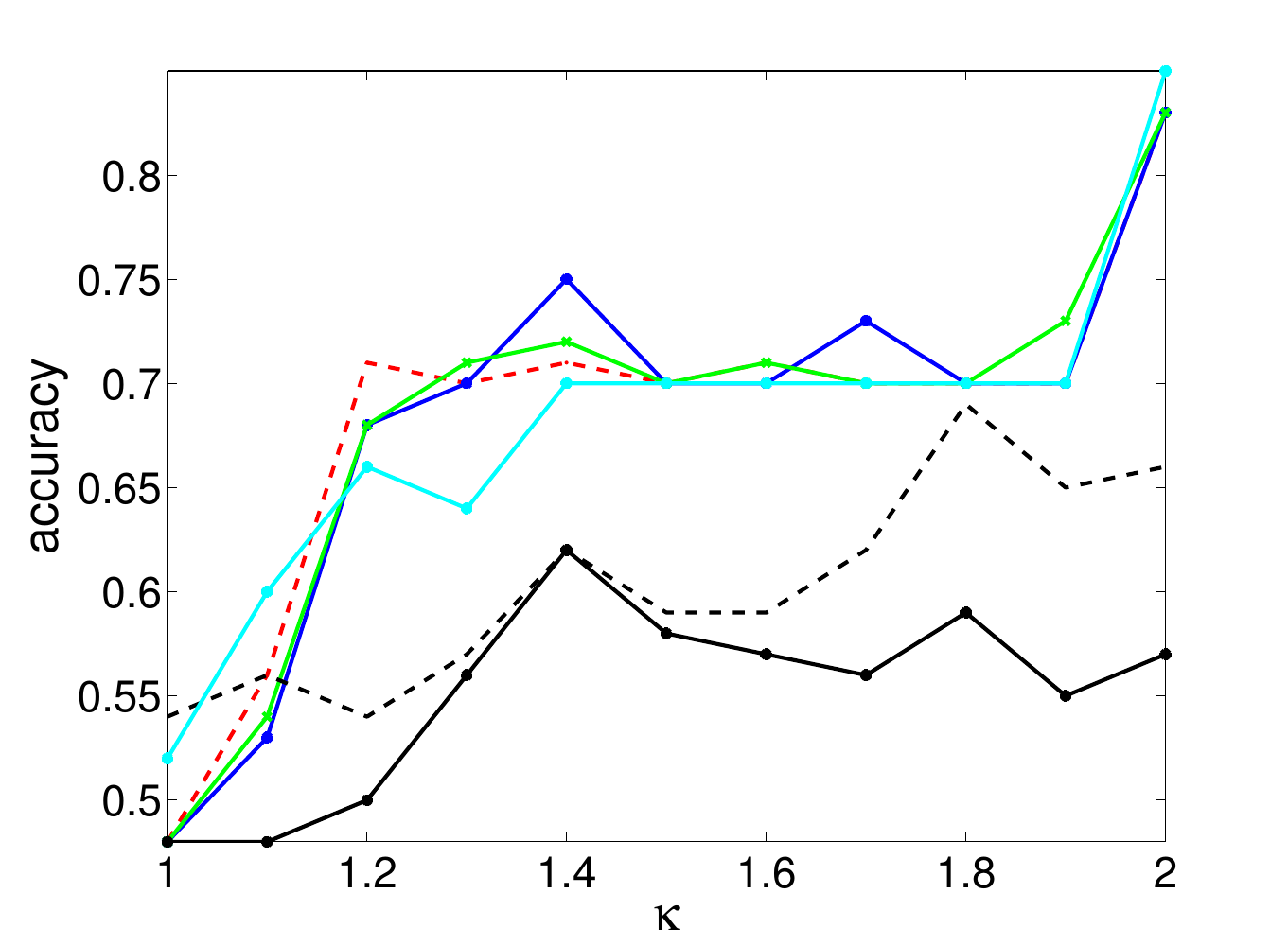}}
\hspace{1cm}
\subfigure[With knowledge of $\beta$; single experiment]
{\includegraphics[width=0.38\textwidth]{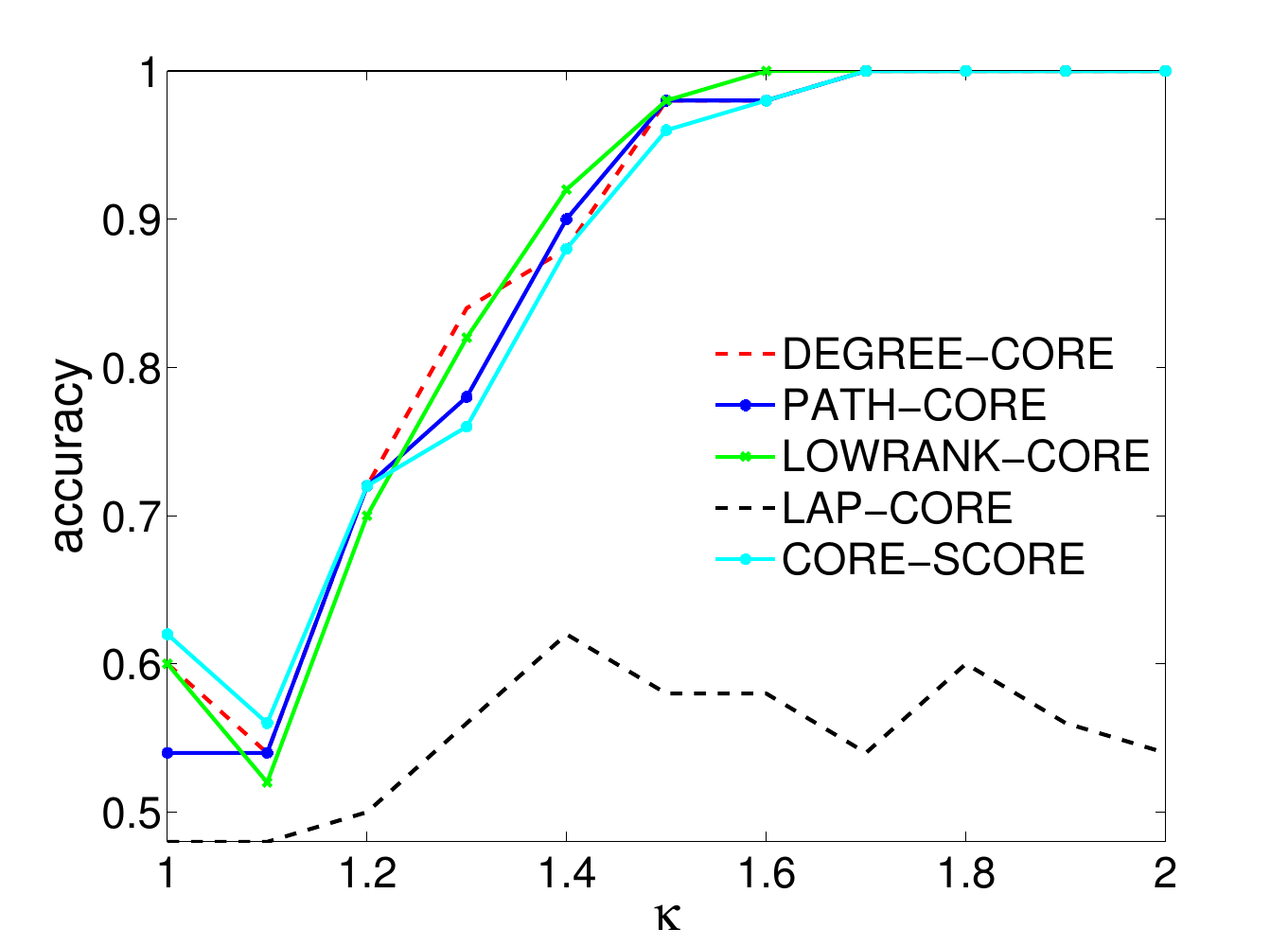}}
\subfigure[Without knowledge of $\beta$; averaged over 100 experiments]
{\includegraphics[width=0.38\textwidth]{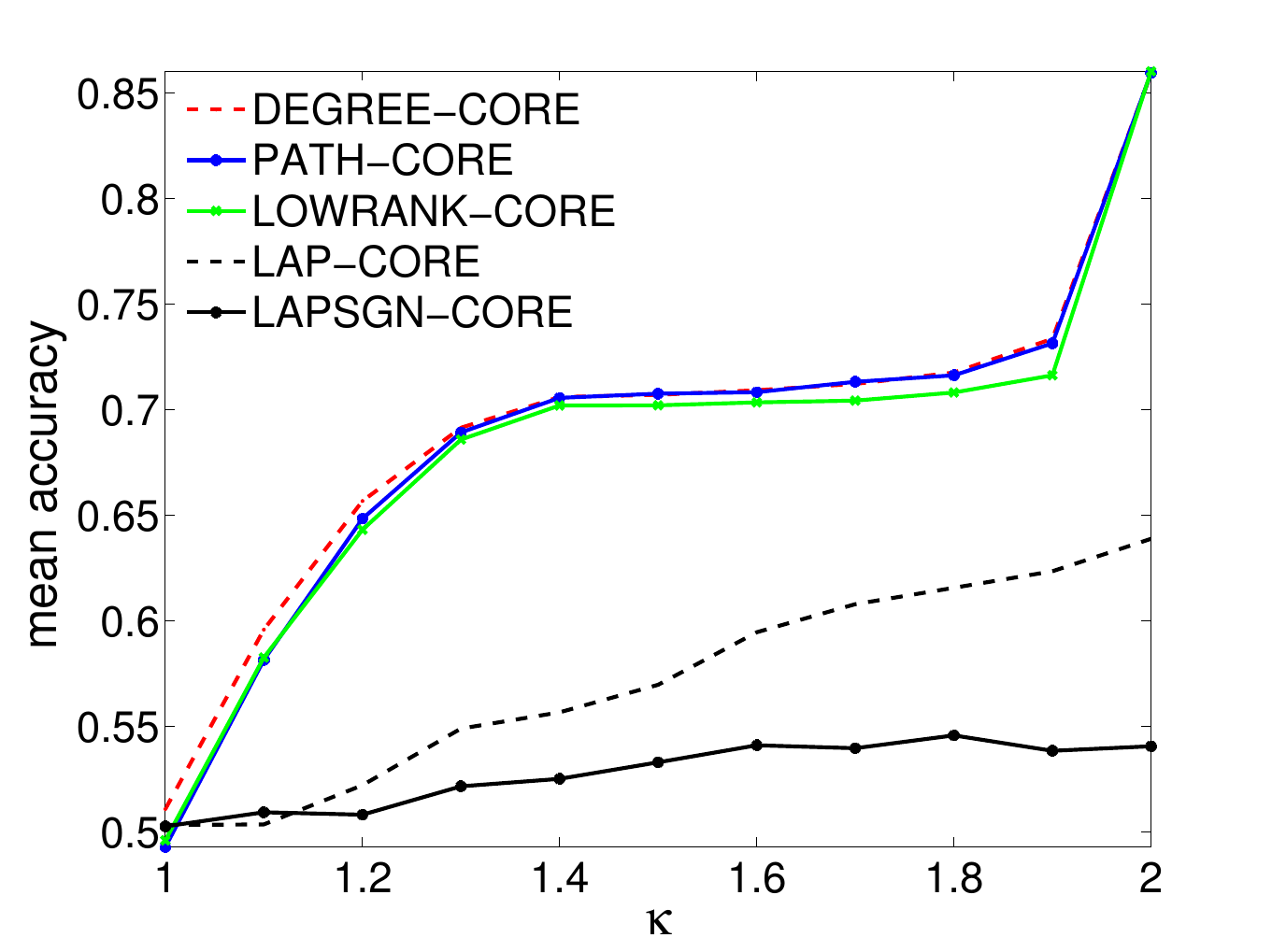}}
\hspace{1cm}
\subfigure[With knowledge of $\beta$; averaged over 100 experiments]
{\includegraphics[width=0.38\textwidth]{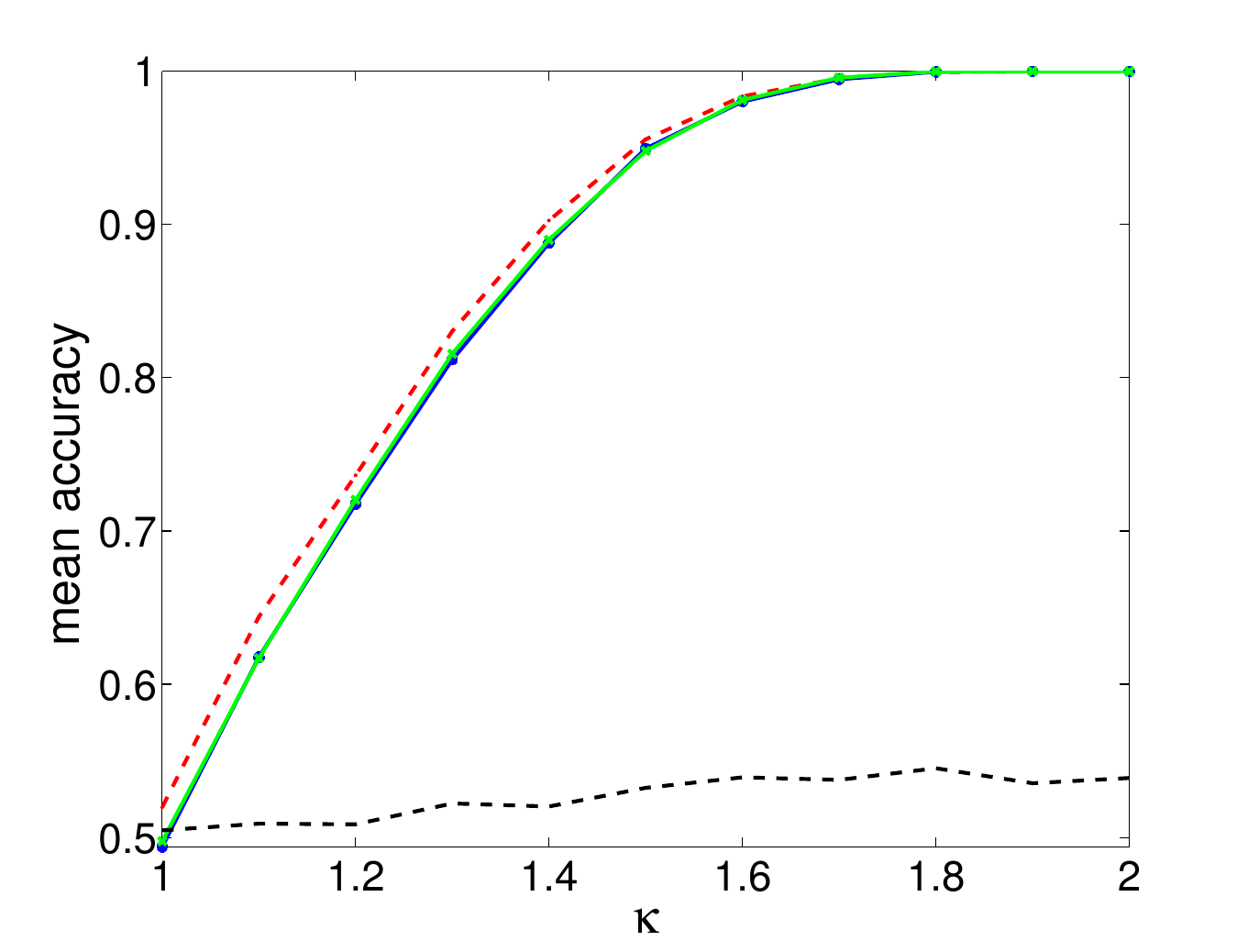}}
\end{center}
\caption{Comparison of methods for core--periphery detection using the graph ensemble $C_1(n,\beta,p,\kappa)$ with $n=100$, $\beta=0.5$, $p=0.25$, and edge probabilities $\mathbf{p}=(p_{cc}, p_{cp}, p_{pp})$, where $p_{cc}=\kappa^2 p$, $p_{cp}=\kappa p$, and $p_{pp}=p$. We vary $\kappa \in [1,2]$ in increments of $0.1$. The top plots illustrate our results for all methods on a single graph from $C_1(n,\beta,p,\kappa)$, and the bottom plots give results averaged over 100 different graphs from the ensemble for all methods except {\sc Core-Score}. The left plots do not use information about the size ($\beta$) of the core, as they rely only on the objective function that one maximizes; the right plots explicitly use knowledge of $\beta$. The colors and symbols in the legend in (c) also apply to (a), and the colors and symbols in the legend in (b) also apply to (d).
}
\label{fig:ex_k210}
\end{figure}

In Fig.~\ref{fig:ex_k210}, we examine the ensemble $C_1(n,\beta,p,\kappa)$ and find that {\sc Path-Core}, {\sc Degree-Core}, {\sc LowRank-Core}, and {\sc Core-Score} yield similar results. When $\beta$ is unknown, we find that {\sc Degree-Core}, {\sc Path-Core}, {\sc LowRank-Core}, and {\sc Core-Score} yield similar results to each other. However, when $\beta$ is known (i.e., when we assume a lower bound on the sizes of the core and periphery sets), we find that {\sc Degree-Core} and {\sc LowRank-Core} tend to perform slightly better than {\sc Core-Score} and {\sc Path-Core}. As expected, the aggregate performance of the various algorithms improves significantly when we assume knowledge of $\beta$. Unfortunately, in both scenarios, the two Laplacian-based methods yield very poor results. Recall that {\sc LapSgn-Core} yields exactly the same results both with and without knowledge of $\beta$, so we only show it in the plots without knowledge of $\beta$.

\begin{figure}[h!]
\begin{center}
\subfigure[Without knowledge of $\beta$; single experiment]
{\includegraphics[width=0.38\textwidth]{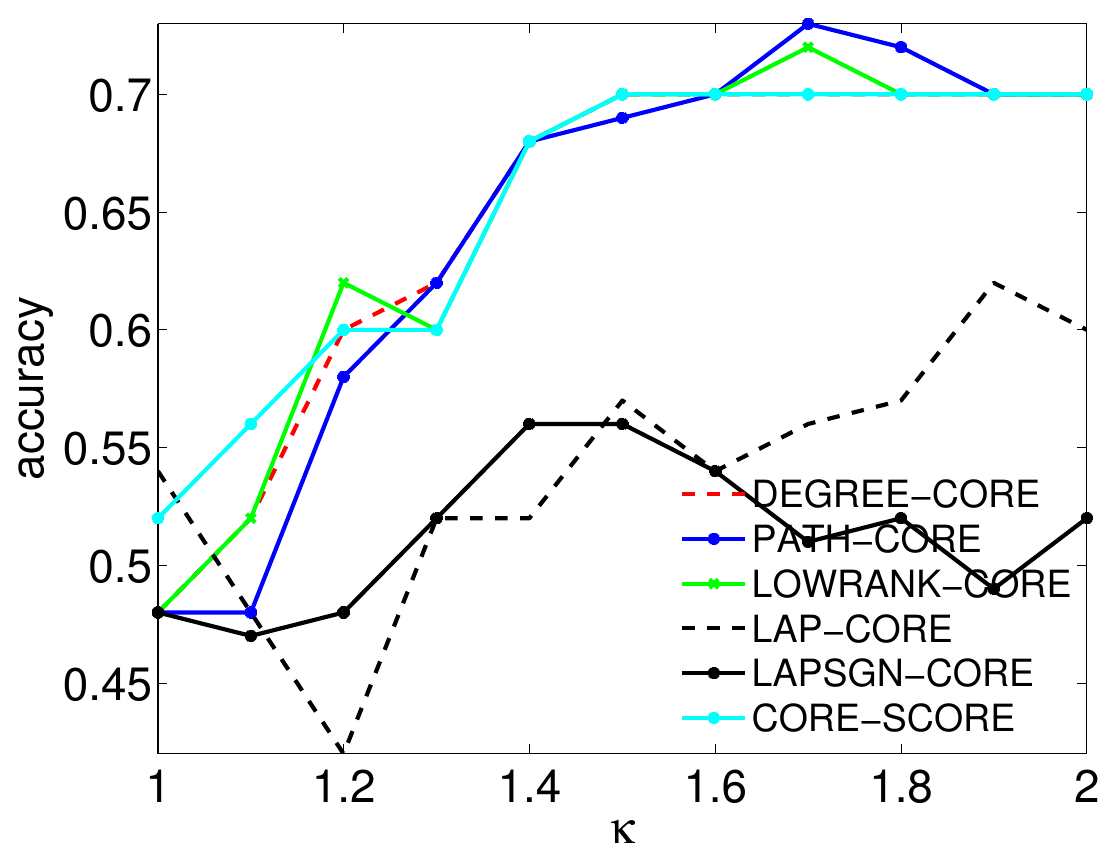}}
\subfigure[With knowledge of $\beta$; single experiment]
{\includegraphics[width=0.38\textwidth]{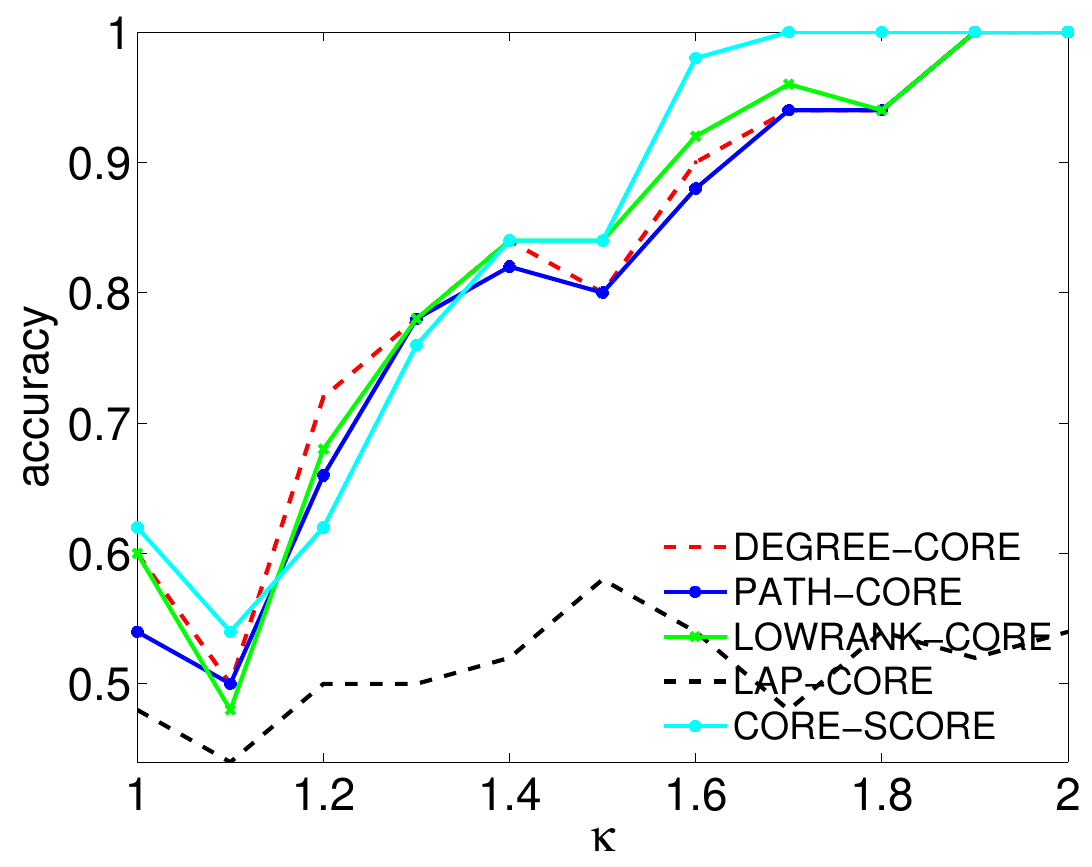}}
\subfigure[Without knowledge of $\beta$; averaged over 100 experiments]
{\includegraphics[width=0.38\textwidth]{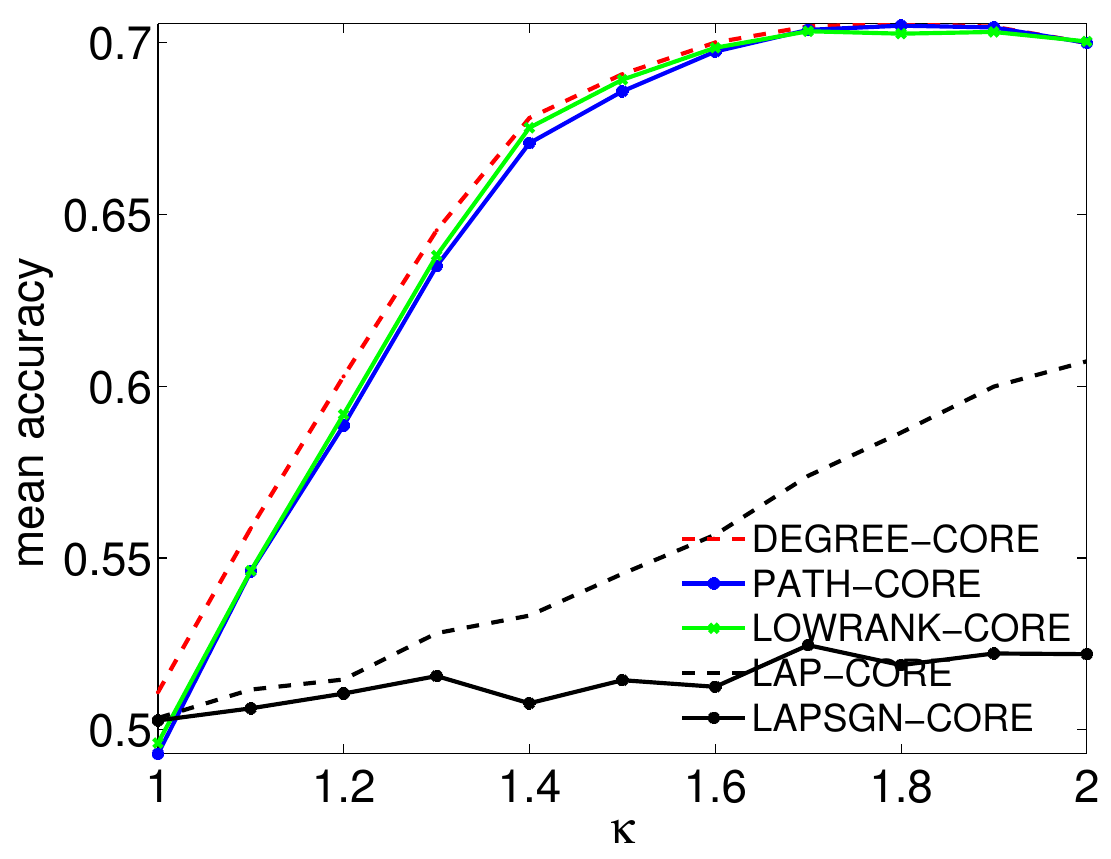}}
\subfigure[With knowledge of $\beta$; averaged over 100 experiments]
{\includegraphics[width=0.38\textwidth]{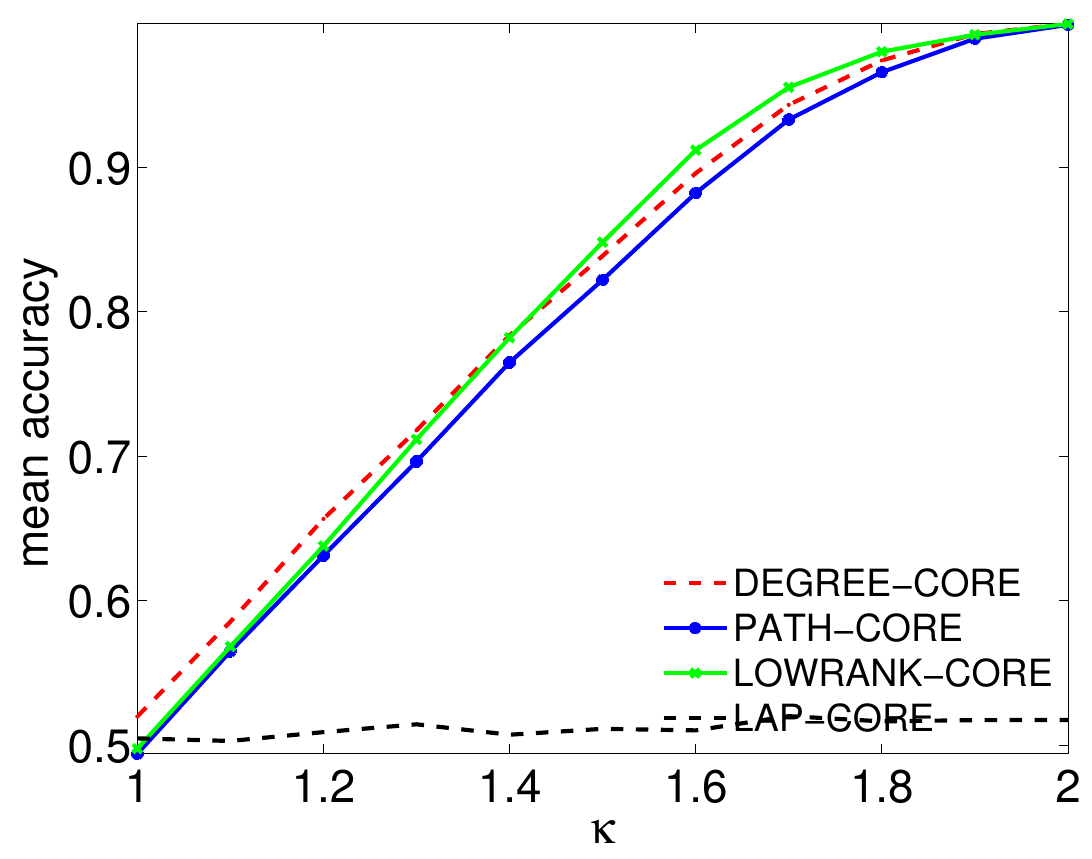}}
\end{center}
\caption{Comparison of methods for detecting core--periphery structure for the graph ensemble $C_2(n,\beta,p,\kappa)$ with $n=100$, $\beta=0.5$, $p=0.25$, and edge probabilities $\mathbf{p}=(p_{cc}, p_{cp}, p_{pp})$, where $p_{cc}=\kappa^2 p$, $p_{cp}=\kappa p$, and $p_{pp}=\kappa p$. We vary $\kappa \in [1,2]$ in increments of $0.1$. 
}
\label{fig:ex_k211}
\end{figure}

In Fig.~\ref{fig:ex_k211}, we plot our numerical results for the ensemble $C_2(n,\beta,p,\kappa)$. When $\beta$ is unknown, {\sc Degree-Core}, {\sc Path-Core}, {\sc LowRank-Core}, and {\sc Core-Score} again yield similar results. When we assume that $\beta$ is known, we find that {\sc Core-Score}, {\sc LowRank-Core}, {\sc Degree-Core} still perform similarly to each other, and they all do slightly better than {\sc Path-Core}. The Laplacian-based methods again perform very poorly, though {\sc Lap-Core} does slightly better than {\sc LapSgn-Core} when $\beta$ is unknown.

\begin{figure}[h!]
\begin{center}
\subfigure[Without knowledge of $\beta$; single experiment]
{  \includegraphics[width=0.38\textwidth]{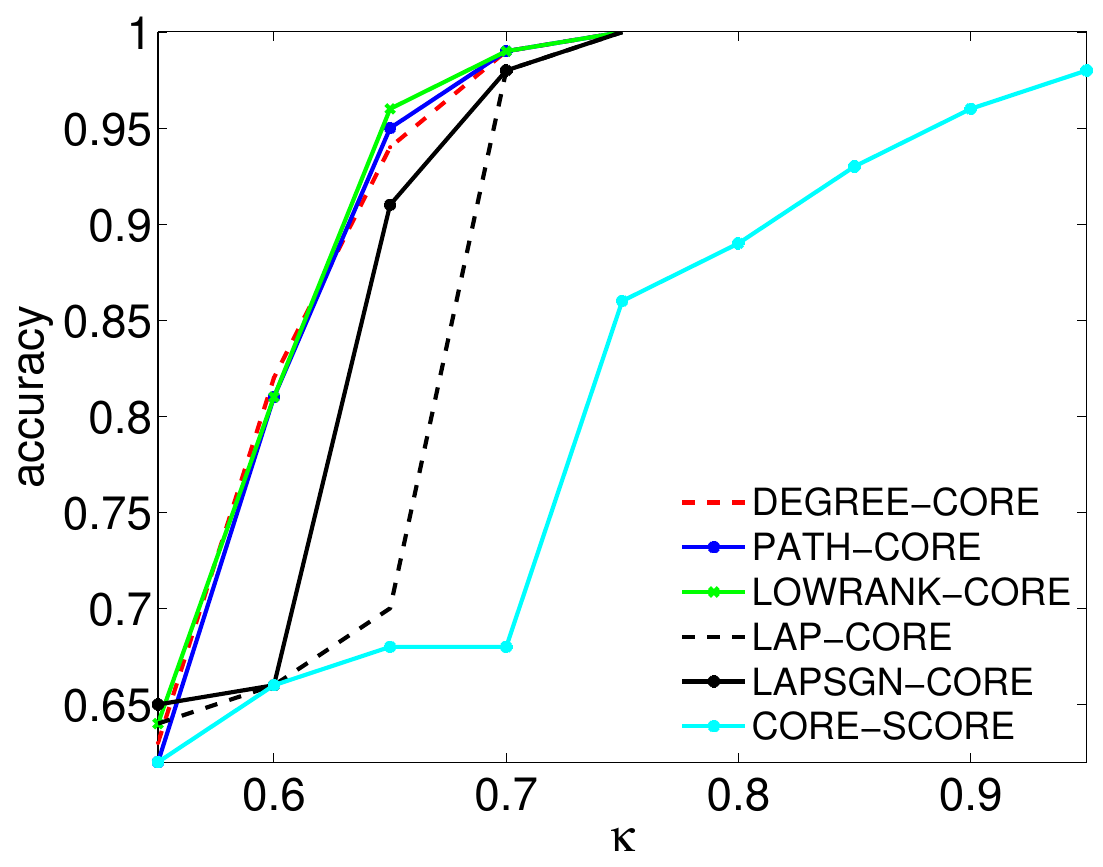}}
\subfigure[With knowledge of $\beta$; single experiment]
{     \includegraphics[width=0.38\textwidth]{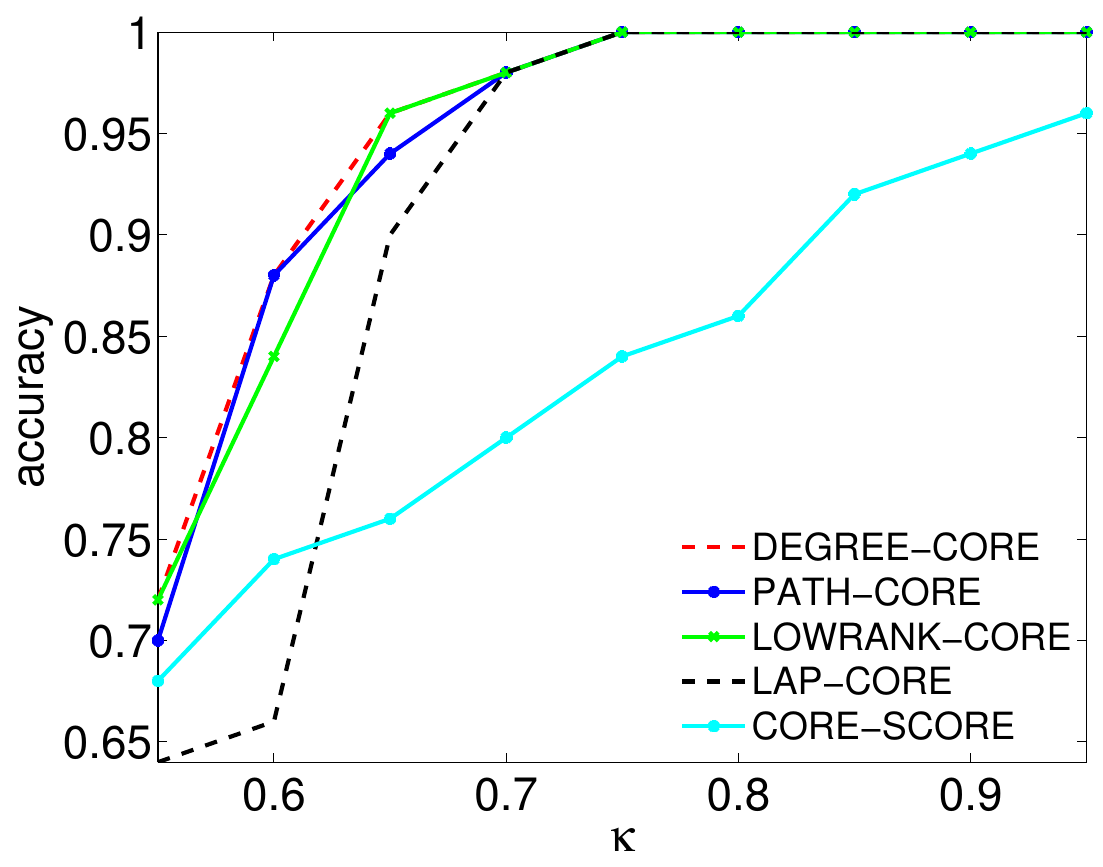}}
\subfigure[Without knowledge of $\beta$; averaged over 100 experiments]
{  \includegraphics[width=0.38\textwidth]{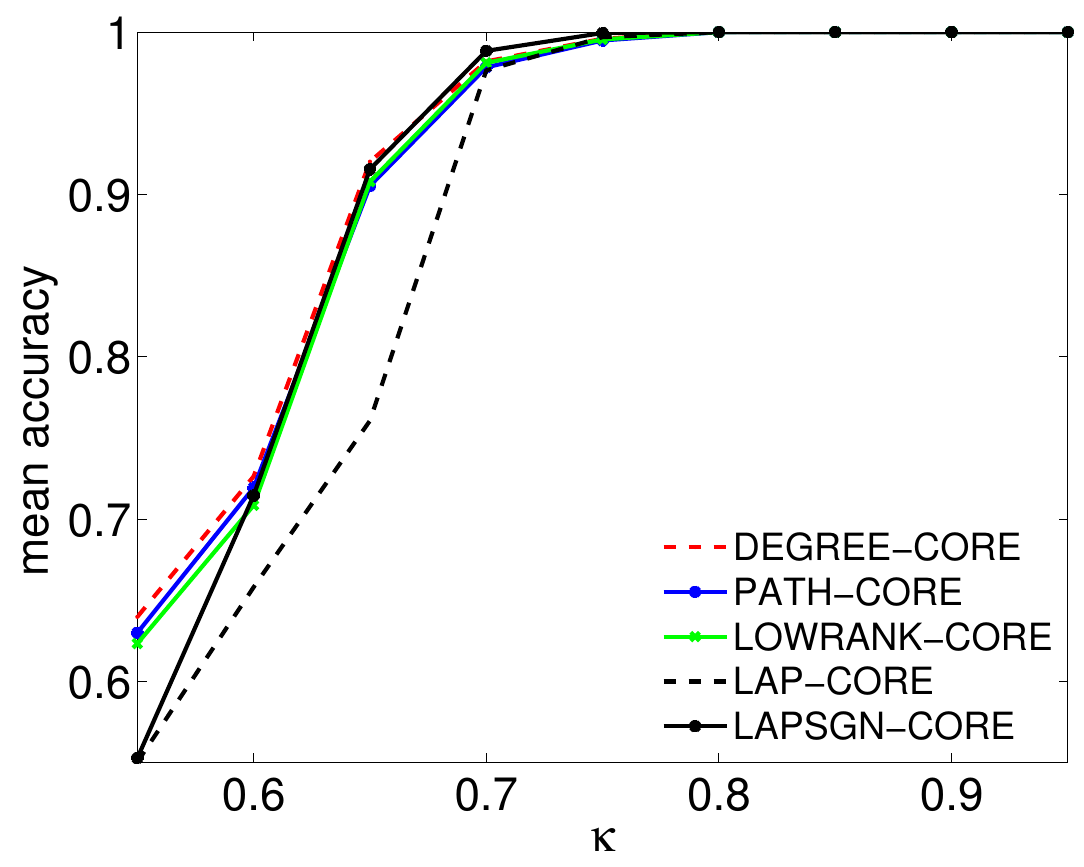}}
\subfigure[With knowledge of $\beta$; averaged over 100 experiments ]
{    \includegraphics[width=0.38\textwidth]{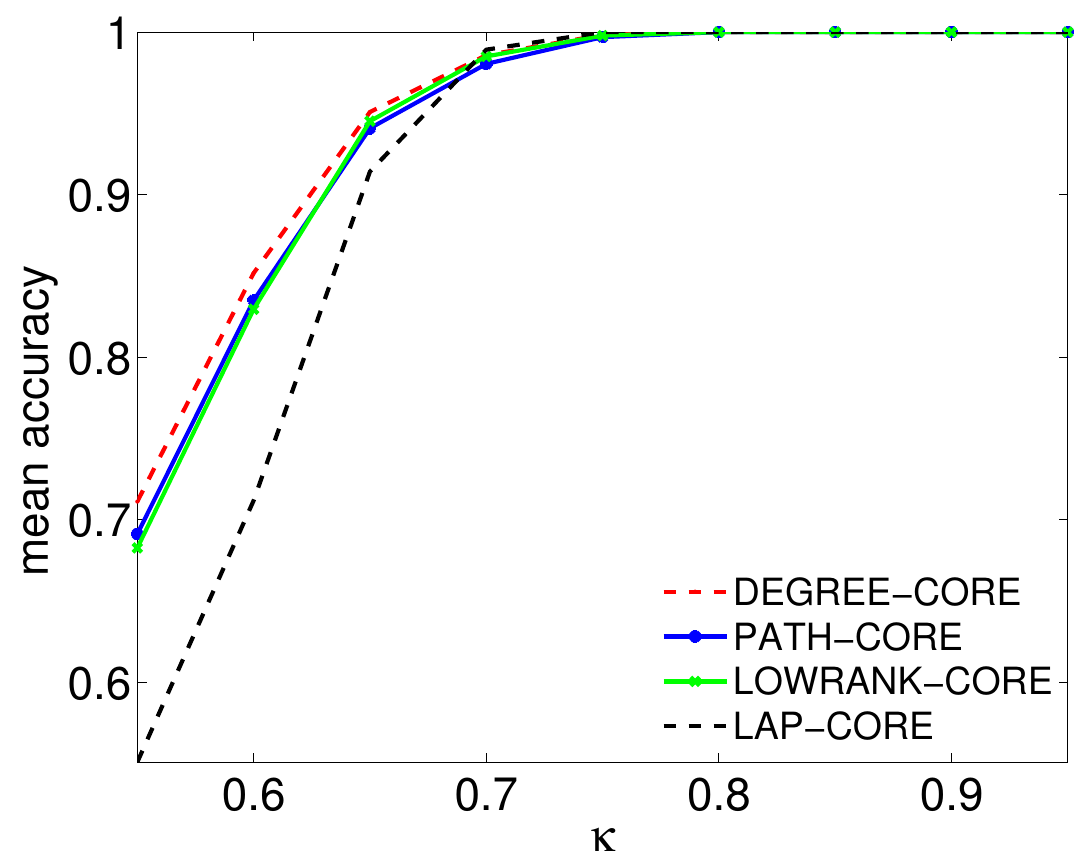}}
\end{center}
\caption{Comparison of the methods for a graph with $n=100$ vertices generated by a core--periphery block model with edge probabilities $\mathbf{p} = (p_{cc},p_{cp},p_{pp}) = (\kappa,\kappa,1-\kappa)$ for $\kappa \in \{0.55, 0.60, \dots,0.95\}$. 
}
\label{fig:ex_rankOneT}
\end{figure}

In Fig.~\ref{fig:ex_rankOneT}, we consider a graph with a core--periphery structure from a random-graph ensemble with edge probabilities $\mathbf{p} = (p_{cc}, p_{cp}, p_{pp}) = (\kappa,\kappa,1-\kappa)$ for different values of $\kappa$. 
The common feature of this set of experiments --- both when the boundary size $\beta$ is known and when it is unknown --- is that {\sc Degree-Core}, {\sc LowRank-Core}, and {\sc Path-Core} give the best results, whereas {\sc Core-Score} consistently comes in last place (except for doing somewhat better than the Laplacian-based methods for values of $\kappa$ in the range $[0.5,1]$) in terms of accuracy. In Fig.~\ref{fig:objFcnEvoAvg}, we consider the values of the objective function \eqref{eq:FindCut}, averaged over 100 runs, that we obtain using the different partitions of a network's vertices into core and periphery sets as we sweep along the sorted scores that we compute using each of the methods (except {\sc Core-Score}, which we omit because of its slow computational time). In Fig.~\ref{fig:objValAttained}, we compare the actual values of the objective function for a single experiment across all methods (including {\sc Core-Score}) as we vary the parameter $\kappa$. We also show the evolution of the value of the objective function as we sweep through the vector of scores from each method.

\begin{figure}[h!]
\begin{center}
\subfigure[$\kappa=0.55$]{\includegraphics[width=0.3\textwidth]{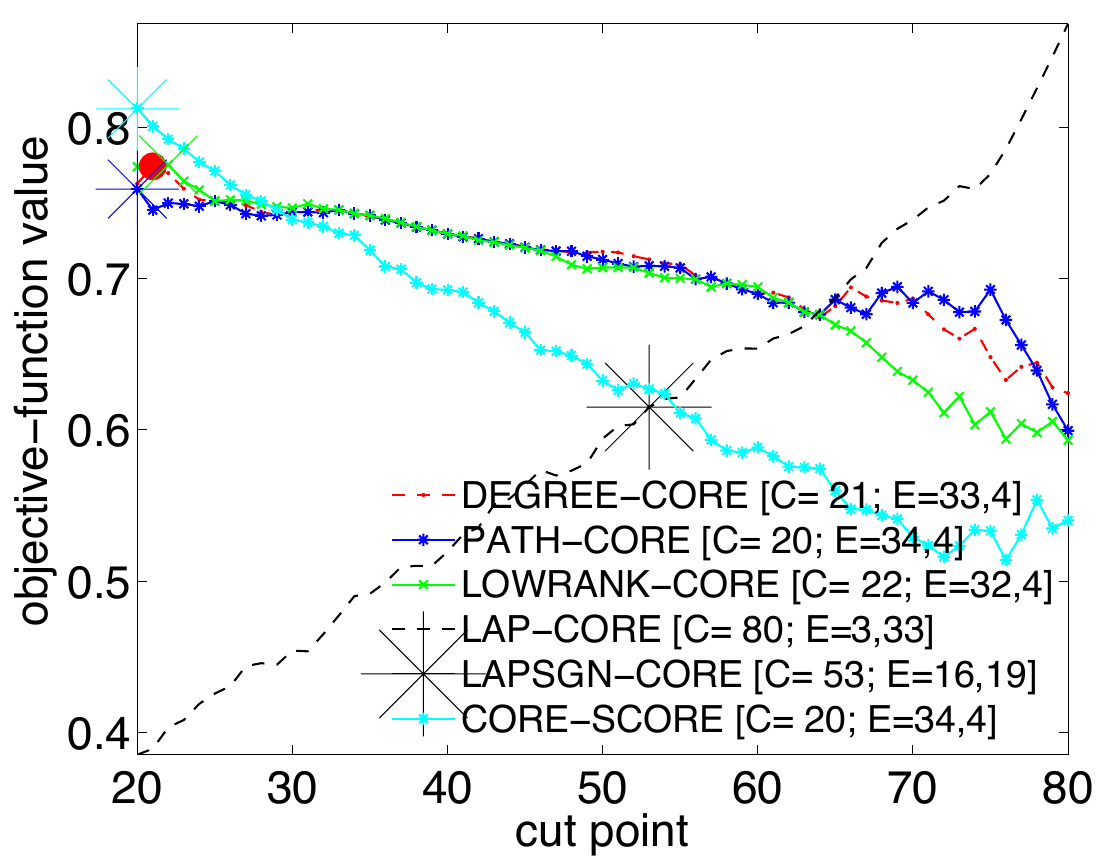}}
\subfigure[$\kappa=0.60$]{\includegraphics[width=0.3\textwidth]{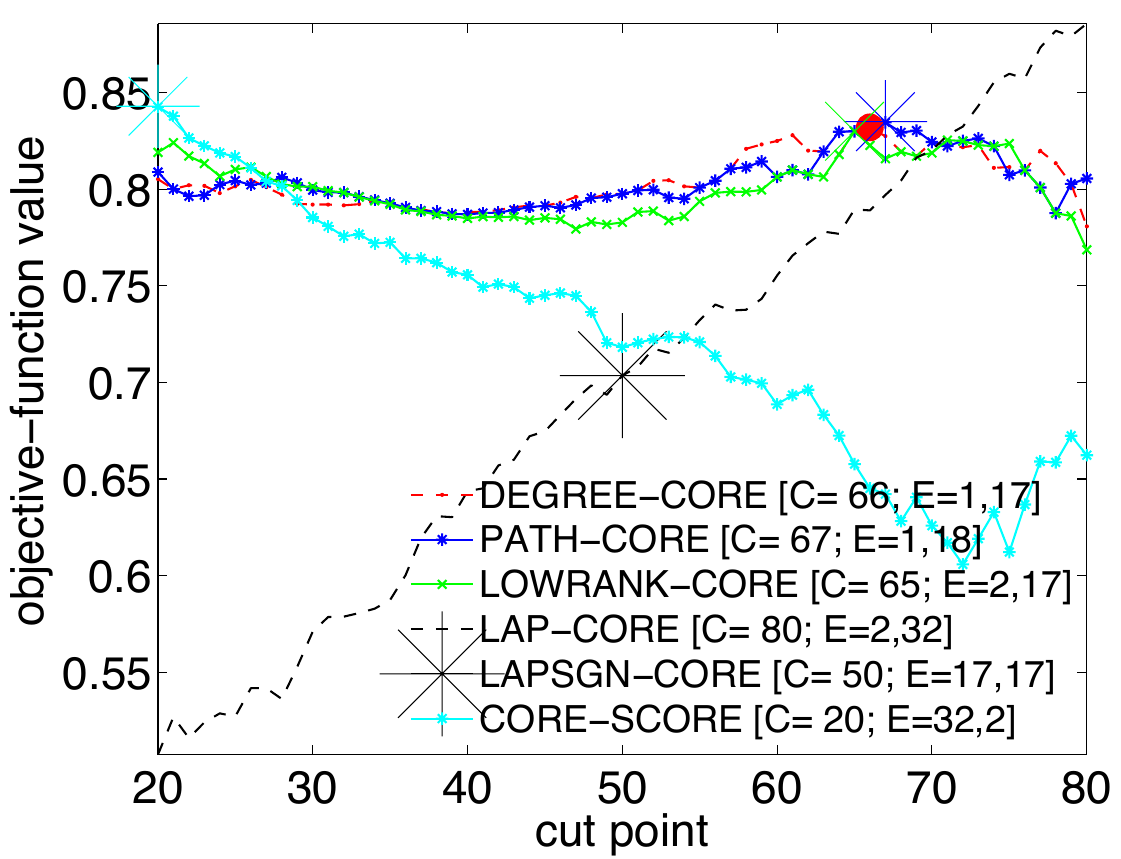}}
\subfigure[$\kappa=0.65$]{\includegraphics[width=0.3\textwidth]{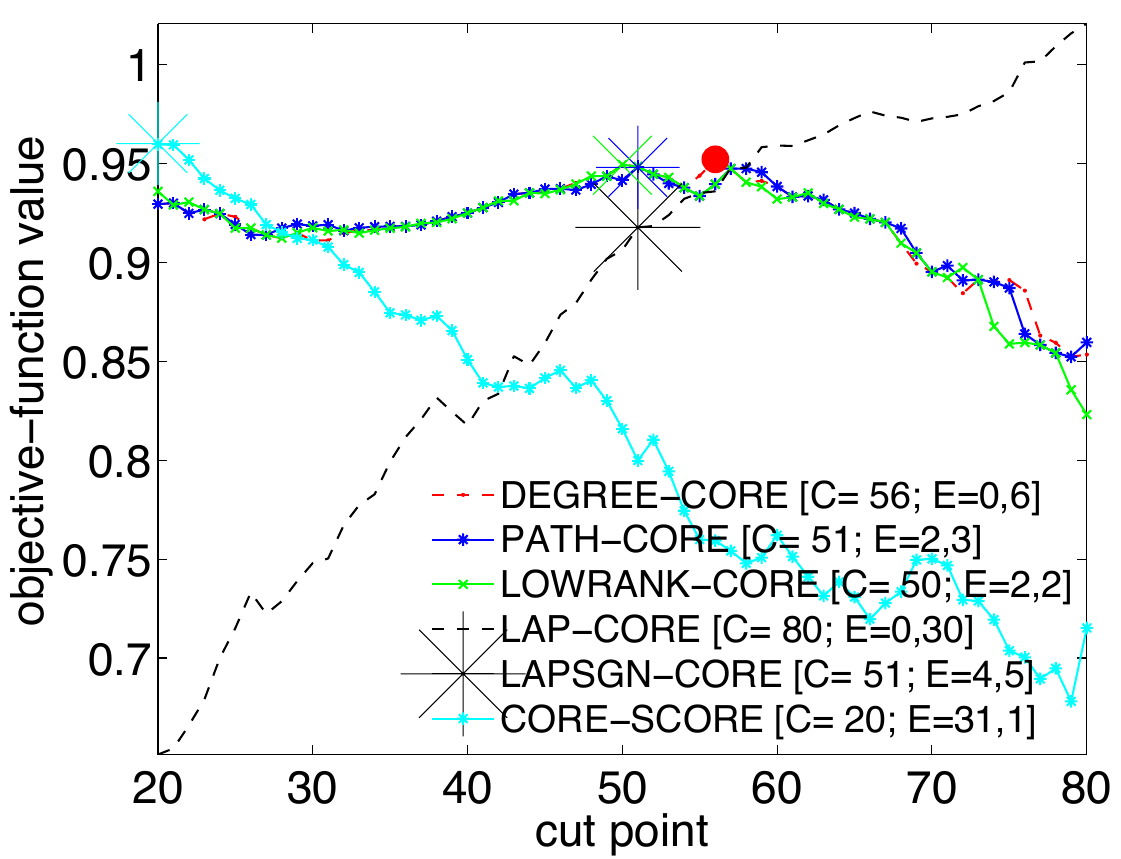}}
\subfigure[$\kappa=0.70$]{\includegraphics[width=0.3\textwidth]{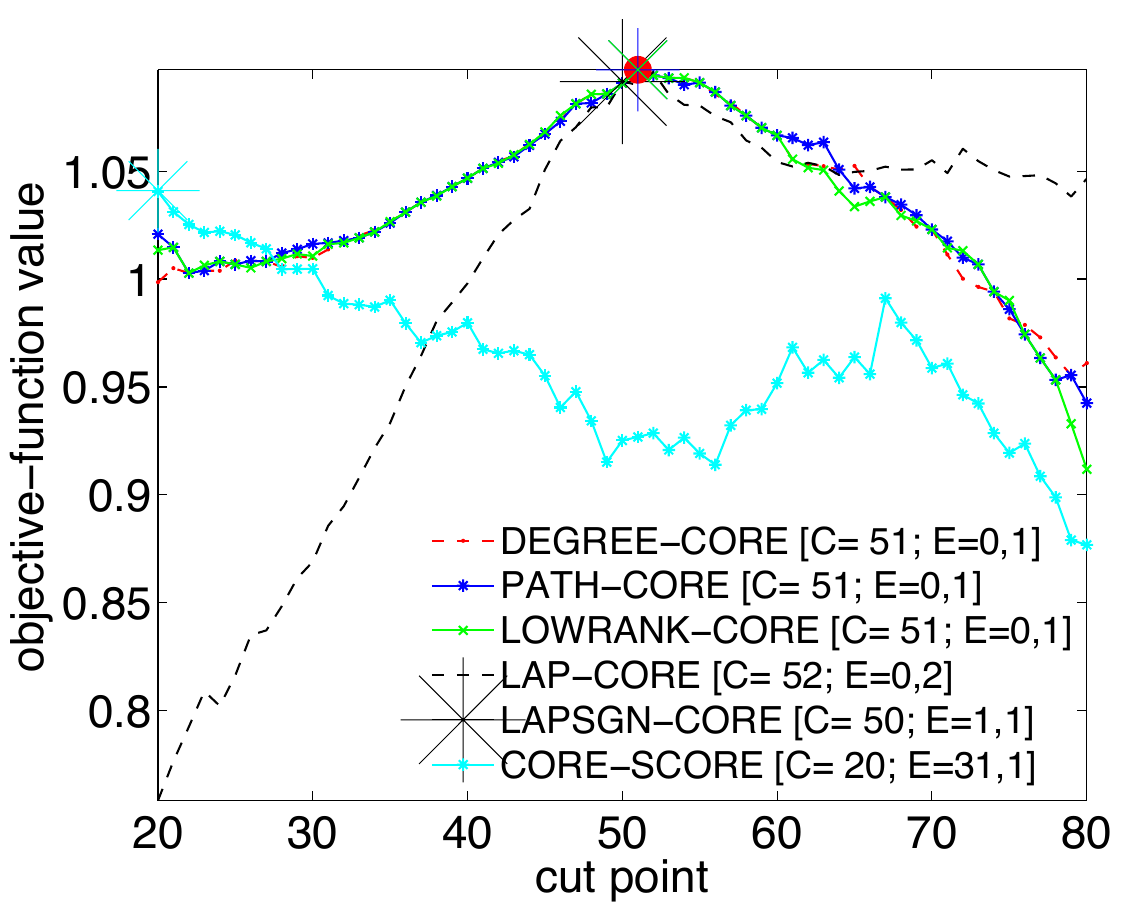}}
\subfigure[$\kappa=0.75$]{\includegraphics[width=0.3\textwidth]{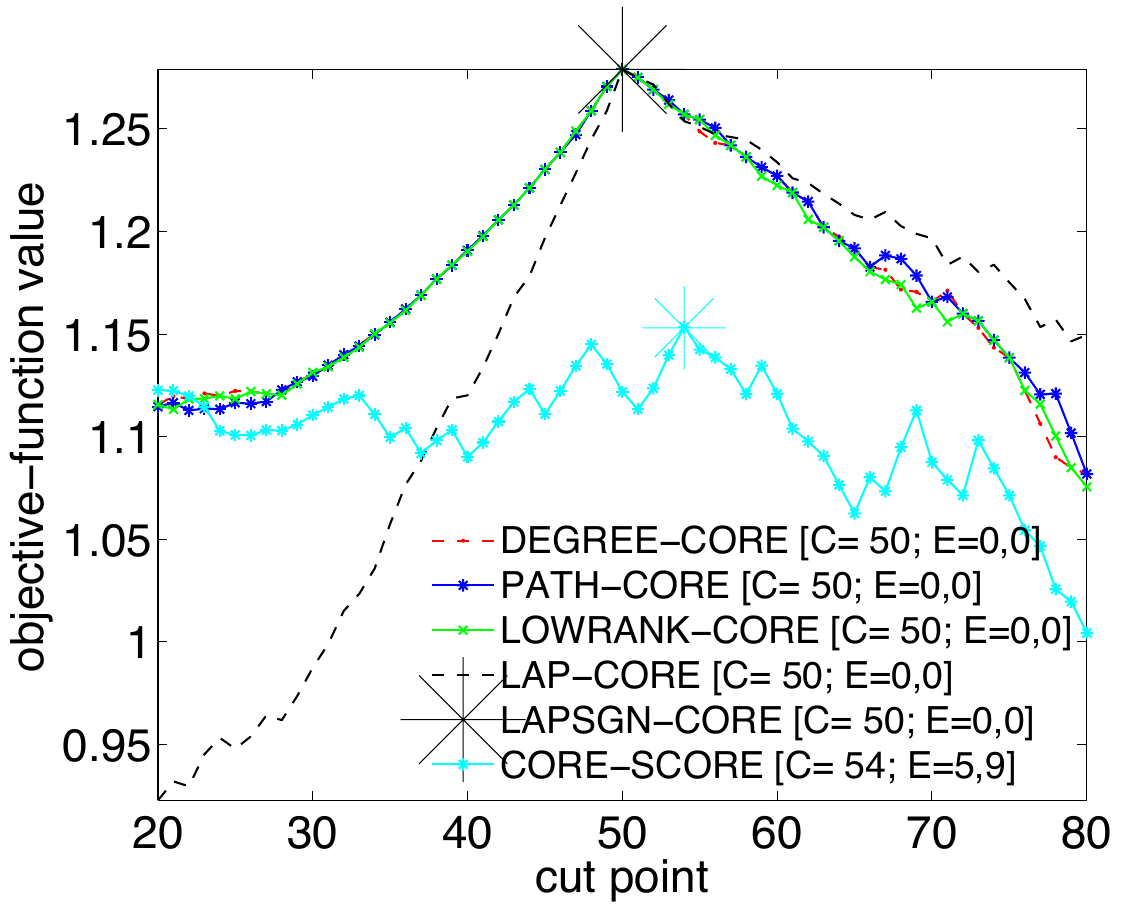}}
\subfigure[$\kappa=0.80$]{\includegraphics[width=0.3\textwidth]{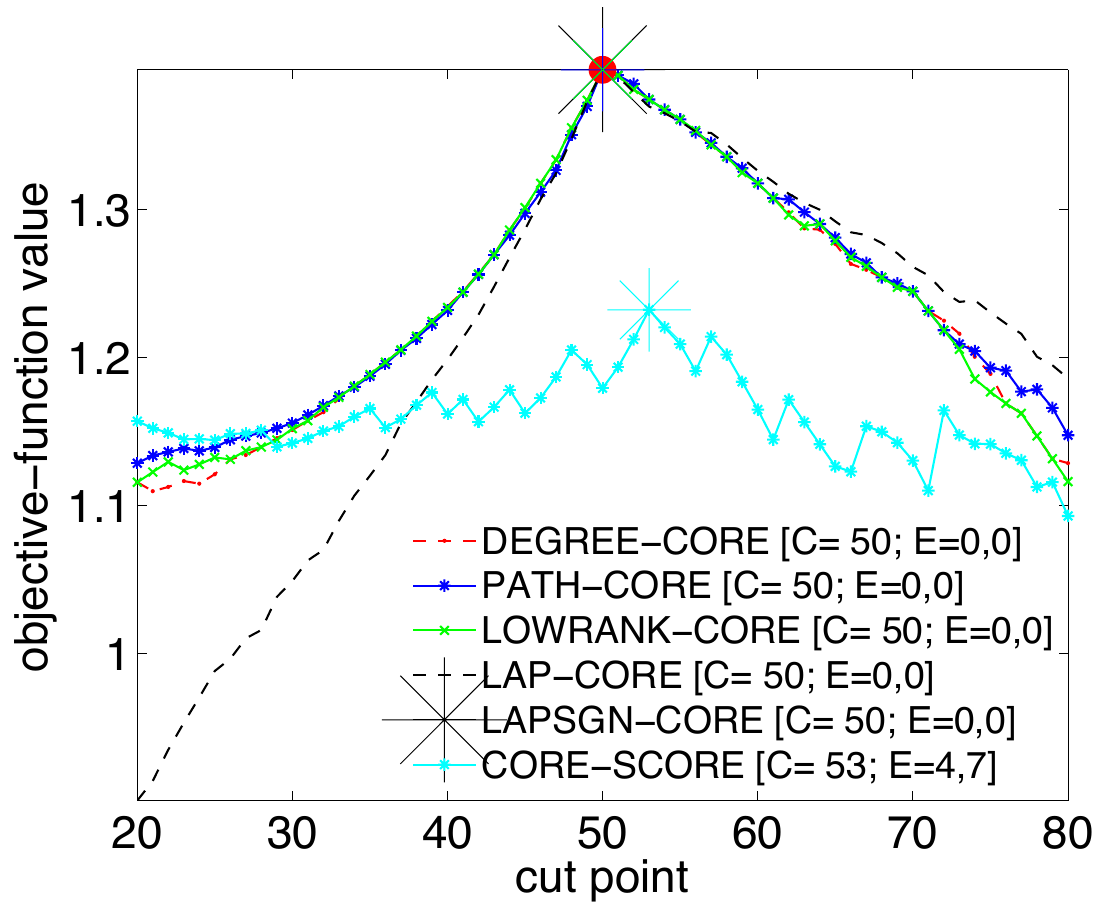}}
\subfigure[$\kappa=0.85$]{\includegraphics[width=0.3\textwidth]{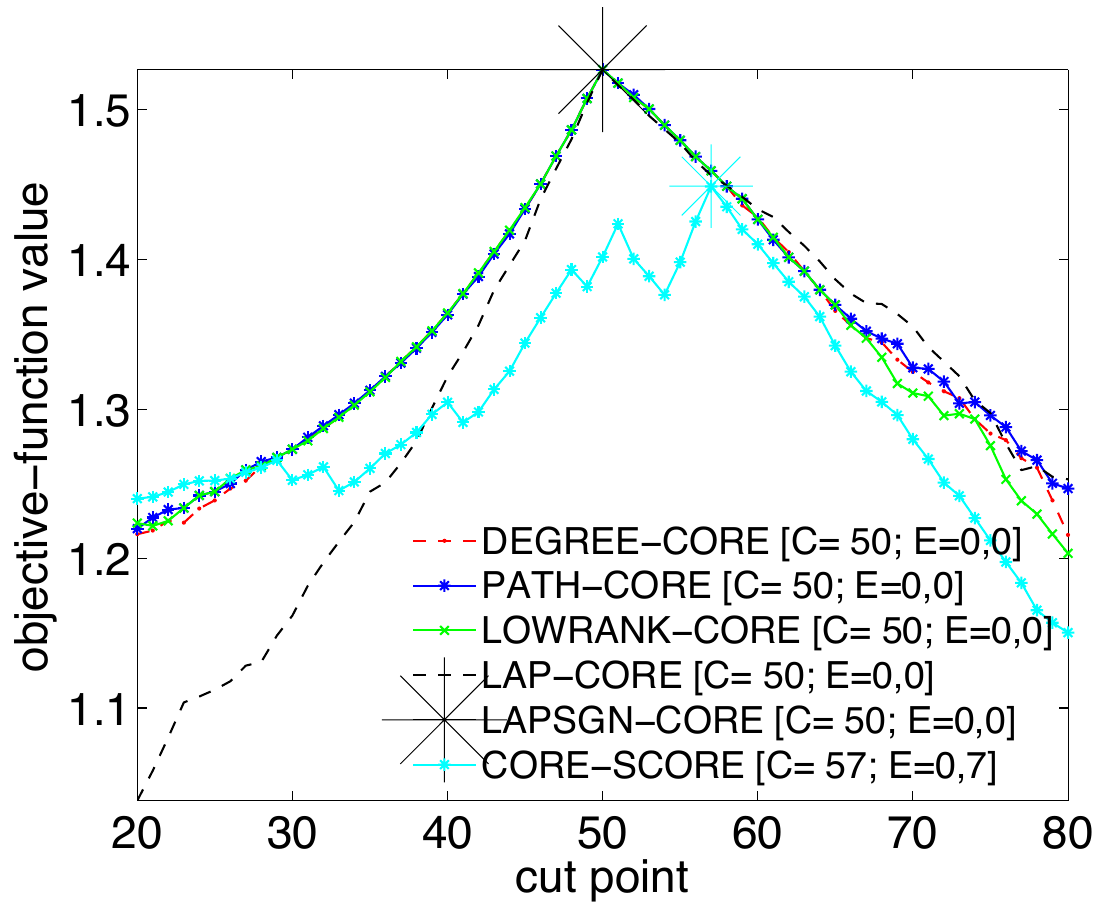}}\hspace{3mm}
\subfigure[$\kappa=0.90$]{\includegraphics[width=0.3\textwidth]{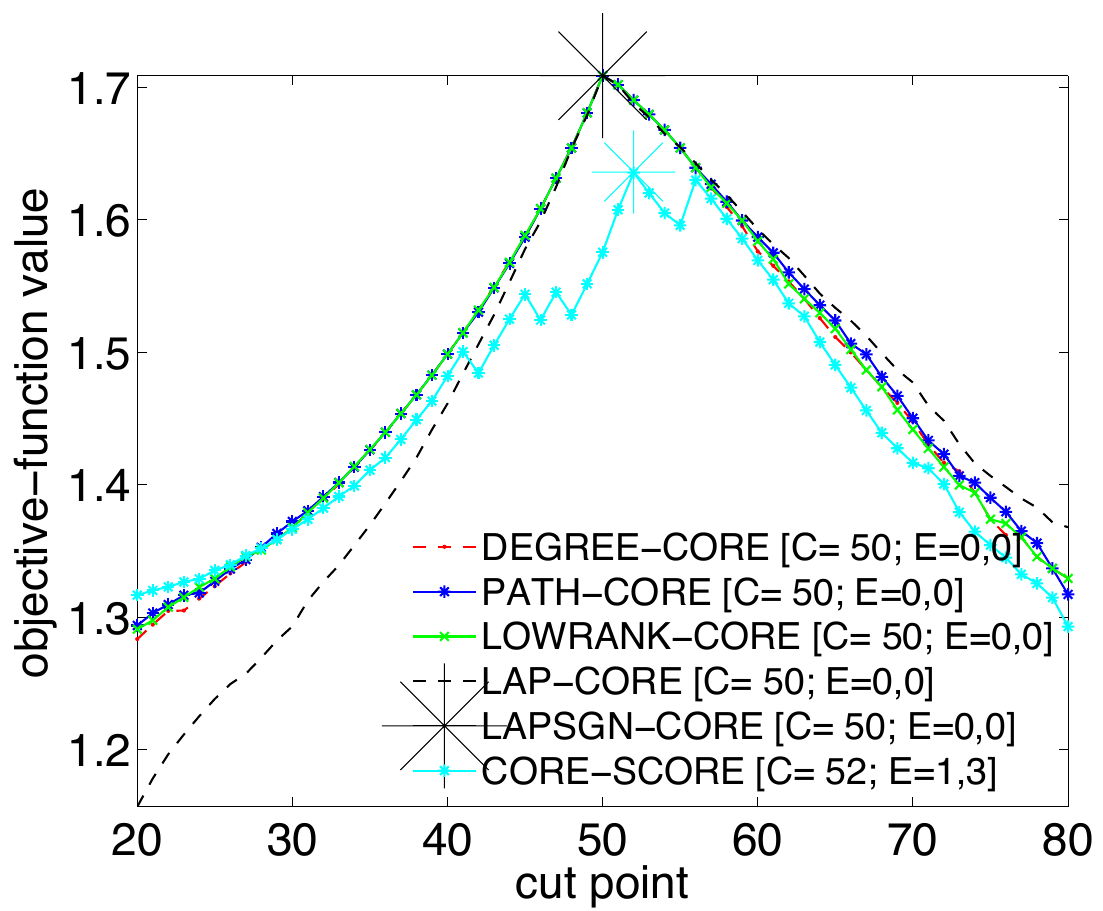}}\hspace{2mm}
\subfigure[$\kappa=0.95$]{\includegraphics[width=0.3\textwidth]{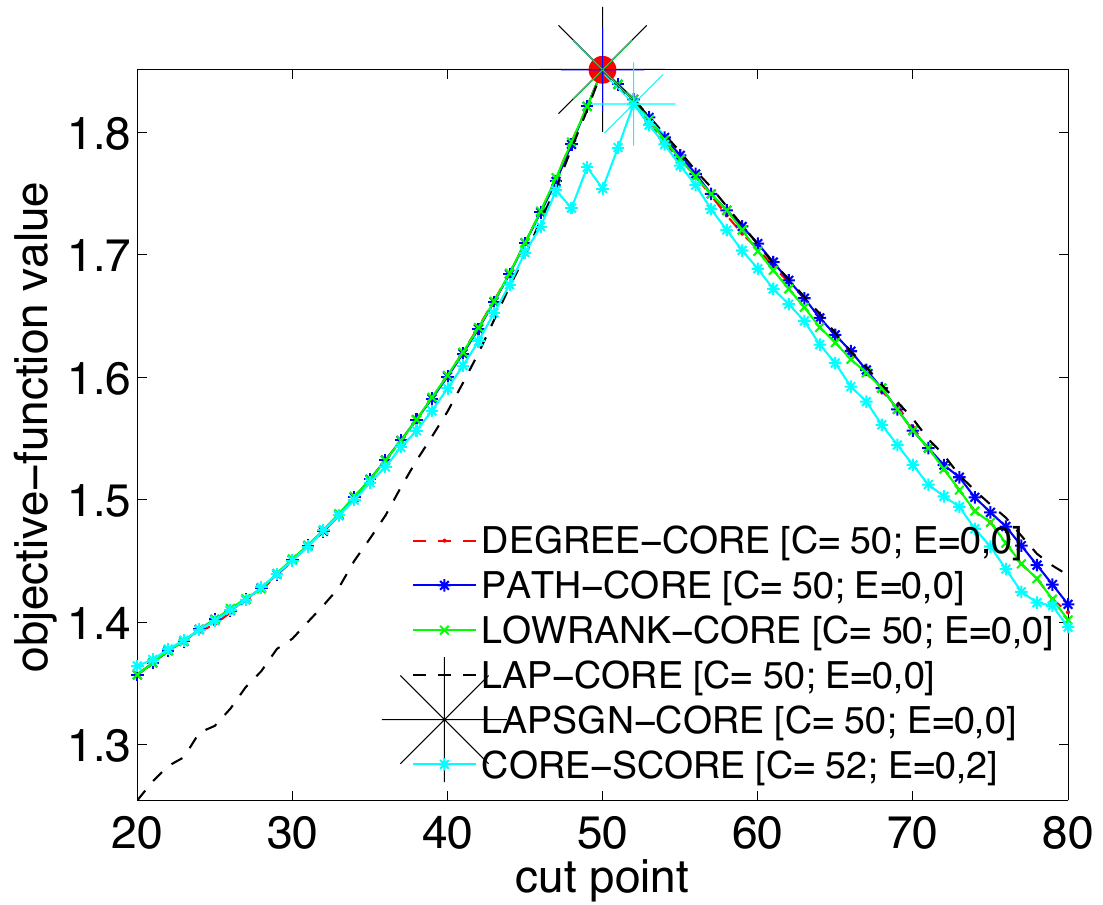}}
\end{center}
\caption{Comparison of the values of the objective function \eqref{eq:FindCut} of the partition of networks into a core set and a periphery set. We calculate these values from the sorted scores from the various methods for detecting core--periphery structure as we vary the parameter $\kappa$ in the ensemble $G(p_{cc},p_{cp},p_{pp})$ from Table \ref{tab:generalBlockModel} with $n=100$. The probability vector in the block model is $\mathbf{p} = (p_{cc}, p_{cp}, p_{pp}) = (\kappa,\kappa,1-\kappa)$. The ``cut point'' refers to the number of vertices in the core set. In the legends, ${\tt C}$ denotes the size of the core set that maximizes the objective function Eq.~\eqref{eq:FindCut}, and ${\bf E} = (y_1,y_2)$ denotes the corresponding $2$-vector of errors. The first component of ${\bf E}$ indicates the number of core vertices that we label as peripheral vertices, and the second indicates the number of peripheral vertices that we label as core vertices.
}
\label{fig:objFcnEvoAvg}
\end{figure}

\begin{figure}[h!]
\begin{center}
\subfigure[Without knowledge of $\beta$]{  \includegraphics[width=0.35\textwidth]{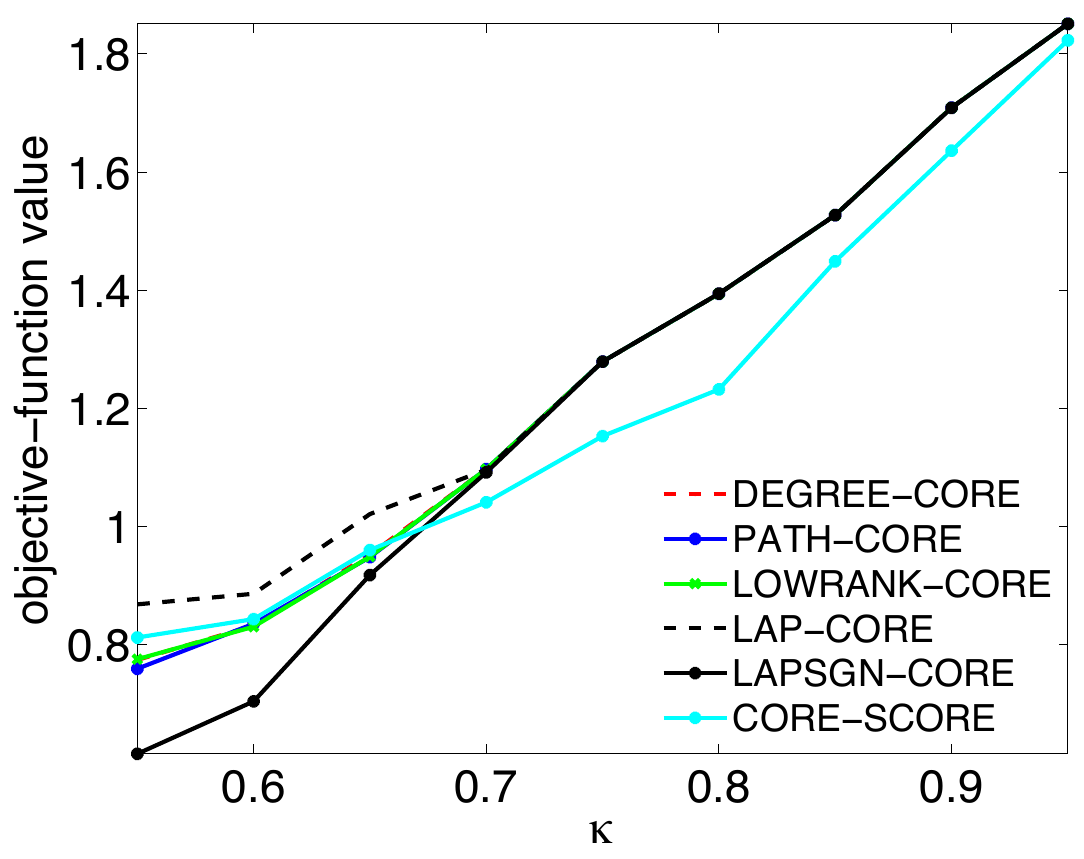}}
\subfigure[With knowledge of $\beta$]{  \includegraphics[width=0.35\textwidth]{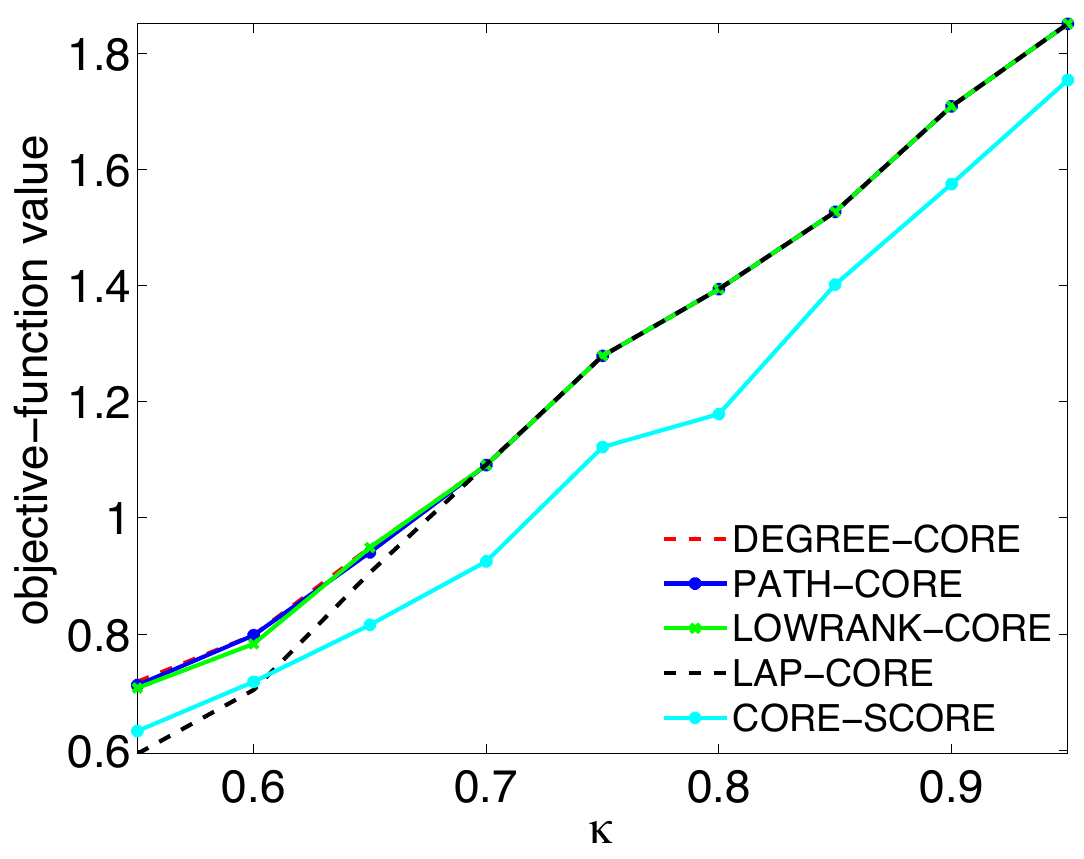}}
\end{center}
\caption{Comparison of the actual values of the objective function \eqref{eq:FindCut} for a single experiment for all methods as a function of the parameter $\kappa$ in the ensemble $G(p_{cc},p_{cp},p_{pp})$ from Table \ref{tab:generalBlockModel} with $n=100$ and $\beta=0.5$. The probability vector in the block model is $\mathbf{p} = (p_{cc},p_{cp},p_{pp}) = (\kappa,\kappa,1-\kappa)$.
}  
\label{fig:objValAttained}
\end{figure}

In Fig.~\ref{fig:runningTimes}, we compare the computation times (in seconds and on a $\log_{10}$ scale) for all of the methods that we examine. The computers that we use for this comparison have 12 CPU cores (Intel(R) Xeon(R) X5650 @ 2.67GHz) and have 48 GB RAM. The most computationally expensive method is {\sc Core-Score}, which is 1--2 orders-of-magnitude slower than {\sc Path-Core}, which is in turn 3--4 orders of magnitude slower than the spectral {\sc LowRank-Core} and {\sc Lap-Score} methods (which have very similar computation times). Finally, as expected, the trivial {\sc Degree-Core} method has the fastest computation times.

\begin{figure}[h!]
\subfigure[$C_1(n,\beta,p,\kappa)$, $n=100, \beta=0.5, p=0.25$]{  \includegraphics[width=0.35\textwidth]{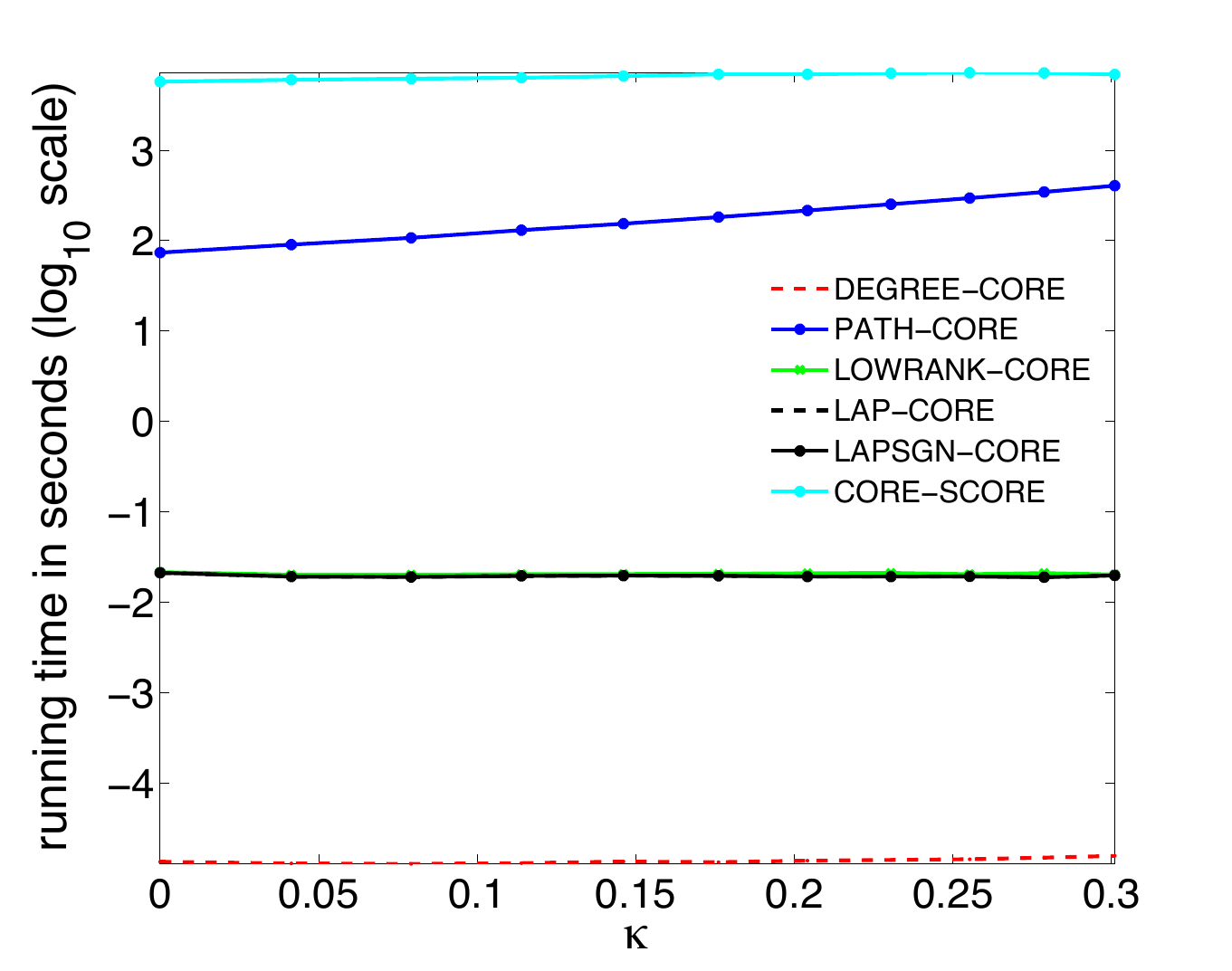}}
\subfigure[$C_2(n,\beta,p,\kappa)$, $n=100, \beta=0.5, p=0.25$]{     \includegraphics[width=0.35\textwidth]{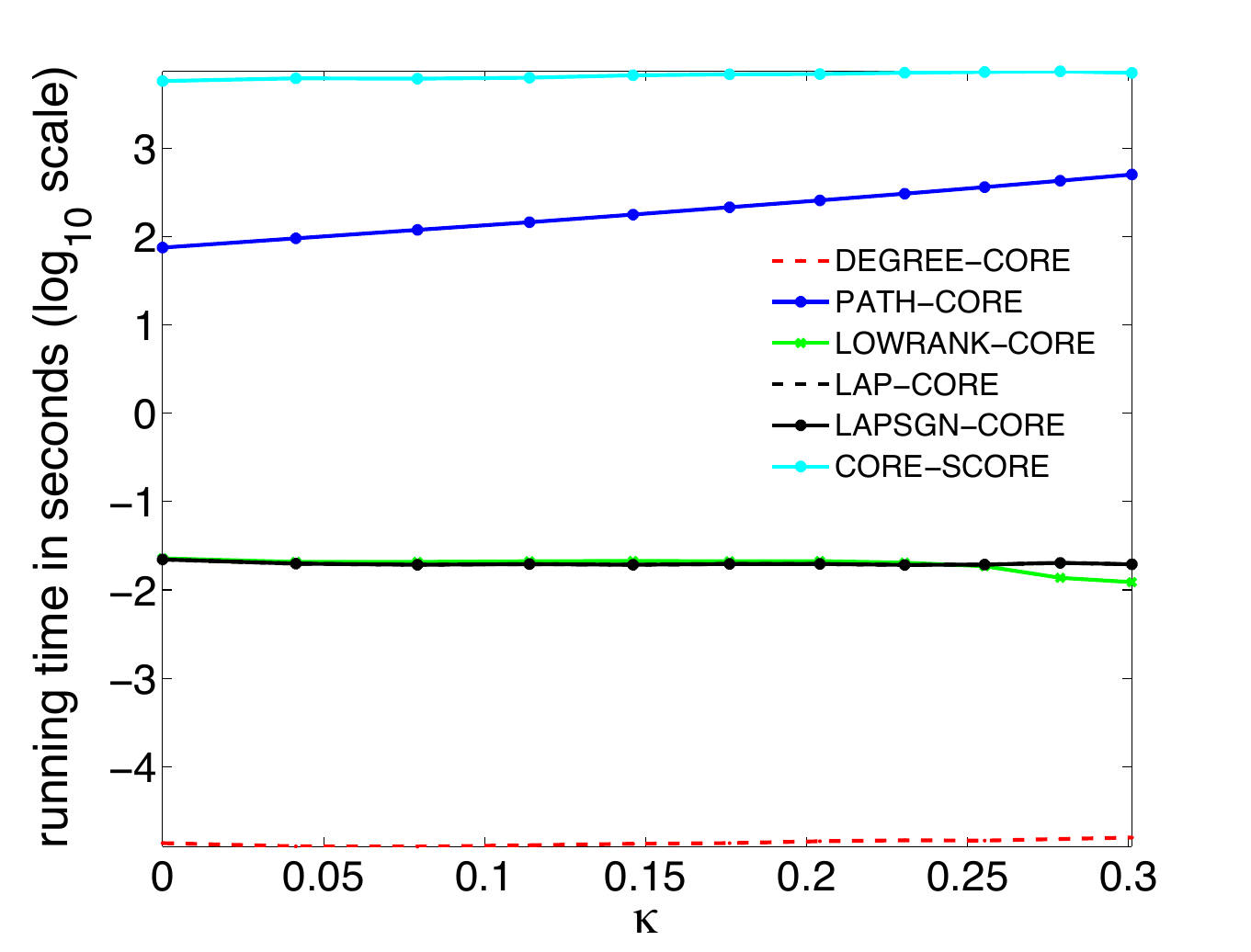}}
\subfigure[$G(p_{cc},p_{cp},p_{pp})$ with $\mathbf{p} = (p_{cc},p_{cp},p_{pp}) = (\kappa,\kappa,1-\kappa)$.]{ \includegraphics[width=0.35\textwidth]{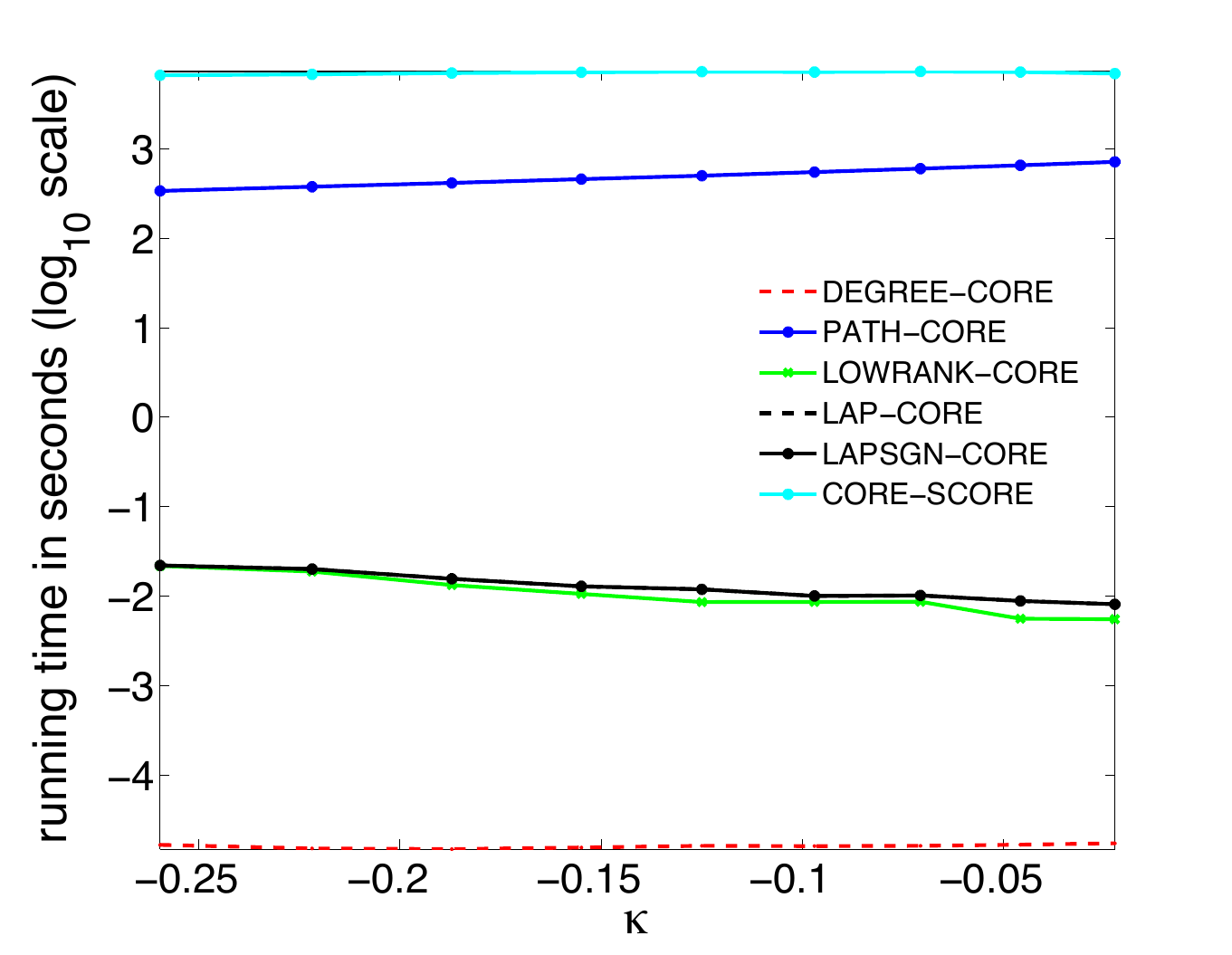}}
\caption{Comparison of the computation times (in seconds and on a $\log_{10}$ scale) for all methods and for three synthetic graph ensembles.
}
\label{fig:runningTimes}
\end{figure}


\subsection{Application to Empirical Data}
\label{sec:real_networks}

In a recent publication \cite{corePerApp}, a subset of us applied the {\sc Path-Core} and {\sc Core-Score} algorithms for detecting core--periphery structure in a variety of real-world networks. In the present paper, we use our various methods on a few other empirical data sets.

\begin{figure}[h!]
\begin{center}
\subfigure[NNS2006; boundary: 10\%]{\includegraphics[width=0.35\textwidth]{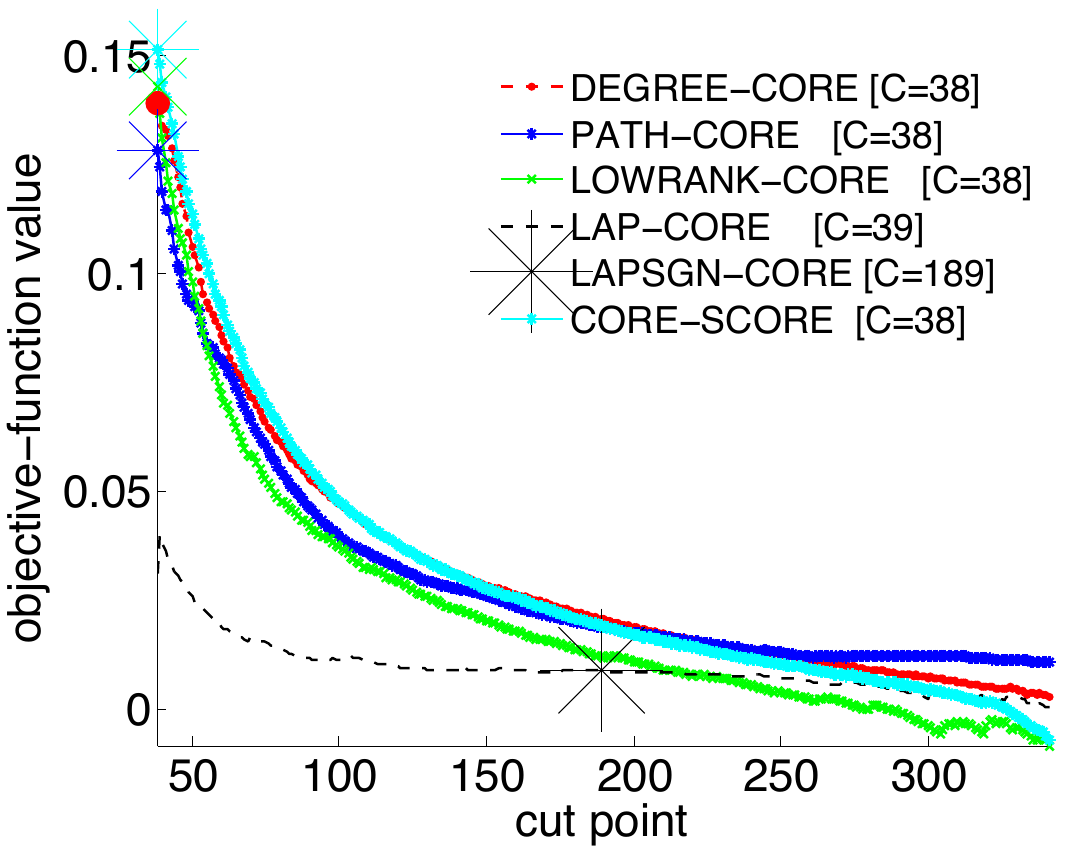}}
\subfigure[NNS2006; boundary: 20\%]{\includegraphics[width=0.35\textwidth]{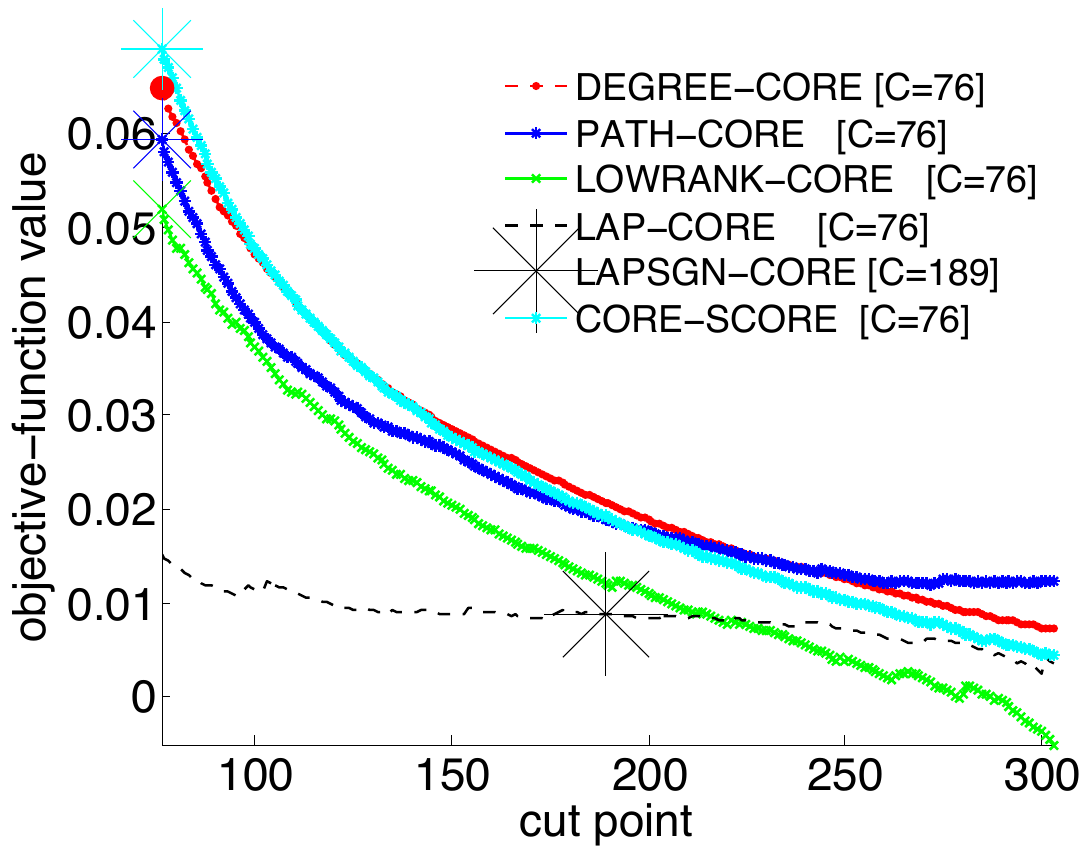}}
\subfigure[NNS2010; boundary: 10\%]{\includegraphics[width=0.35\textwidth]{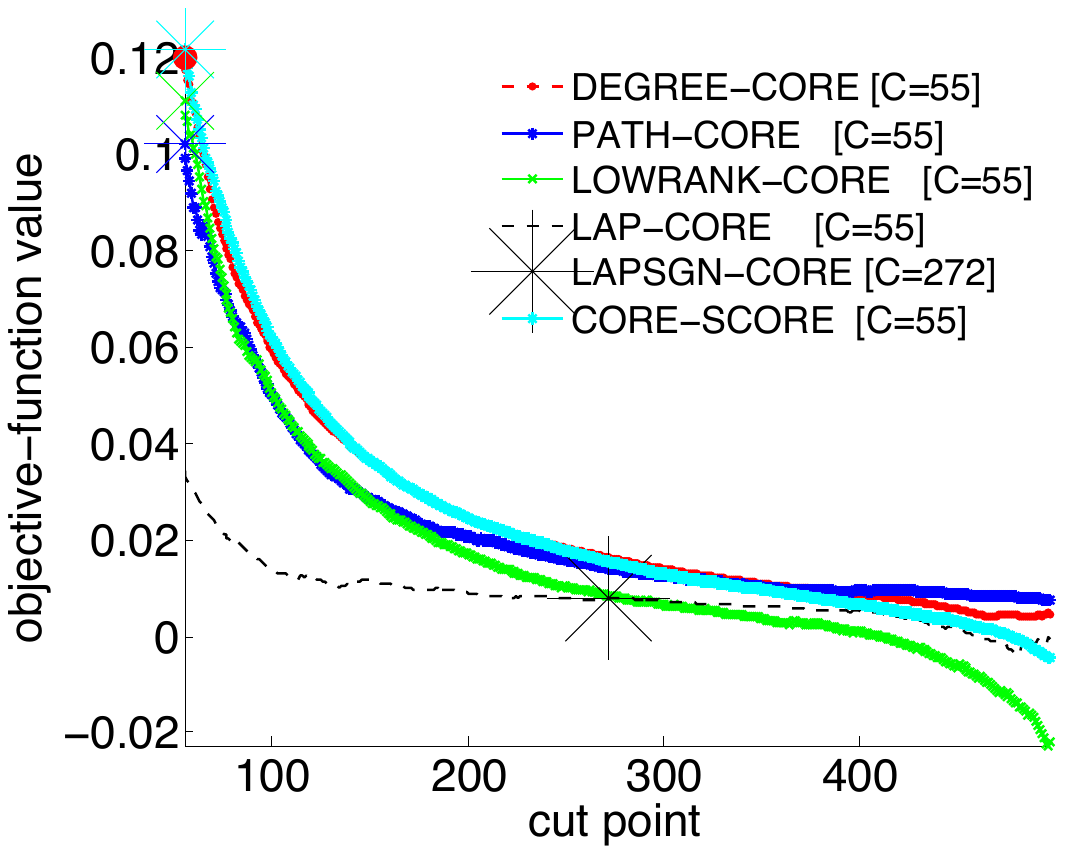}}
\subfigure[NNS2010; boundary: 20\%]{\includegraphics[width=0.35\textwidth]{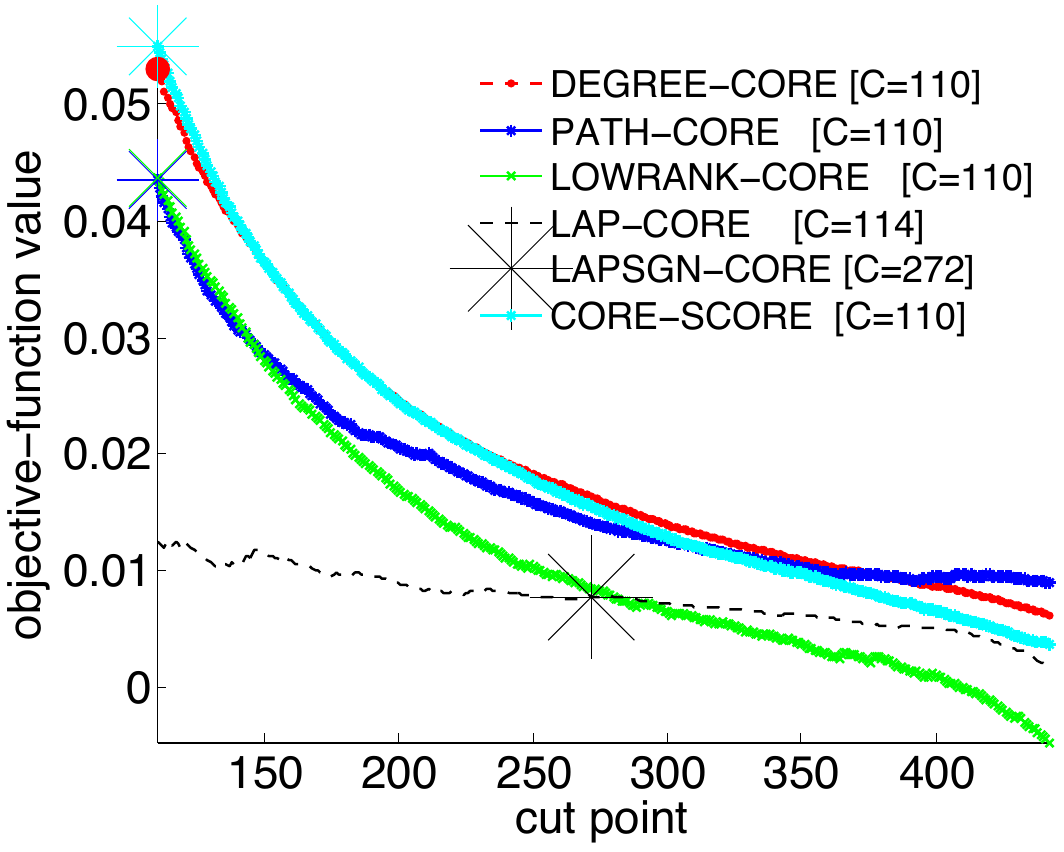}}
\end{center}
\caption{Comparison of the methods for detecting core--periphery structure for networks of network scientists in (a,b) a data set from 2006 \cite{New06} and (c,d) a data set from 2010 \cite{MapGenerator} for the objective function in Eq.~\eqref{eq:FindCut}. We assume a minimum size for the core and periphery sets of at least (a,c) 10\% of the vertices and (b,d) 20\% of the vertices. We mark the cut points that maximize the objective functions on the curves as a large asterisk for {\sc LapSgn-Core} and using other symbols whose colors match the colors of the corresponding curves for the other methods. The cut point refers to the number of core vertices, and the ${\tt C}$ values in the legends are the cut points that maximize the objective function~\eqref{eq:FindCut}. In other words, the optimal solution places ${\tt C}$ vertices in the core set. 
}
\label{fig:obj_function_nns}
\end{figure}

\begin{figure}[h!]
\begin{center}
\subfigure[Caltech; boundary: 10\%]{\includegraphics[width=0.35\textwidth]{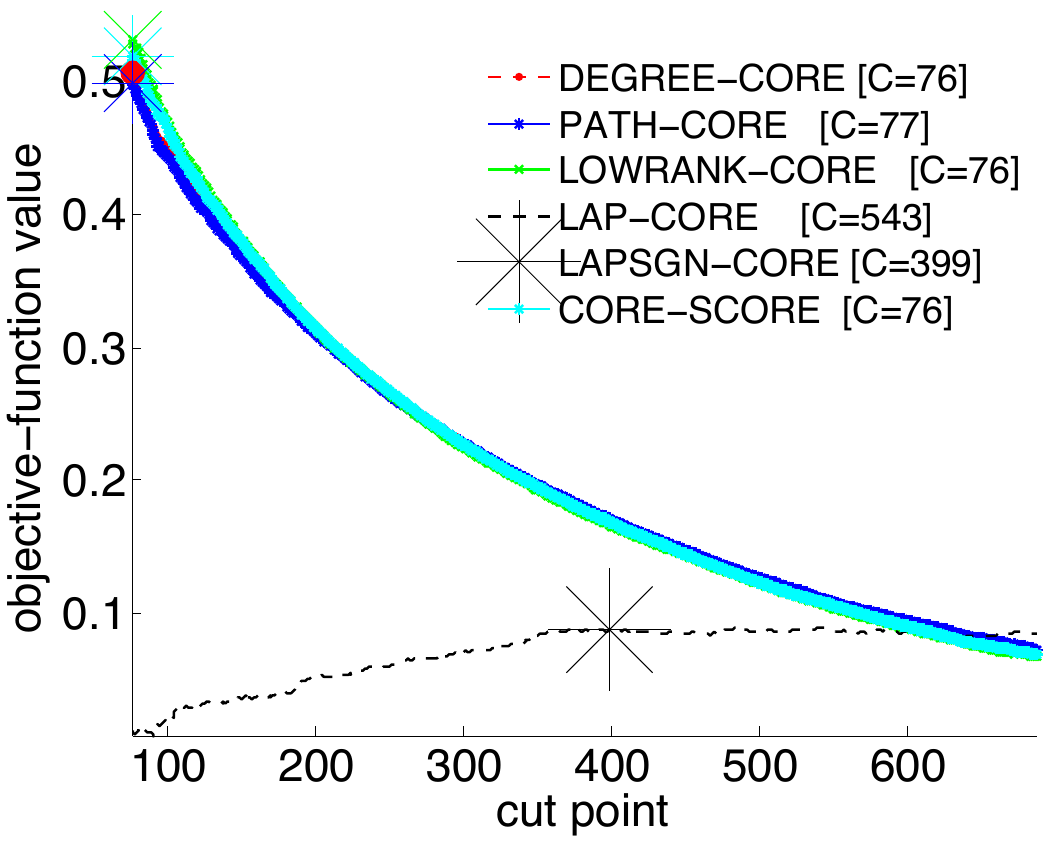}}
\subfigure[Caltech; boundary: 20\%]{\includegraphics[width=0.35\textwidth]{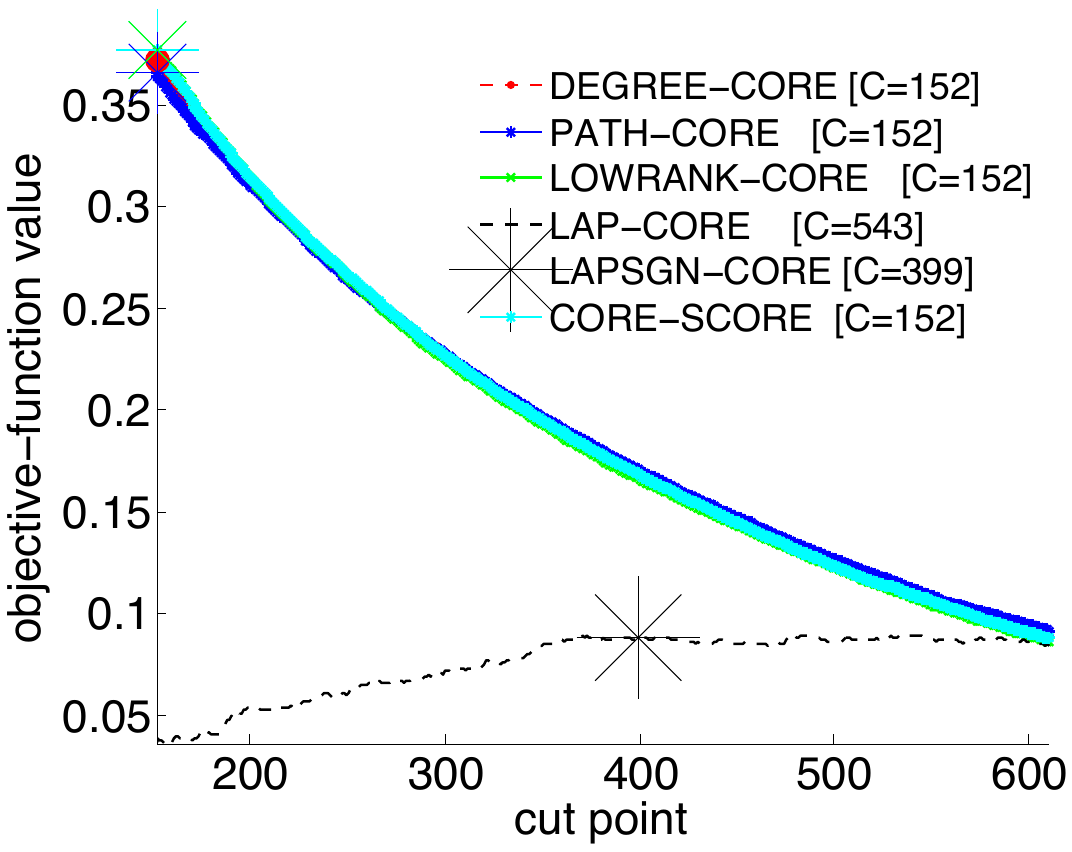}}
\subfigure[Reed; boundary: 10\%]{\includegraphics[width=0.35\textwidth]{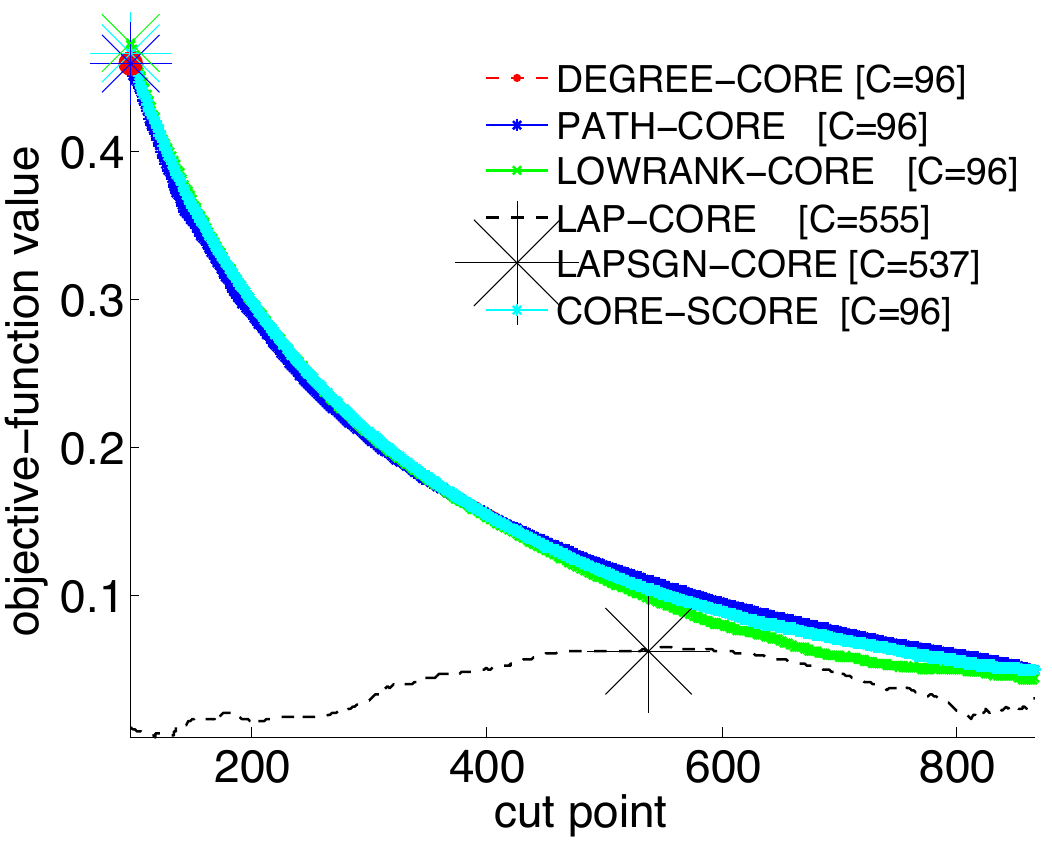}}
\subfigure[Reed; boundary: 20\%]{\includegraphics[width=0.35\textwidth]{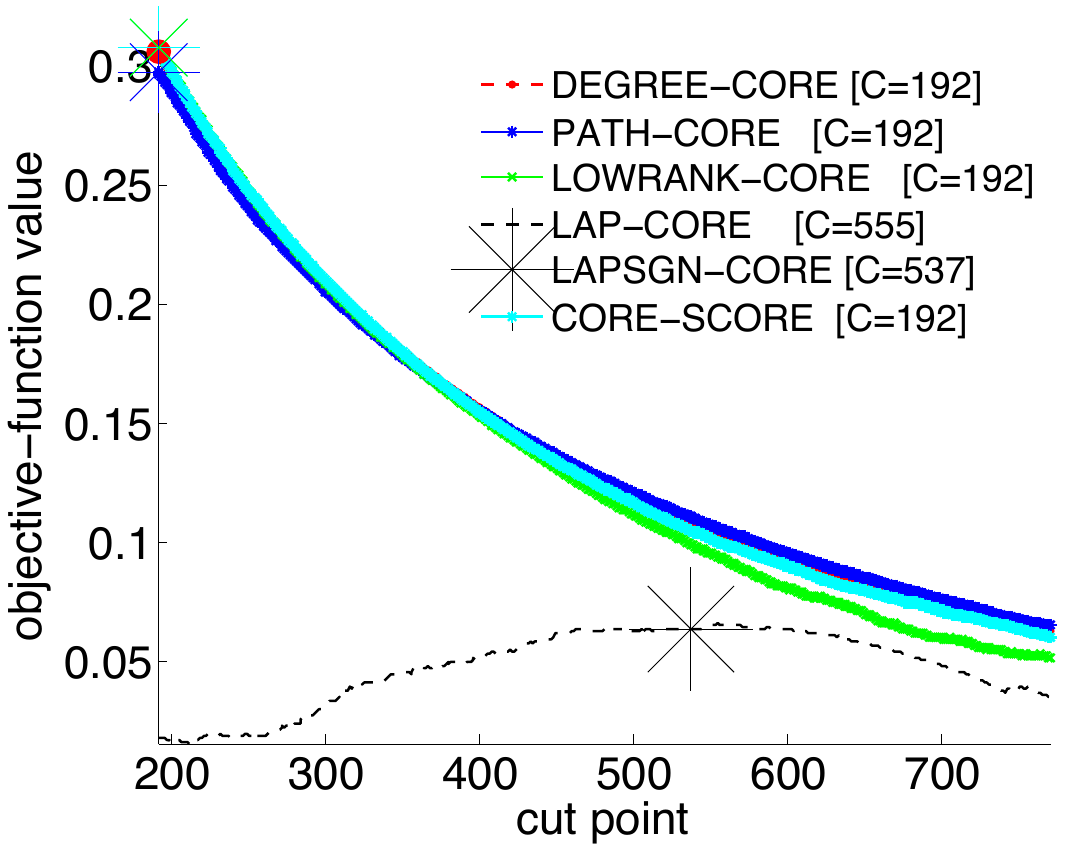}}
\end{center}
\caption{Comparison of the methods for detecting core--periphery structure for Facebook networks \cite{mason13,FacebookData} of (a,b) Caltech and (c,d) Reed College for the objective function in Eq.~\eqref{eq:FindCut}. We assume a minimum size for the core and periphery sets of at least (a,c) 10\% of the vertices and (b,d) 20\% of the vertices. We mark the cut points that maximize the objective functions on the curves as a large asterisk for {\sc LapSgn-Core} and using other symbols for the other methods. The cut point refers to the number of core vertices, and the ${\tt C}$ values in the legends are the cut point that maximizes the objective function~\eqref{eq:FindCut}. In other words, the optimal solution places ${\tt C}$ vertices in the core set.
}
\label{fig:obj_function_facebook}
\end{figure}

We consider four examples of social networks using all of the methods that we have discussed for detecting core--periphery structure. The first two graphs are publicly-available networks of network scientists from 2006 (NNS2006) \cite{New06} and 2010 (NNS2010) \cite{MapGenerator} with 379 and 552 vertices, respectively, in their largest connected components (LCCs). Reference \cite{puckmason} considered core--periphery structure in both of these networks. The vertices are scholars (predominantly from physics) who study network science, and the weight of each (undirected) edge represents the strength of a coauthorship relationship.  (See the original references for additional discussion of these networks and for more details about the weights, which are not necessarily defined in the same way in the two networks.) The other two networks are two universities (Caltech and Reed College) from the Facebook100 data set \cite{mason13,FacebookData}, which consists of a single-time snapshot from the online social network Facebook in autumn 2005 for each of 100 universities in the United States. Caltech has 762 vertices in its LCC, and Reed has 962 vertices in its LCC.

In Figs.~\ref{fig:obj_function_nns} (for the networks of network scientists) and \ref{fig:obj_function_facebook} (for the Facebook networks), we present the objective-function values \eqref{eq:FindCut} for each method for detecting core--periphery structure. 
In Table \ref{tab:correlation_real_networks} in Appendix 4, we compare the Pearson and Spearman correlations between the coreness values of the methods for these empirical networks. For these networks, we find that the values of {\sc Degree-Core}, {\sc Core-Score}, {\sc Path-Core}, and {\sc LowRank-Core} are usually strongly correlated to each other, whereas the {\sc Lap-Core} values are very different (and sometimes almost entirely uncorrelated). We find similar results when we use a similarity measure to compare partitions into a core set and periphery set from maximizing the objective function \eqref{eq:FindCut}. We compute a similarity between two measures using the expression
\begin{equation}
	S_\textrm{frac} = \frac{w_1}{w_1 + w_0}\,,
\label{def:S_frac}
\end{equation}
where $w_1$ is the number of vertices classified in the same way (i.e., either both as core vertices or both as peripheral vertices) in both measures, and $w_0$ is the number of vertices that are classified differently in the two measures. (Thus, $w_0 + w_1 = n$ is the total number of vertices.) One can also observe that the two networks of network scientists are similar to each other and that the two Facebook networks are similar to each other in terms of their correlations and core--periphery partitions. See Table \ref{tab:correlation_real_networks} in Appendix 4, and also see Figs.~\ref{fig:obj_function_nns} and \ref{fig:obj_function_facebook}. For instance, the core--periphery separation points of {\sc Lap-Core} and {\sc LapSgn-Core} yield much closer $S_\textrm{frac}$ values for Facebook networks than for the networks of network scientists.


\section{Summary and Discussion} \label{sec:future}

We introduced several new methods for detecting core--periphery structure in graphs, and we compared these methods to each other and with {\sc Core-Score} (an existing method) using both synthetic and empirical networks. 
Our approach based on transportation relies on computing shortest paths in a graph between a pair of adjacent vertices after temporarily removing the edge between the two vertices. Another approach, which is motivated by the existence of a low-rank structure in networks that exhibit core--periphery structure, relies on a low-rank approximation of the adjacency matrix of a graph. We also introduced two methods that rely on the bottom eigenvector of the random-walk Laplacian associated to a graph. Finally, we introduced an objective function that helps in the classification of vertices into core and peripheral vertices, and we showed how one can use this objective function after obtaining a vector of scores to measure coreness (using any of the above methods). Core--periphery structure is a common feature of real-world networks, and it is important to continue to develop methods to detect it and to compare the performance of such methods against each other on a wide variety of networks. We have introduced and explored the performance of several new methods in this paper. The different methods that we introduce are based on rather different ideas, and it is very important to explore core--periphery structure from a multitude of perspectives.

Given the common use of $k$-cores in the consideration of core parts of networks, it is also interesting to examine the assignment of vertices into core and periphery sets based only on vertex degrees. Although using vertex degree as a measure of centrality or likelihood of belonging to a core can often produce inaccurate results \cite{puckmason}, it can sometimes be true that a degree-based classification of vertices as core vertices or peripheral vertices should be successful for certain random-graph ensembles (and certain empirical networks) \cite{XZhang2014}. One can thus ask what properties such ensembles ought to have. More generally, it is also important to compare coreness scores with other centrality measures \cite{barucca2015,puckmason,corePerApp}. Another interesting question is whether one can use current methods for solving the group-synchronization problem (such as the eigenvector method and semidefinite programming \cite{sync,Goemans_Williamson,gb08}) for the detection of core--periphery structure in various families in networks.

An important future application is to examine core--periphery structure in temporal and multilayer networks \cite{holme12,mason12,kivela2014}. Community structure (see, e.g., \cite{mason12,jeub2015b}) has been studied in such contexts, and it should also be very insightful to also consider core--periphery structure in multilayer networks.  Another interesting direction is developing additional objective functions with which to classify vertices into core and periphery sets.  

Networks have many different types of mesoscale structures. In most research thus far, community structure has taken center stage. Other mesoscale structures, such as role assignment \cite{rossi2014} and core--periphery structure \cite{cp-review}, are also very important. These ideas are worthy of considerably more exploration.


\section*{Acknowledgements}

S.H.L. and M.A.P. were supported by a grant (EP/J001795/1) from the Engineering and Physical Sciences Research Council (EPSRC), and M.A.P. and P.R. were supported by the James S. McDonnell Foundation (\#220020177). S.H.L. was supported by Basic Science Research Program through the National Research Foundation of Korea (NRF) funded by the Ministry of Education (2013R1A1A2011947).
M.C. thanks Radek Erban and OCCAM at University of Oxford for their warm hospitality while hosting him for two months during Spring 2012 (during which this project was initiated) and is grateful to Amit Singer for his guidance and support via Award Number R01GM090200 from the NIGMS and Award Number FA9550-09-1-0551 from AFOSR. M.C. and P.R. also acknowledge support from AFOSR MURI grant FA9550-10-1-0569, ONR grant N000141210040, and ARO MURI grant W911NF-11-1-0332. Part of this work was undertaken while M.C. and P.R. were attending the Semester Program on Network Science and Graph Algorithms at the Institute for Computational and Experimental Research in Mathematics (ICERM) at Brown University.



\clearpage

\section*{Appendix 1: Algorithm for Computing Path-Core scores}
\label{sec:appendix_PathCore}

Let $G(V,E)$ be an unweighted graph without self-edges or multi-edges (i.e., it is a simple graph). Recall that we define the {\sc Path-Core} score (\ref{def:PathCoreCentrality}) of a vertex $i\in V$ as the sum over all adjacent vertex pairs in $G$ of the fraction of shortest nontrivial paths containing $i$ between each vertex pair in $V(G)\setminus  i$. By ``nontrivial,'' we mean that the direct edge between those adjacent vertices does not count as a path. 
Our algorithm has strong similarities to the algorithm presented in~\cite{brandes}, and we follow some of the notation introduced therein. Let $d_G(j,i)$ be the ``distance'' between vertices $j$ and $i$; we define this distance as the minimum length of any path that connects $j$ and $i$ in $G$. Let $\sigma_{st}(i)$ be the number of shortest paths between $s$ and $t$ that contain $i$. Define the set of \emph{predecessors} of a vertex $i$ on shortest paths from $s$ as 
\begin{equation*}
	P_s(i)=\{ j \in V : (j,i)  \in E\,, d_G(s,i)=d_G(s,j)+1 \}\,.
\end{equation*}

We use the following observation: if $i$ lies on a shortest path between $s$ and $t$, then 
\begin{equation*}
	\sigma_{st}(i)=\left( \sum_{k \in P_s(i)} \sigma_{si}(k) \right) \times \left( \sum_{l \in P_t(i)} \sigma_{it}(l) \right)\,.
\end{equation*}	
This will help us count the number of shortest paths on which a vertex lies without keeping track of the locations of these shortest paths. In the {\sc Path-Score} algorithm, $\sigma_s(i)$ is the number of paths between $s$ and $i$ of length $d_{G'}(s,i)$ if and only if $i$ lies on a shortest path between $s$ and $t$ (i.e., if $(s,t)$ is the edge that is currently removed), where $G'$ is the graph $G\setminus (s,t)$. The algorithm records the distance between $s$ and $i$ in $G'$ as $d_s(i)$. In Algorithm \ref{pathscorepuck}, we calculate {\sc Path-Core} scores for every vertex.

\begin{algorithm}
\caption{{\sc Path-Core}: Computes {\sc Path-Core} scores for all vertices of a graph $G$.}
\label{pathscorepuck}
  \begin{multicols}{2}
\begin{algorithmic}[1]
\footnotesize
\REQUIRE $G$
\ENSURE $C_P$
\item[]
\STATE $C_P(w) \leftarrow 0, \, w \in V;$
\item[]
\FOR{$ (s,t )\in E(G) $}
\item[]
\STATE $G' \leftarrow G \setminus ( s,t ) ; $\label{startit}
\STATE $\sigma_s(w), \sigma_t(w)   \leftarrow 0,  v \in V ;$ \STATE $\sigma_s(s), \sigma_t(t) \leftarrow 1;$
\STATE $d_s(w),d_t(w)\leftarrow -1, \, v \in V ;$
\STATE$ d_s(s),d_t(t) \leftarrow 0;$
\STATE $Q \leftarrow$ empty queue;
\item[]
\STATE enqueue $ s \rightarrow Q;$
\WHILE{$Q$ not empty}
\STATE dequeue $w \leftarrow Q;$
\FOR{each $u \in \Gamma_{G'} (w)$}
\IF{$d_s(u)<0$}
\STATE enqueue $ u \rightarrow Q;$
\STATE $d_s(u) \leftarrow d_s(w)+1;$
\ENDIF

\ENDFOR
\ENDWHILE
\item[]
\STATE enqueue $ t \rightarrow Q;$
\WHILE{$Q$ not empty}
\STATE dequeue $w \leftarrow Q;$
\FOR{each $u \in \Gamma_{G'} (w)$}
\IF{$d_t(u)<0$}
\STATE enqueue $ u \rightarrow Q;$
\STATE $d_t(u) \leftarrow d_t(w)+1;$
\ENDIF
\IF{ $d_s(u)<d_s(w)$}
\STATE $\sigma_t(u)=\sigma_t(u)+\sigma_t(w);$
\ENDIF
\ENDFOR
\ENDWHILE
\item[]
\STATE enqueue $ s \rightarrow Q;$
\WHILE{$Q$ not empty}
\STATE dequeue $w \leftarrow Q;$
\FOR{each $u \in \Gamma_{G'} (w)$}
\IF{$d_t(u)<d_t(w)$}
\STATE enqueue $ u \rightarrow Q;$
\STATE $\sigma_s(u)=\sigma_s(u)+\sigma_s(w);$
\ENDIF
\ENDFOR
\ENDWHILE
\item[]
\FOR{$w \in V \setminus (s,t)$}
\STATE $C_P(w)=C_P(w)+ \sigma_s(w) \cdot \sigma_t(w) / \sigma_s(t);$\label{endit}
\ENDFOR
\item[]
\ENDFOR
\end{algorithmic}  \end{multicols}
\end{algorithm}

\normalsize
\begin{lemma}
Algorithm \ref{pathscorepuck} outputs the {\sc Path-Core} scores for all vertices in an unweighted graph $G$.
\end{lemma}

\begin{proof}
It suffices to show for one edge $(s,t) \in E(G)$ and one iteration (i.e., lines \ref{startit}--\ref{endit}) that the algorithm counts, for each vertex $w \in V(G) \setminus (s,t)$, the number of shortest paths between $s$ and $t$ that contain $w$. This number $\sigma_{s,t}(w)$ is given by the algorithm as $\sigma_s(w) \cdot \sigma_t(w)$. In this case, $\sigma_s(w)$ is the number of paths between $s$ and $w$ of length $d_{G'}(s,w)$ (where the graph $G' = G \setminus (s,t)$) if and only if $w$ lies on a shortest path between $s$ and $t$. 

Algorithm \ref{pathscorepuck} performs three breadth-first-searches (BFSs). In the first BFS, it searches from vertex $s$ and records the distances from $s$ to all other vertices. It then performs a BFS starting from vertex $t$. During this second BFS, it records the distances to all vertices from $t$, and it also records $\sigma_t(w)$ for vertices that lie on a shortest path between $s$ and $t$. The {\sc Path-Score} algorithm knows that $u$ lies on a shortest path between $s$ and $t$ if it has a distance from $s$ that is less than the distance from $s$ of its predecessor in the BFS from $t$. In other words, if $d_t(w)<d_t(u)$, then an edge $(w,u)$ lies on a shortest path between $s$ and $t$ if and only if $d_s(u)<d_s(w)$. Additionally,
\begin{equation*} 
	\sigma_t(u)=\sum_{w \in P_t(u) } \sigma_t(w)\,.
\end{equation*}	

In the second BFS, Algorithm \ref{pathscorepuck} finds a vertex $u$ exactly once for each of its predecessors $w\in P_t(u)$, and it adds $\sigma_t(w)$ to $\sigma_t(u)$. Therefore, in the second BFS, for each vertex $u \in V(G) \setminus (s,t)$, {\sc Path-Score} records $\sigma_t(v)$ as the number of shortest paths from $t$ to $u$ if $u$ is on a shortest path between $s$ and $t$. If it is not, then $\sigma_t(u)$ is still $0$. 

By the same arguments, in the third BFS, for each vertex $u \in V(G) \setminus (s,t)$, {\sc Path-Core} records $\sigma_s(u)$ as the number of shortest paths from $s$ to $u$ if $u$ is on a shortest path between $s$ and $t$. If it is not, then $\sigma_s(u)$ is still $0$.

It should now be clear that for all $w \in V(G) \setminus (s,t)$, it follows that $\sigma_s(w) \cdot \sigma_t(w)$ yields $\sigma_{s,t}(w)$.
\end{proof}

\vspace{.2 in}

\begin{lemma}
Algorithm \ref{pathscorepuck} finishes in $\mathcal{O}(m^2)$ time.
\end{lemma}
\begin{proof}
Algorithm \ref{pathscorepuck} iterates (i.e., it runs lines \ref{startit}--\ref{endit}) once for each edge. In one iteration, it performs three BFSs. During a BFS, every edge of $G'$ is considered exactly once; this is an $\mathcal{O}(1)$ time procedure. Therefore, every iteration of {\sc Path-Core} runs in $\mathcal{O}(m)$ time, and the temporal complexity of {\sc Path-Core} is $\mathcal{O}(m^2)$.
\end{proof}

\vspace{.2 in}

For weighted graphs, one can implement an algorithm that is very similar to Algorithm \ref{pathscorepuck}. This algorithm uses Dijkstra's algorithm for shortest paths instead of BFS, and it runs in $\mathcal{O}(m+n \log n)$ time instead of $\mathcal{O}(m)$, so the total temporal complexity becomes $\mathcal{O}(m^2 + mn \log n)$.


\section*{Appendix 2: Symmetry in the Random-Walk Laplacian}

We now show that a symmetry relation like $\eqref{KelMans}$ exists for the random-walk Laplacian associated to an unweighted graph only under certain conditions. Additionally, the most obvious version of such a statement does not hold. To see this, let ${\bf x}$ be an eigenvector of $\bar{L}$ (which is nontrivial, so ${\bf x} \perp {\bf \mb{1}}_{n}$). We use the notation $\bar{D} =\mathrm{diag}(n-1-d_i)$, where $d_i$ denotes the degree of vertex $i$, and calculate
\begin{align}\label{above}
	\bar{L}	&=  \bar{D}^{-1}  \bar{A} \notag \\
			&=  \bar{D}^{-1}   ( J_n - A - I_n)  \notag \\
			&= \bar{D}^{-1}   ( J_n - I_n)  - \bar{D} ^{-1 }  A \notag \\
			&= \bar{D}^{-1}   ( J_n - I_n)  - \bar{D} ^{-1 }  D D^{-1} A \notag \\
			&= \bar{D}^{-1}   ( J_n - I_n)  - \bar{D} ^{-1 }  D  L \,.
\end{align}
Because $\bar{D}^{-1 } = \mathrm{diag}\left(\frac{1}{n-1-d_i}\right)$  and $D = \mathrm{diag}(d_i)$ are diagonal matrices, it follows that $\bar{D}^{-1 } D = \mathrm{diag}\left(\frac{d_i}{n-1-d_i}\right) $.
Given an eigenvector ${\bf x}$ of $\bar{L}$, we obtain $\bar{L} {\bf x} = \bar{\lambda} {\bf x}$ for some eigenvalue $\bar{\lambda}$. Because $\bar{L}$ is a row-stochastic matrix, it has the trivial eigenvalue $\bar{\lambda}_1=1$ with associated eigenvector $\bar{{\bf v}}_{1} ={\bf \mb{1}}_{n}$. We apply both sides of Eq.~\eqref{above} to the eigenvector ${\bf x}$ and note that $J_n  {\bf x} ={\bf 0}$ because $ {\bf x} \perp {\bf \bar{v}}_{1} = {\bf \mb{1}}_{n}$. We thereby obtain
\begin{align}\label{above2}
	\bar{\lambda} {\bf x}      &=    \bar{L} {\bf x}	\notag \\
					 &=   \bar{D}^{-1}   ( J_n - I_n) {\bf x}  - \bar{D} ^{-1 }  D  L {\bf x} \notag \\
					 &= - \bar{D}^{-1}   I_n  {\bf x}  - \bar{D} ^{-1 }  D  L {\bf x}\,. 
\end{align}
Multiplying both sides of Eq.~\eqref{above2} by $D^{-1 }  \bar{D}$ on the left yields
\begin{equation}
	D^{-1 }  \bar{D}     \bar{\lambda} {\bf x}   = -  D^{-1 }  \bar{D}     \bar{D}^{-1}    {\bf x}  - L {\bf x}\,,
\end{equation}
so
\begin{align*}
	 L {\bf x}    &   =  - D^{-1 }  \bar{D}     \bar{\lambda} {\bf x}    -  D^{-1 }   {\bf x}  \\
	         &   =  - (   \bar{D}  \bar{\lambda}    +  I_n) D^{-1 }  {\bf x} \\
	         &   =  -   \mathrm{diag}\left( \frac{ \bar{\lambda} (n-1-d_i) +1}{d_i} \right)  {\bf x} \,. \\
\end{align*}
Therefore, ${\bf x}$ is not an eigenvector of $L$ unless
\begin{equation*}
	\frac{ \bar{\lambda} (n-1-d_i) +1}{d_i} = \theta  \quad \text{for all} \quad i\in \{1,\dots,n\}
\end{equation*}
for some constant $\theta$. In other words, $d_i = \frac{\bar{\lambda} (n-1) }{ \theta + \bar{\lambda}}$ for  $i \in \{1,\dots,n\}$, so the graph is $d$-regular, with 
$ d_1=d_2=\dots = d_n = d = \frac{\bar{\lambda} (n-1) }{ \theta + \bar{\lambda}}$. Therefore,
\begin{equation*}
	\theta = \frac{ \bar{\lambda} (n-1-d) +1}{d}\,,
\end{equation*}	
so the eigenvector ${\bf x}$ of $\bar{L}$ is also an eigenvector of $L$ (with a corresponding eigenvalue of $-\theta$).


\section*{Appendix 3: Planted High-Degree Vertices}

To illustrate the sensitivity of the {\sc Degree-Core} method to the presence of high-degree peripheral vertices, we conduct a numerical experiment in which we intentionally plant high-degree vertices in the periphery set.  This helps illustrate that it is dangerous to use methods like $k$-core decomposition (which has very strong demands that vertices have a high degree to be construed as core vertices) to study core--periphery structure \cite{cp-review}. In Fig.~\ref{fig:exPlantedDeg}, we consider a graph from the ensemble $G(p_{cc},p_{cp},p_{pp},n_c,n_p)$ with $n=100$ vertices, edge probabilities $(p_{cc}, p_{cp}, p_{pp}) = (0.4,0.4,0.2)$, $n_c$ core vertices, $n_p$ peripheral vertices (with $n=n_c + n_p$), and planted high-degree vertices in the periphery set.  To perturb the graph $G$ from the above ensemble to plant high-degree peripheral vertices, we proceed as follows. First, we select each peripheral vertex with independent probability $0.1$. Second, we connect each such vertex to $15$ non-neighboring peripheral vertices that we choose uniformly at random. In the left panel of Fig.~\ref{fig:exPlantedDeg}, we show an example with a boundary size of $10\%$, so we are assuming that the core and periphery sets each have at least $0.1n = 10$ vertices. We the search for a cut point in the interval $[10,90]$. In the right panel, we consider a larger boundary size and assume that the core and the periphery sets each have at least $25$ vertices. We now search for an optimal cut in the interval $[25,75]$. In the two planted-degree scenarios for which the size of the core set is unknown, all methods yield many misclassified vertices, although the {\sc LapSgn-Core} method has the lowest number (28) of misclassifications in both cases.

\begin{figure}[h!]
\begin{center}
\includegraphics[width=0.38\textwidth]{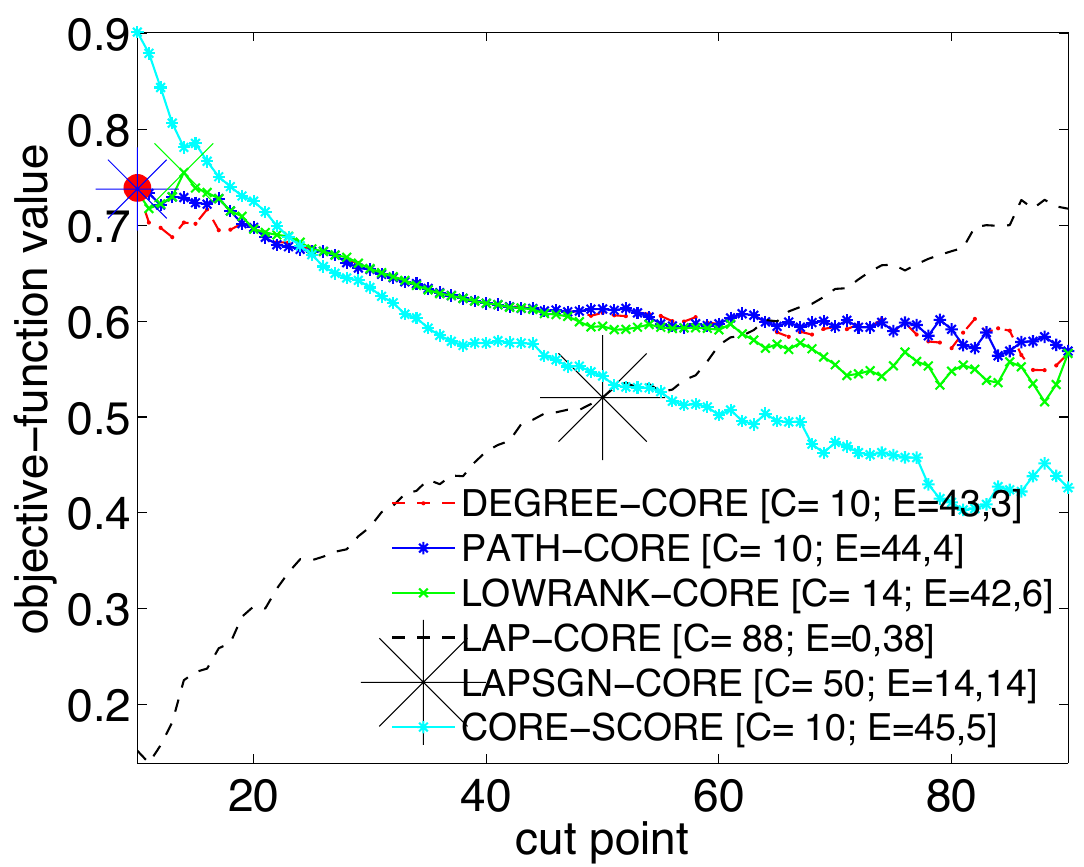}
\includegraphics[width=0.38\textwidth]{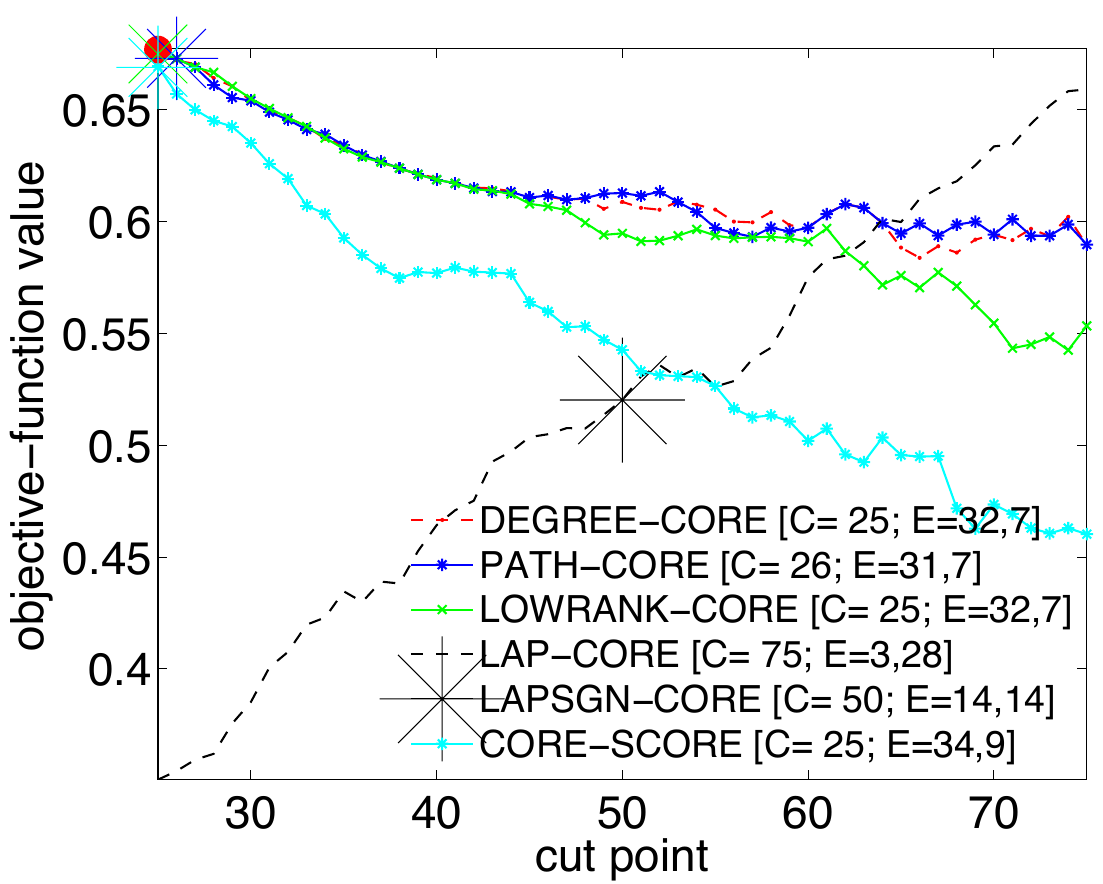}
\end{center}
\caption{Comparison of the methods for one realization of the graph ensemble $G(p_{cc},p_{cp},p_{pp},n, n_c, n_p)$ with $n=100$ vertices, $n_c=50$ core vertices, $n_p=50$ peripheral vertices, edge probabilities $(p_{cc}, p_{cp}, p_{pp}) = (0.4,0.4,0.2)$, and planted high-degree vertices for the objective function in Eq.~\eqref{eq:FindCut}. The cut point refers to the number of core vertices. In the legends, ${\tt C}$ denotes the size of the core set that maximizes the objective function Eq.~\eqref{eq:FindCut}, and ${\bf E} = (y_1,y_2)$ denotes the corresponding $2$-vector of errors. The first component of ${\bf E}$ indicates the number of core vertices that we label as peripheral vertices, and the second indicates the number of peripheral vertices that we label as core vertices. In this graph, each peripheral vertex has a probability of $0.1$ of becoming adjacent to $15$ additional non-neighboring peripheral vertices that we select uniformly at random. We mark the cut points that maximize the objective functions on the curves as a large asterisk for {\sc LapSgn-Core} and using other symbols whose colors match the colors of the corresponding curves for the other methods. The cut point refers to the number of core vertices.
}
\label{fig:exPlantedDeg}
\end{figure}


\section*{Appendix 4: Correlations Between Coreness Values from Different Methods}

In Table \ref{tab:correlation_real_networks}, we consider several empirical networks and examine the numerical values of the Pearson and Spearman correlations between the coreness values that we obtain for the core--periphery detection methods that we examine.

\begin{table}  
\begin{minipage}[b]{0.99\linewidth}
\centering
\begin{tabular}{c|cccccccccc}
\hline
\emph{NNS2006} & $q,\mathcal{C}$ & $q,\mathcal{P}$ & $q,\mathcal{R}$ & $q,\mathcal{L}$ & $\mathcal{C},\mathcal{P}$ & $\mathcal{C},\mathcal{R}$ & $\mathcal{C},\mathcal{L}$ & $\mathcal{P},\mathcal{R}$ & $\mathcal{P},\mathcal{L}$ & $\mathcal{R},\mathcal{L}$ \\
Pearson & $0.79^\ast$ & $0.89^\ast$ & $0.72^\ast$ & $0.02$ & $0.64^\ast$ & $0.56^\ast$ & $0.03$ & $0.62^\ast$ & $0.03$ & $-0.01$ \\
Spearman & $0.79^\ast$ & $0.62^\ast$ & $0.43^\ast$ & $0.04$ & $0.37^\ast$ & $0.65^\ast$ & $-0.05$ & $0.14^\ast$ & $-0.01$ & $0.01$ \\
$S_\textrm{frac}$ & $0.93$ & $0.90$ & $0.79$ & $0.69$ & $0.88$ & $0.82$ & $0.69$ & $0.78$ & $0.69$ & $0.67$ \\
($z$-score) & $(19.5^\dagger)$ & $(17.4^\dagger)$ & $(8.2^\dagger)$ & $(1.1)$ & $(15.7^\dagger)$ & $(10.2^\dagger)$ & $(1.1)$ & $(7.5^\dagger)$ & $(0.7)$ & $(-0.9)$ \\
\cline{2-11}
 & $q,\mathcal{LS}$ & $\mathcal{C},\mathcal{LS}$ & $\mathcal{P},\mathcal{LS}$ & $\mathcal{R},\mathcal{LS}$ & $\mathcal{L},\mathcal{LS}$ & & & & & \\
$S_\textrm{frac}$ & $0.50$ & $0.47$ & $0.47$ & $0.44$ & $0.70$ & & & & & \\
($z$-score) & $(0.1)$ & $(-1.5)$ & $(-1.9)$ & $(-3.4^\dagger)$ & $(11.9^\dagger)$ & & & & & \\
\hline
\emph{NNS2010} & $q,\mathcal{C}$ & $q,\mathcal{P}$ & $q,\mathcal{R}$ & $q,\mathcal{L}$ & $\mathcal{C},\mathcal{P}$ & $\mathcal{C},\mathcal{R}$ & $\mathcal{C},\mathcal{L}$ & $\mathcal{P},\mathcal{R}$ & $\mathcal{P},\mathcal{L}$ & $\mathcal{R},\mathcal{L}$ \\
Pearson & $0.78^\ast$ & $0.84^\ast$ & $0.71^\ast$ & $0.01$ & $0.62^\ast$ & $0.46^\ast$ & $0.01$ & $0.56^\ast$ & $0.02$ & $0.01$ \\
Spearman & $0.84^\ast$ & $0.56^\ast$ & $0.39^\ast$ & $0.10$ & $0.38^\ast$ & $0.56^\ast$ & $0.04$ & $0.17^\ast$ & $0.08$ & $0.03$ \\
$S_\textrm{frac}$ & $0.96$ & $0.88$ & $0.80$ & $0.71$ & $0.87$ & $0.82$ & $0.72$ & $0.76$ & $0.70$ & $0.75$ \\
($z$-score) & $(29.1^\dagger)$ & $(20.2^\dagger)$ & $(12.4^\dagger)$ & $(3.7^\dagger)$ & $(19.8^\dagger)$ & $(14.0^\dagger)$ & $(4.1^\dagger)$ & $(8.7^\dagger)$ & $(1.9^\dagger)$ & $(7.8^\dagger)$ \\
\cline{2-11}
 & $q,\mathcal{LS}$ & $\mathcal{C},\mathcal{LS}$ & $\mathcal{P},\mathcal{LS}$ & $\mathcal{R},\mathcal{LS}$ & $\mathcal{L},\mathcal{LS}$ & & & & & \\
$S_\textrm{frac}$ & $0.54$ & $0.52$ & $0.50$ & $0.51$ & $0.71$ & & & & & \\
($z$-score) & $(2.7^\dagger)$ & $(1.5)$ & $(-0.1)$ & $(0.9)$ & $(16.9^\dagger)$ & & & & & \\
\hline
\emph{FB-Caltech} & $q,\mathcal{C}$ & $q,\mathcal{P}$ & $q,\mathcal{R}$ & $q,\mathcal{L}$ & $\mathcal{C},\mathcal{P}$ & $\mathcal{C},\mathcal{R}$ & $\mathcal{C},\mathcal{L}$ & $\mathcal{P},\mathcal{R}$ & $\mathcal{P},\mathcal{L}$ & $\mathcal{R},\mathcal{L}$ \\
Pearson & $0.96^\ast$ & $0.97^\ast$ & $0.98^\ast$ & $0.02$ & $0.86^\ast$ & $0.97^\ast$ & $0.01$ & $0.93^\ast$ & $0.01$ & $0.02$ \\
Spearman & $1.00^\ast$ & $0.99^\ast$ & $0.99^\ast$ & $0.09$ & $0.98^\ast$ & $1.00^\ast$ & $0.08$ & $0.97^\ast$ & $0.09^\ast$ & $0.07$ \\
$S_\textrm{frac}$ & $0.98$ & $0.97$ & $0.97$ & $0.43$ & $0.97$ & $0.99$ & $0.42$ & $0.96$ & $0.42$ & $0.42$ \\
($z$-score) & $(38.4^\dagger)$ & $(36.9^\dagger)$ & $(36.7^\dagger)$ & $(5.7^\dagger)$ & $(35.5^\dagger)$ & $(38.4^\dagger)$ & $(5.0^\dagger)$ & $(34.4^\dagger)$ & $(5.5^\dagger)$ & $(4.7^\dagger)$ \\
\cline{2-11}
 & $q,\mathcal{LS}$ & $\mathcal{C},\mathcal{LS}$ & $\mathcal{P},\mathcal{LS}$ & $\mathcal{R},\mathcal{LS}$ & $\mathcal{L},\mathcal{LS}$ & & & & & \\
$S_\textrm{frac}$ & $0.53$ & $0.53$ & $0.53$ & $0.53$ & $0.81$ & & & & & \\
($z$-score) & $(4.7^\dagger)$ & $(4.2^\dagger)$ & $(4.4^\dagger)$ & $(4.1^\dagger)$ & $(25.7^\dagger)$ & & & & & \\
\hline
\emph{FB-Reed} & $q,\mathcal{C}$ & $q,\mathcal{P}$ & $q,\mathcal{R}$ & $q,\mathcal{L}$ & $\mathcal{C},\mathcal{P}$ & $\mathcal{C},\mathcal{R}$ & $\mathcal{C},\mathcal{L}$ & $\mathcal{P},\mathcal{R}$ & $\mathcal{P},\mathcal{L}$ & $\mathcal{R},\mathcal{L}$ \\
Pearson & $0.92^\ast$ & $0.95^\ast$ & $0.98^\ast$ & $-0.01$ & $0.77^\ast$ & $0.94^\ast$ & $-0.02$ & $0.90^\ast$ & $-0.01$ & $-0.01$ \\
Spearman & $0.99^\ast$ & $0.98^\ast$ & $0.96^\ast$ & $0.07$ & $0.96^\ast$ & $0.98^\ast$ & $0.07$ & $0.90^\ast$ & $0.09^\ast$ & $0.05$ \\
$S_\textrm{frac}$ & $0.99$ & $0.98$ & $0.98$ & $0.51$ & $0.97$ & $0.99$ & $0.51$ & $0.96$ & $0.51$ & $0.51$ \\
($z$-score) & $(41.8^\dagger)$ & $(39.7^\dagger)$ & $(40.2^\dagger)$ & $(6.3^\dagger)$ & $(39.1^\dagger)$ & $(42.0^\dagger)$ & $(6.5^\dagger)$ & $(38.2^\dagger)$ & $(6.5^\dagger)$ & $(6.3^\dagger)$ \\
\cline{2-11}
 & $q,\mathcal{LS}$ & $\mathcal{C},\mathcal{LS}$ & $\mathcal{P},\mathcal{LS}$ & $\mathcal{R},\mathcal{LS}$ & $\mathcal{L},\mathcal{LS}$ & & & & & \\
$S_\textrm{frac}$ & $0.51$ & $0.52$ & $0.52$ & $0.52$ & $0.98$ & & & & & \\
($z$-score) & $(5.4^\dagger)$ & $(5.9^\dagger)$ & $(5.6^\dagger)$ & $(5.7^\dagger)$ & $(42.5^\dagger)$ & & & & & \\
\hline
\end{tabular}
\caption{Pearson and Spearman correlation coefficients for various coreness measures and the similarity measure $S_\textrm{frac}$ for core--periphery partitioning with a boundary of 20\% of the vertices (see the right panels in Figs.~\ref{fig:obj_function_nns} and \ref{fig:obj_function_facebook}) between the objective function in Eq.~\eqref{eq:FindCut} for several empirical networks. We use the notation $q$ for {\sc Degree-Core}, $\mathcal{C}$ for {\sc Core-Score}, $\mathcal{P}$ for {\sc Path-Core}, $\mathcal{R}$ for {\sc LowRank-Core}, $\mathcal{L}$ for {\sc Lap-Core}, and $\mathcal{LS}$ for {\sc LapSgn-Core}.  We use the designation $^\ast$ for correlation values that have a p-value smaller than $0.01$ and the designation $^\dagger$ for z-scores whose absolute value is larger than $2$. We construe these results as statistically significant. We calculate the z-scores by randomly permuting the vertex indices (with 10000 different applications of such a permutation for each calculation) as described in \cite{mason13}: $z = (S_\textrm{frac} - \mu)/ (\textrm{std})$ where $\mu$ and ``$\textrm{std}$'', respectively, are the means and standard deviations of the $S_\textrm{frac}$ values for random permutations.
}
\label{tab:correlation_real_networks}
\end{minipage}
\end{table}

\end{document}